\documentclass[preprint,11pt]{imsart}
\RequirePackage[T1]{fontenc}
\usepackage{color}
\usepackage[latin1]{inputenc}
\usepackage[T1]{fontenc}
\usepackage{float}
\RequirePackage{amsfonts,amsthm,amsmath}
\RequirePackage[colorlinks,citecolor=blue,urlcolor=blue]{hyperref}
\usepackage{a4wide}
\newtheorem{thm}{Theorem}
\newtheorem{prop}[thm]{Proposition}
\newtheorem{cor}[thm]{Corollary}
\newtheorem{lemme}{Lemma}
\newtheorem{rmq}{Remark}
 \usepackage[latin1]{inputenc}
 \usepackage{amsmath}
 \usepackage{amsfonts}
 \usepackage{enumerate}
  \usepackage{multirow}
 \usepackage{graphicx}
\startlocaldefs
\theoremstyle{plain}
\endlocaldefs
\usepackage{sansmath}
\sansmath

\hypersetup{pdfstartview=FitH}
\ifpdf%
    \graphicspath{{PDFFIG/}}
\else
    \graphicspath{{EPSFIG/}}
\fi%

\begin{document}
\begin{frontmatter}
\title{Estimating FARIMA models with uncorrelated but non-independent error terms\\[0.5cm]
}
\runtitle{Estimation of weak FARIMA models}
%
\begin{aug}
\author{\fnms{\Large{Yacouba}}
\snm{\Large{Boubacar Ma\"{\i}nassara}}
\ead[label=e1]{yacouba.boubacar$\_$mainassara@univ-fcomte.fr}}\\[2mm]
\author{\fnms{\Large{Youssef}}
\snm{\Large{Esstafa}}
\ead[label=e2]{youssef.esstafa@univ-fcomte.fr}}\\[2mm]
\author{\fnms{\Large{Bruno}}
\snm{\Large{Saussereau}}
\ead[label=e3]{bruno.saussereau@univ-fcomte.fr}}
\runauthor{Y. Boubacar Ma\"{\i}nassara, Y. Esstafa and B. Saussereau}
\affiliation{Universit\'e Bourgogne Franche-Comt\'e}

\address{\hspace*{0cm}\\
Universit\'e Bourgogne Franche-Comt\'e, \\
Laboratoire de math\'{e}matiques de Besan\c{c}on, \\ UMR CNRS 6623, \\
16 route de Gray, \\ 25030 Besan\c{c}on, France.\\[0.2cm]
\printead{e1}\\
\printead{e2}\\
\printead{e3}}
\end{aug}
\vspace{0.5cm}
\begin{abstract} In this paper we derive the asymptotic properties of the least squares estimator (LSE) of fractionally integrated autoregressive moving-average (FARIMA) models under the assumption that the errors are uncorrelated but not necessarily independent nor martingale differences. We relax considerably the independence and even the martingale difference
assumptions on the innovation process to extend the range of application of the FARIMA models. We propose  a consistent estimator of the asymptotic covariance matrix of the LSE which may be very different from that obtained in the standard framework.  A self-normalized approach to  confidence interval construction for weak FARIMA model parameters is also presented. All our results are done under a mixing  assumption on the noise. Finally, some simulation studies  and an application to the daily returns  of stock market indices  are presented to corroborate our theoretical work. 
\end{abstract}
\begin{keyword}[class=AMS]
\kwd[Primary ]{62M10}
\kwd[; secondary ]{91B84}
\end{keyword}
\begin{keyword} Nonlinear processes; FARIMA models; Least-squares estimator; Consistency; Asymptotic normality; Spectral density estimation; Self-normalization; Cumulants
\end{keyword}
\end{frontmatter}
%
\section{Introduction}
Long memory processes takes a large part in the literature of time series (see for instance \cite{Granger&Joyeux}, \cite{Fox&Taqqu1986}, \cite{Dahlhaus1989}, \cite{Hosking}, \cite{Beran2013}, \cite{palma}, among others). They also play an important role in many scientific disciplines and applied fields such as hydrology, climatology, economics, finance, to name a few. 
To model the long memory phenomenon, a widely used model is the fractional autoregressive
integrated moving average  (FARIMA, for short) model. Consider a second order centered stationary process $X:=(X_t)_{t\in\mathbb{Z}}$ satisfying a  FARIMA$(p,d_0,q)$ representation of the form 
\begin{equation}\label{FARIMA}
a(L)(1-L)^{d_0}X_t=b(L)\epsilon_t,
\end{equation}
where $d_0\in\left]-1/2,1/2\right[$ is the long memory parameter, $L$ stands for the back-shift operator and
$a(L)=1-\sum_{i=1}^pa_iL^i$ is the autoregressive (AR for short) operator and  $b(L)=1-\sum_{i=1}^qb_iL^i$ is the moving average (MA for short) operator  (by convention $a_0=b_0=1$). 
The operators $a$ and $b$  represent the short memory part of the model. The linear innovation process $\epsilon:=(\epsilon_t)_{t\in \mathbb Z}$ is assumed to be a stationary sequence satisfies
\begin{itemize}
\item[\hspace*{1em} {\bf (A0):}]
\hspace*{1em} $\mathbb{E}\left[ \epsilon_t\right]=0, \ \mathrm{Var}\left(\epsilon_t\right)=\sigma_{\epsilon}^2 \text{ and } \mathrm{Cov}\left(\epsilon_t,\epsilon_{t+h}\right)=0$ for all $t\in\mathbb{Z}$ and all $h\neq 0$.
\end{itemize}
Under the above assumptions the process $\epsilon$ is called a weak white noise.
Different sub-classes of FARIMA models can be distinguished depending on the noise assumptions.
It is customary to say that  $X$ is a strong FARIMA$(p,d_0,q)$ representation and we will do this henceforth if in \eqref{FARIMA} $\epsilon$
is a  strong white noise, namely an independent and identically distributed (iid for short) sequence of random variables with mean 0 and common
variance. A strong white noise is obviously  a weak white noise because independence entails uncorrelatedness. Of course the
converse is not true. Between weak and strong noises, one can
say that $\epsilon$ is a semi-strong white noise if ${\epsilon}$ is a stationary martingale
difference, namely a sequence such that $\mathbb{E}(\epsilon_{t}|\epsilon_{t-1},\epsilon_{t-2},\dots)=0$.
 An example of semi-strong white noise is the generalized autoregressive conditional heteroscedastic (GARCH) model (see \cite{FZ2010}).
If $\epsilon$ is a semi-strong white noise in \eqref{FARIMA}, $X$ is called a semi-strong FARIMA$(p,d_0,q)$.
If no additional assumption is made on $\epsilon$, that is if $\epsilon$ is only a weak white noise (not
necessarily iid, nor a martingale difference), the representation \eqref{FARIMA} is called a weak FARIMA$(p,d_0,q)$.
It is clear from these definitions that the following inclusions hold:
$$\left\{\text{strong FARIMA}(p,d_0,q) \right\}\subset\left\{\text{semi-strong FARIMA}(p,d_0,q)
\right\}\subset\left\{\text{weak FARIMA}(p,d_0,q)\right\}.$$
Nonlinear models are becoming more and more employed because numerous real
time series exhibit nonlinear dynamics. For instance conditional
heteroscedasticity can not be generated by FARIMA models with iid noises.\footnote{ To cite few
examples of nonlinear processes, let us mention the self-exciting
threshold autoregressive (SETAR), the smooth transition
autoregressive (STAR), the exponential autoregressive (EXPAR), the
bilinear, the random coefficient autoregressive (RCA), the
functional autoregressive (FAR) (see \cite{Tong1990} and \cite{FY2008} for references on these nonlinear time series models).}
As mentioned by \cite{fz05,fz98} in the case of ARMA models, many important classes of nonlinear processes admit  weak ARMA
representations in which the linear innovation is not a martingale difference. 
The main issue with nonlinear models is that they are generally hard
to identify and implement. These technical difficulties certainly explain
the reason why the asymptotic theory of  FARIMA model estimation
is mainly limited to the strong or semi-strong FARIMA model. 

Now we present some of the main works about FARIMA model estimation when the noise is strong or semi-strong. For the estimation of long-range dependent processes, the commonly used
estimation method is based on the Whittle frequency domain maximum likelihood estimator (MLE)
(see for instance \cite{Dahlhaus1989}, \cite{Fox&Taqqu1986}, \cite{Taqqu1997}, \cite{GS1990}).
The asymptotic properties of the MLE of FARIMA models are well-known under the restrictive assumption
that the errors $\epsilon_t$ are independent or  martingale difference (see \cite{Beran1995}, \cite{Beran2013}, \cite{palma}, \cite{baillie1996}, \cite{ling-li}, \cite{hauser1998}, among others). \cite{hualde2011}, \cite{Nielsen2015} and \cite{CAVALIERE2017} have considered the problem of conditional sum-of squares estimation and inference on parametric fractional time series models driven by conditionally (unconditionally) heteroskedastic shocks.
All the works mentioned above assume either strong or semi-strong innovations.
In the modeling of financial time series, for example, the GARCH assumption on the errors is often used (see for instance \cite{baillie1996}, \cite{hauser1998})
to capture the  conditional  heteroscedasticity.
There is no doubt that it is  important to have a soundness inference procedure for the parameter in the
FARIMA model when the (possibly dependent) error is subject to unknown conditional  heteroscedasticity. Little is thus known when
the martingale difference assumption is relaxed. Our aim in this
paper is to consider a flexible FARIMA specification and to relax the
independence assumption (and even the martingale difference
assumption) in order to be able to cover weak FARIMA representations of general nonlinear models. This is why it is interesting to consider weak FARIMA models.

A  very few works deal with the asymptotic behavior of the MLE of weak FARIMA models. To our knowledge, \cite{shaox, shao2010} are the only papers on this subject. 
Under weak assumptions on the noise process, the author has obtained the asymptotic normality of the Whittle estimator (see \cite{Whittle1953}). Nevertheless, the inference
problem is not addressed. This is due to the fact that the asymptotic covariance matrix of the Whittle estimator
involves the integral of the fourth-order cumulant spectra of the dependent errors $\epsilon_t$. 
Using non-parametric bandwidth-dependent methods, one build an estimation of this integral but there is no guidance on the choice of the bandwidth in the
estimation procedures (see \cite{shaox,tani,kee,chiu} for further details).  
The difficulty is caused by the dependence in $\epsilon_t$. Indeed, for strong noise, 
a bandwidth-free consistent estimator of the asymptotic covariance matrix is available. When $\epsilon_t$ is dependent, no explicit formula for a consistent
estimator of the asymptotic variance matrix seems to be provided in the literature (see \cite{shaox}). 

In this work we propose to adopt for weak FARIMA models the estimation procedure developed in \cite{fz98} so we use the  least squares
estimator (LSE for short). 
We show that a strongly mixing property and the existence of moments are sufficient to obtain a consistent and asymptotically normally
distributed  least squares estimator for the parameters of a weak FARIMA representation.
For technical reasons, we often use an assumption on the summability of cumulants. This can be a consequence of a mixing and moments assumptions (see \cite{doukhan1989}, for more details). 
These kind of hypotheses enable us to  circumvent the problem of the lack of speed of convergence (due to  the long-range dependence) in the infinite AR or MA representations. 
We fix this gap by proposing  rather sharp estimations of the infinite AR and MA representations in the presence of long-range dependence (see Subsection \ref{prelim} for details). 

In our opinion there are three major contributions in this work. The first one is to show that the estimation procedure developed
in \cite{fz98} can be extended to weak FARIMA models. This goal is achieved thanks to Theorem \ref{convergence} and Theorem \ref{n.asymptotique} in which the consistency and the asymptotic normality are stated. 
The second one is to provide an answer to the open problem raised by  \cite{shaox} (see also \cite{shao2010})
on the asymptotic covariance matrix estimation. We propose  in our work a weakly consistent estimator of the asymptotic variance matrix (see Theorem \ref{convergence_Isp}).  
Thanks to this estimation of the asymptotic variance matrix, we can construct a confidence region for the estimation of the parameters. Finally another method to construct such confidence region is achieved thanks to  an alternative method using a self normalization procedure (see Theorem \ref{self-norm-estimator2}). 

The paper is organized as follows. Section \ref{estimation} shown  that the least squares
estimator for the parameters of a weak FARIMA model is consistent when the weak white noise $(\epsilon_t)_{t\in\mathbb{Z}}$ is
ergodic and stationary, and that the LSE is asymptotically normally distributed
when $(\epsilon_t)_{t\in\mathbb{Z}}$ satisfies mixing assumptions. The
asymptotic variance of the LSE may be very different in the weak
and strong cases. Section \ref{estimOmega} is devoted to the estimation of
this covariance matrix. We also propose a self-normalization-based approach to constructing a confidence region for the
parameters of  weak FARIMA models which avoids to estimate the asymptotic covariance
matrix. We gather in Section \ref{fig} all our figures and tables. These simulation studies and illustrative applications on real data are presented and discussed in Section \ref{num-ill}.
The proofs of the main results are collected in Section \ref{proofs}. 

\medskip
In all this work, we shall use the matrix norm defined  by $\|A\|=\sup_{\|x\|\leq 1}\|Ax\|=\rho^{1/2}( A^{'}A)$, when $A$ is a $\mathbb{R}^{k_1\times k_2}$ matrix, $\|x\|^2=x'x$ is the Euclidean norm of the vector $x\in\mathbb{R}^{k_2}$, and $\rho(\cdot)$ denotes the spectral radius.

\section{Least squares estimation}\label{estimation}
In this section we present the parametrization and the assumptions that are used in the sequel. 
Then we state the asymptotic properties of the LSE of weak FARIMA models. 
\subsection{Notations and assumptions}
We make the following standard assumption on the roots of the AR and MA polynomials in \eqref{FARIMA}.
\begin{itemize}
\item[\hspace*{1em} {\bf (A1):}]
The polynomials $a(z)$ and $b(z)$
have all their roots outside of the unit disk with no common factors.
\end{itemize}

Let $\Theta^{*}$ be the space
\begin{align*}
\Theta^{*}:=&  \Big
\{(\theta_1,\theta_2,...,\theta_{p+q})
\in{\Bbb R}^{p+q},\text{ where }  a_{\theta}(z)=1-\sum_{i=1}^{p}\theta_{i}z^{i}\, \text{
and } b_{\theta}(z)=1-\sum_{j=1}^{q}\theta_{p+j}z^{j} \\
& \hspace{6.5cm} \text{ have all
their zeros outside the unit disk}\Big \}\ .
\end{align*}
Denote by $\Theta$ the cartesian product $\Theta^{*}\times\left[ d_1,d_2\right]$, where $\left[ d_1,d_2\right]\subset\left] -1/2,1/2\right[$ with $d_1-d_0>-1/2$.
The unknown parameter of interest $\theta_0=(a_1,a_2,\dots,a_p,b_1,b_2,\dots,b_q,d_0)$ is supposed to belong to the parameter space $\Theta$.

The fractional difference operator $(1-L)^{d_0}$ is defined, using the generalized binomial series, by
\begin{equation*}
(1-L)^{d_0}=\sum_{j\geq 0}\alpha_j(d_0)L^j,
\end{equation*}
where for all $j\ge 0$, $\alpha_j(d_0)=\Gamma(j-d_0)/\left\lbrace \Gamma(j+1)\Gamma(-d_0)\right\rbrace $ and $\Gamma(\cdot)$ is the Gamma function. 
Using the Stirling formula we obtain that for large $j$,  $\alpha_j(d_0)\sim j^{-d_0-1}/\Gamma(-d_0)$ (one refers to \cite{Beran2013} for further details). 

For all $\theta\in\Theta$ we define $( \epsilon_t(\theta)) _{t\in\mathbb{Z}}$ as the second order stationary process which is the solution of
\begin{equation}\label{FARIMA-th}
\epsilon_t(\theta)=\sum_{j\geq 0}\alpha_j(d)X_{t-j}-\sum_{i=1}^p\theta_i\sum_{j\geq 0}\alpha_j(d)X_{t-i-j}+\sum_{j=1}^q\theta_{p+j}\epsilon_{t-j}(\theta).
\end{equation}
Observe that, for all $t\in\mathbb{Z}$, $\epsilon_t(\theta_0)=\epsilon_t$ a.s.  Given a realization  $X_1,\dots,X_n$ of length $n$, $\epsilon_t(\theta)$ can be approximated, for $0<t\leq n$, by $\tilde{\epsilon}_t(\theta)$ defined recursively by
\begin{equation}\label{exp-epsil-tilde}
\tilde{\epsilon}_t(\theta)=\sum_{j=0}^{t-1}\alpha_j(d)X_{t-j}-\sum_{i=1}^p\theta_i\sum_{j=0}^{t-i-1}\alpha_j(d)X_{t-i-j}+\sum_{j=1}^q\theta_{p+j}\tilde{\epsilon}_{t-j}(\theta),
\end{equation}
\noindent with $\tilde{\epsilon}_t(\theta)=X_t=0$ if $t\leq 0$.
It will be shown that these initial values are asymptotically
negligible and, in particular, that
$\epsilon_t(\theta)-\tilde{\epsilon}_t(\theta)\to 0$ in $\mathbb{L}^2$ as $t\to\infty$ (see Remark~\ref{rmq:important} hereafter).
Thus the choice of the initial values has no influence on the asymptotic properties of the model parameters estimator.

Let $\Theta^{*}_{\delta}$ denote the compact set
$$\Theta^{*}_{\delta}=\left\lbrace \theta\in\mathbb{R}^{p+q}; \text{ the roots of the polynomials } a_{\theta}(z) \text{ and } b_{\theta}(z) \text{ have modulus } \geq 1+\delta\right\rbrace.$$
We define the set $\Theta_{\delta}$ as the cartesian product of $\Theta^{*}_{\delta}$ by $\left[ d_1,d_2\right] $, i.e. $\Theta_{\delta}=\Theta^{*}_{\delta}\times\left[ d_1,d_2\right]$, where $\delta$ is a positive constant chosen such that $\theta_0$ belongs to $\Theta_{\delta}$.

The random variable $\hat{\theta}_n$ is called least squares estimator if it satisfies, almost surely,
\begin{equation}\label{theta_chap}
\hat{\theta}_n=\underset{\theta\in\Theta_{\delta}}{\mathrm{argmin}} \ Q_n(\theta),
\text{ where }Q_n(\theta)=\frac{1}{n}\sum_{t=1}^n\tilde{\epsilon}_t^2(\theta).
\end{equation}
Our main results are proven under the following assumptions:
\begin{itemize}
\item[\hspace*{1em} {\bf (A2):}]
\hspace*{1em} The process $(\epsilon_t)_{t\in\mathbb{Z}}$ is strictly stationary and ergodic.
\end{itemize}
The consistency of the least squares estimator will be proved under the three above assumptions (\textbf{(A0)}, \textbf{(A1)} and \textbf{(A2)}).
For the asymptotic normality of the LSE, additional assumptions are required.
It is necessary to assume
that $\theta_0$ is not on the boundary of the parameter space
${\Theta_\delta}$.
\begin{itemize}
\item[\hspace*{1em} {\bf (A3):}]
\hspace*{1em} We have $\theta_0\in\overset{\circ}{\Theta_\delta}$, where $\overset{\circ}{\Theta_\delta}$ denotes the interior of $\Theta_\delta$.
\end{itemize}
The stationary process $\epsilon$ is not supposed to be an independent sequence. So one needs to control its dependency by means of its  strong mixing coefficients $\left\lbrace \alpha_{\epsilon}(h)\right\rbrace _{h\geq 0}$ defined by
$$\alpha_{\epsilon}\left(h\right)=\sup_{A\in\mathcal F_{-\infty}^t,B\in\mathcal F_{t+h}^{+\infty}}\left|\mathbb{P}\left(A\cap
B\right)-\mathbb{P}(A)\mathbb{P}(B)\right|,$$
where $\mathcal F_{-\infty}^t=\sigma (\epsilon_u, u\leq t )$ and $\mathcal F_{t+h}^{+\infty}=\sigma (\epsilon_u, u\geq t+h)$.

We shall need  an integrability assumption on the moments of the noise $\epsilon$ and a summability condition on the strong mixing coefficients $(\alpha_{\epsilon}(h))_{h\ge 0}$. 
\begin{itemize}
\item[\hspace*{1em} {\bf (A4):}]
\hspace*{1em}  There exists an integer $\tau$ such that for some $\nu\in]0,1]$, we have $\mathbb{E}|\epsilon_t|^{\tau+\nu}<\infty$ and $\sum_{h=0}^{\infty}(h+1)^{k-2} \left\lbrace \alpha_{\epsilon}(h)\right\rbrace ^{\frac{\nu}{k+\nu}}<\infty$ for $k=1,\dots,\tau$. 
\end{itemize}

Note that {\bf (A4)} implies the following weak assumption on the joint cumulants of the innovation process $\epsilon$ (see \cite{doukhan1989}, for more details).
\begin{itemize}
\item[\hspace*{1em} {\bf (A4'):}]
\hspace*{1em} There exists an integer $\tau\ge 2$ such that $C_\tau:=\sum_{i_1,\dots,i_{\tau-1}\in\mathbb{Z}}|\mathrm{cum}(\epsilon_0,\epsilon_{i_1},\dots,\epsilon_{i_{\tau-1}})|<\infty \ .$
\end{itemize}
In the above expression, $\mathrm{cum}(\epsilon_0,\epsilon_{i_1},\dots,\epsilon_{i_{\tau-1}})$ denotes the $\tau-$th order cumulant of the stationary process. Due to the fact that the $\epsilon_t$'s are centered,  we notice that for fixed $(i,j,k)$
$$\mathrm{cum}(\epsilon_0,\epsilon_i,\epsilon_j,\epsilon_k)=\mathbb{E}\left[\epsilon_0\epsilon_i\epsilon_j\epsilon_k\right]-\mathbb{E}\left[\epsilon_0\epsilon_i\right]\mathbb{E}\left[\epsilon_j\epsilon_k\right]-\mathbb{E}\left[\epsilon_0\epsilon_j\right]\mathbb{E}\left[\epsilon_i\epsilon_k\right]-\mathbb{E}\left[\epsilon_0\epsilon_k\right]\mathbb{E}\left[\epsilon_i\epsilon_j\right].$$
Assumption \textbf{(A4)} is a usual technical hypothesis which is useful when one proves the asymptotic normality (see \cite{fz98} for example). Let us notice however that we impose a stronger convergence speed for the mixing coefficients than in the works on weak ARMA processes. This is due to the fact that the coefficients in the AR or MA representation of $\epsilon_t(\theta)$ have no more exponential decay because of the fractional operator (see Subsection \ref{prelim} for details and comments). 

As mentioned before, Hypothesis  \textbf{(A4)} implies  \textbf{(A4')} which is also a technical assumption usually used in the fractionally integrated ARMA processes framework (see for instance \cite{s2010JRSSBa}) or even in an ARMA context (see \cite{FZ2007,zl}). One remarks that in \cite{shao2010}, the author emphasized that a geometric moment contraction implies \textbf{(A4')}. This provides an alternative to strong mixing assumptions but, to our knowledge, there is no relation between this two kinds of hypotheses. 
\subsection{Asymptotic properties}

The asymptotic properties of the LSE  of the weak FARIMA model are stated in the following two theorems. 
\begin{thm}{(Consistency).}\label{convergence} Assume that $(\epsilon_t)_{t\in\mathbb{Z}}$ satisfies $\eqref{FARIMA}$ and belonging to $\mathbb{L}^2$. Let $( \hat{\theta}_n)_{n\ge 1}$ be a sequence of least squares estimators. Under Assumptions \textbf{(A0)}, \textbf{(A1)} and \textbf{(A2)}, we have
\begin{equation*}
\hat{\theta}_n\xrightarrow[n\to \infty]{\mathbb{P}} \, \theta_0.
\end{equation*}
\end{thm}
The proof of this theorem is given in Subsection~\ref{sectconv}.

In order to state our asymptotic normality result, we define the function
\begin{align}\label{On}
O_n(\theta)& =\frac{1}{n}\sum_{t=1}^n\epsilon_t^2(\theta),
\end{align}
where the sequence $\left( \epsilon_t(\theta)\right)_{t\in\mathbb{Z}}$ is given by \eqref{FARIMA-th}.
We consider the following information matrices
$$I(\theta)=\lim_{n\rightarrow\infty}Var\left\lbrace \sqrt{n}\frac{\partial}{\partial\theta}O_n(\theta)\right\rbrace \text{  and  } J(\theta)=\lim_{n\rightarrow\infty}\left[ \frac{\partial^2}{\partial\theta_i\partial\theta_j}O_n(\theta)\right] \text{a.s.}$$
The existence of these matrices are proved when one demonstrates the following result. 
\begin{thm}{(Asymptotic normality).}\label{n.asymptotique} We assume that $(\epsilon_t)_{t\in\mathbb{Z}}$ satisfies \eqref{FARIMA}. 
Under \textbf{(A0)}-\textbf{(A3)} and Assumption {\bf (A4)} with  $\tau=4$, the sequence $( \sqrt{n}( \hat{\theta}_n-\theta_0)) _{n\ge 1}$  has a limiting centered normal distribution with covariance matrix $\Omega:=J^{-1}(\theta_0)I(\theta_0)J^{-1}(\theta_0)$.
\end{thm}
The proof of this theorem is given in Subsection~\ref{sectna}.
\begin{rmq}Hereafter (see more precisely  \eqref{matrixJexplicit}), we will be able to prove that $$J(\theta_0)= 2\mathbb{E}\left[ \frac{\partial}{\partial\theta}\epsilon_t(\theta_0)\frac{\partial}{\partial\theta'}\epsilon_t(\theta_0)\right] \ \text{a.s.}$$
Thus the matrix $J(\theta_0)$ has the same expression in the strong and weak FARIMA cases (see Theorem 1 of \cite{Beran1995}).
On the contrary, the matrix $I(\theta_0)$ is in general much more complicated in the
weak case than in the strong case.
\end{rmq}
\begin{rmq}
\label{IegalJ}  In the standard strong FARIMA case, {\em i.e.}
when {\bf (A2)} is replaced by the assumption that $(\epsilon_t)_{t\in\mathbb{Z}}$ is
iid, we have $I(\theta_0)=2\sigma_{\epsilon}^2J(\theta_0)$. Thus the asymptotic
covariance matrix is then reduced as $\Omega_S:=2\sigma_{\epsilon}^2J^{-1}(\theta_0)$. Generally, when
the noise is not an independent sequence, this simplification can not be made and
we have $I(\theta_0)\ne 2\sigma_{\epsilon}^2J(\theta_0)$.
The true asymptotic covariance matrix $\Omega=J^{-1}(\theta_0)I(\theta_0)J^{-1}(\theta_0)$ obtained
in the weak FARIMA framework can be very different from $\Omega_S$.
As a consequence, for the statistical inference on the
parameter,  the ready-made softwares used to
fit FARIMA do not provide a correct estimation of $\Omega$ for weak
FARIMA processes because the standard
time series analysis softwares use empirical estimators of $\Omega_S$.
The problem also holds in the weak ARMA case (see
\cite{FZ2007} and the references therein).This is why it is interesting to find  an estimator of $\Omega$  which is consistent
for both weak and (semi-)strong FARIMA cases.
\end{rmq}
Based on the above remark, the next section deals with two different methods  in order to find  an estimator of $\Omega$. 
\section{Estimating the asymptotic variance matrix}\label{estimOmega}
For statistical inference problem, the asymptotic variance $\Omega$ has to be estimated. In particular
Theorem \ref{n.asymptotique} can be used to obtain confidence intervals and significance tests for the parameters.

First of all, the matrix $J(\theta_0)$ can be  estimated empirically by the square matrix $\hat{J}_n$ of order $p+q+1$ defined by:
\begin{equation}\label{estim_J}
\hat{J}_n=\frac{2}{n}\sum_{t=1}^n\left\lbrace \frac{\partial}{\partial\theta}\tilde{\epsilon}_t\left(\hat{\theta}_n\right)\right\rbrace \left\lbrace \frac{\partial}{\partial\theta^{'}}\tilde{\epsilon}_t\left(\hat{\theta}_n\right)\right\rbrace .
\end{equation}
The convergence of $\hat{J}_n$ to $J(\theta_0)$ is  classical (see Lemma \ref{Conv_almost_sure_J_n_vers_J} in Subsection \ref{sectna} for details).

In the standard strong FARIMA case, in view of remark~\ref{IegalJ}, we have  $\hat{\Omega}_S:=2\hat{\sigma}_{\epsilon}^2\hat{J}_n^{-1}$ with $\hat{\sigma}_{\epsilon}^2=Q_n(\hat{\theta}_n)$. Thus $\hat{\Omega}_S$  is a consistent estimator of $\Omega_S$.
In the general weak FARIMA case, this estimator is not consistent when $I(\theta_0)\neq 2\sigma_{\epsilon}^2J(\theta_0)$. So we need a consistent estimator of $I(\theta_0)$.
\subsection{Estimation of the asymptotic matrix $I(\theta_0)$}\label{ii}
For all $t\in\mathbb{Z}$, let
\begin{align}\label{Ht_process}
H_t(\theta_0)&=2\epsilon_t(\theta_0)\frac{\partial}{\partial\theta}\epsilon_t(\theta_0)
=\left(2\epsilon_t(\theta_0)\frac{\partial}{\partial\theta_1}\epsilon_t(\theta_0),\dots,2\epsilon_t(\theta_0)\frac{\partial}{\partial\theta_{p+q+1}}\epsilon_t(\theta_0)\right) ^{'}.
\end{align}
We shall see in the proof of Lemma~\ref{exit_I} that 
\begin{align*}
I(\theta_0)&=\lim_{n\rightarrow\infty}\mathrm{Var}\left( \frac{1}{\sqrt{n}}\sum_{t=1}^nH_t(\theta_0)\right)= \sum_{h=-\infty}^{+\infty}\mathrm{Cov}\left( H_t(\theta_0),H_{t-h}(\theta_0)\right) .
\end{align*}
Following the arguments developed in \cite{BMCF2012}, the matrix $I(\theta_0)$ can be estimated using Berk's approach (see \cite{berk}). More precisely, by interpreting $I(\theta_0)/2\pi$ as the spectral density of the stationary process $(H_t(\theta_0))_{t\in\mathbb{Z}}$ evaluated at frequency $0$, we can use a parametric autoregressive estimate of the spectral density of $(H_t(\theta_0))_{t\in\mathbb{Z}}$ in order to estimate
the matrix $I(\theta_0)$.

For any $\theta\in\Theta$, $H_t(\theta)$ is a measurable function of $\left\lbrace \epsilon_s,s\leq t\right\rbrace $.
The stationary process $(H_t(\theta_0))_{t\in\mathbb{Z}}$ admits the following Wold decomposition $H_t(\theta_0)=u_t+\sum_{k=1}^{\infty}\psi_ku_{t-k}$, where $(u_t)_{t\in\mathbb{Z}}$ is a $(p+q+1)-$variate weak white noise with variance matrix $\Sigma_u$.

Assume that $\Sigma_u$ is non-singular, that $\sum_{k=1}^{\infty}\left\|\psi_k\right\|<\infty$, and that $\det(I_{p+q+1}+\sum_{k=1}^{\infty}\psi_kz^k)\neq 0$ if $\left|z\right|\leq 1$. Then $(H_t(\theta_0))_{t\in\mathbb{Z}}$ admits a weak multivariate $\mathrm{AR}(\infty)$ representation (see \cite{Akutowicz1957}) of the form
\begin{equation}\label{AR_infty}
\Phi(L)H_t(\theta_0):=H_t(\theta_0)-\sum_{k=1}^{\infty}\Phi_kH_{t-k}(\theta_0)=u_t,
\end{equation}
such that $\sum_{k=1}^{\infty}\left\|\Phi_k\right\|<\infty$ and $\det\left\lbrace \Phi(z)\right\rbrace \neq 0$ if $\left|z\right|\leq 1$. %

Thanks to the previous remarks, the estimation of $I(\theta_0)$ is therefore based on the following expression
$$I(\theta_0)=\Phi^{-1}(1)\Sigma_u\Phi^{-1}(1).$$
Consider the regression of $H_t(\theta_0)$ on $H_{t-1}(\theta_0),\dots,H_{t-r}(\theta_0)$ defined by
\begin{equation}\label{AR_tronquee}
H_t(\theta_0)=\sum_{k=1}^{r}\Phi_{r,k}H_{t-k}(\theta_0)+u_{r,t},
\end{equation}
where $u_{r,t}$ is uncorrelated with $H_{t-1}(\theta_0),\dots,H_{t-r}(\theta_0)$.
Since ${H}_t(\theta_0)$ is not observable, we introduce $\hat{H}_t\in\mathbb{R}^{p+q+1}$ obtained by replacing $\epsilon_t(\cdot)$ by $\tilde{\epsilon}_t(\cdot)$ and $\theta_0$ by $\hat{\theta}_n$ in $\eqref{Ht_process}$: 
\begin{align}\label{Ht_chapeau}
\hat H_t&=2\tilde\epsilon_t(\hat\theta_n)\frac{\partial}{\partial\theta}\hat\epsilon_t(\theta_n) \ .
\end{align}
Let $\hat{\Phi}_r(z)=I_{p+q+1}-\sum_{k=1}^r\hat\Phi_{r,k}z^k$, where $\hat\Phi_{r,1},\dots,\hat\Phi_{r,r}$ denote the coefficients of the LS regression
of $\hat{H}_t$ on $\hat{H}_{t-1},\dots,\hat{H}_{t-r}$. Let $\hat{u}_{r,t}$ be the residuals of this regression and  let $\hat{\Sigma}_{\hat{u}_r}$ be the empirical variance (defined in \eqref{coef-reg} below) of $\hat{u}_{r,1},\dots,\hat{u}_{r,r}$.
The LSE of $\underline{\Phi}_r=\left( \Phi_{r,1},\dots,\Phi_{r,r}\right)$ and  $\Sigma_{u_r}=\mathrm{Var}(u_{r,t})$ are given by
\begin{equation}\label{coef-reg}
\underline{\hat{\Phi}}_r=\hat{\Sigma}_{\hat{H},\underline{\hat{H}}_r}\hat{\Sigma}_{\underline{\hat{H}}_r}^{-1} \  \text{ and } \ \hat{\Sigma}_{\hat{u}_r}=\frac{1}{n}\sum_{t=1}^n\left(\hat{H}_t-\underline{\hat{\Phi}}_r \underline{\hat{H}}_{r,t}\right) \left( \hat{H}_t-\underline{\hat{\Phi}}_r \underline{\hat{H}}_{r,t}\right)^{'},
\end{equation}
where
$$\underline{\hat{H}}_{r,t}=( \hat{H}_{t-1}^{'},\dots,\hat{H}_{t-r}^{'}) ^{'},\quad\hat{\Sigma}_{\hat{H},\underline{\hat{H}}_r}=\frac{1}{n}\sum_{t=1}^n\hat{H}_t\underline{\hat{H}}_{r,t}^{'}
\text{ and } \hat{\Sigma}_{\underline{\hat{H}}_r}=\frac{1}{n}\sum_{t=1}^n\underline{\hat{H}}_{r,t}\underline{\hat{H}}_{r,t}^{'},$$
with by convention $\hat{H}_t=0$ when $t\leq 0$. We assume that  $\hat{\Sigma}_{\underline{\hat{H}}_r}$ is non-singular (which holds true asymptotically).
%

In the case of linear processes with independent innovations, Berk (see \cite{berk}) has shown that the spectral density can be consistently estimated by fitting autoregressive models of order $r=r(n)$, whenever $r$ tends to infinity and $r^3/n$ tends to $0$ as $n$ tends to infinity.  There are differences with \cite{berk}: $(H_t(\theta_0))_{t\in\mathbb{Z}}$ is multivariate, is not directly observed and is replaced by $(\hat{H}_t)_{t\in\mathbb{Z}}$. It is shown that this result remains valid for the multivariate linear process $(H_t(\theta_0))_{t\in\mathbb{Z}}$ with non-independent innovations (see \cite{BMCF2012,yac2}, for references in weak (multivariate) ARMA models). We will extend the results of \cite{BMCF2012} to weak FARIMA models.

The asymptotic study of the estimator of $I(\theta_0)$ using the spectral density method is given in the following theorem. 
\begin{thm}\label{convergence_Isp} 
We assume \textbf{(A0)}-\textbf{(A3)} and Assumption {\bf (A4')} with  $\tau=8$. In addition, we assume that the innovation process $(\epsilon_t)_{t\in\mathbb{Z}}$ of the FARIMA$(p,d_0,q)$ model \eqref{FARIMA} is such that the process $(H_t(\theta_0))_{t\in\mathbb{Z}}$ defined in \eqref{Ht_process} admits a multivariate AR$(\infty)$ representation \eqref{AR_infty}, where $\|\Phi_k\|=\mathrm{o}(k^{-2})$ as $k\to\infty$, the roots of $\det(\Phi(z))=0$ are outside the unit disk, and $\Sigma_u=\mathrm{Var}(u_t)$ is non-singular. 
Then, the spectral estimator of $I(\theta_0)$
$$\hat{I}^{\mathrm{SP}}_n:=\hat{\Phi}_r^{-1}(1)\hat{\Sigma}_{\hat{u}_r}\hat{\Phi}_r^{'-1}(1)\xrightarrow[]{} I(\theta_0)=\Phi^{-1}(1)\Sigma_u\Phi^{-1}(1)$$
in probability when $r=r(n)\to\infty$ and $r^5(n)/n^{1-2(d_0-d_1)}\to0$ as $n\to\infty$ (remind that  $d_0\in [ d_1{,}d_2]\subset] -1/2{,}1/2[$).
\end{thm}
The proof of this theorem is given in Subsection~\ref{sectI}.

A second method to estimate the asymptotic matrix (or rather avoiding estimate it) is proposed in the next subsection. 
\subsection{A self-normalized approach to confidence interval construction in weak FARIMA models}\label{sn}
We have seen previously that we may obtain confidence intervals for weak FARIMA model parameters as soon as we can construct a convergent estimator of the variance matrix $I(\theta_0)$ (see Theorems \ref{n.asymptotique} and \ref{convergence_Isp}). 
The parametric approach based on an autoregressive estimate of the spectral density of $(H_t(\theta_0))_{t\in\mathbb{Z}}$ that we used before has the drawback of choosing the truncation parameter $r$ in \eqref{AR_tronquee}. 
This choice of the order truncation is often crucial and difficult. So the aim of this section is to avoid such a difficulty. 

This section is also of interest because, to our knowledge, it has not been studied for weak FARIMA models.
Notable exception is \cite{shaox} who studied this problem in a short memory case (see Assumption 1 in  \cite{shaox} that implies that the process $X$ is short-range dependent).

We propose an alternative method to obtain confidence intervals for weak FARIMA models by avoiding the estimation of the asymptotic covariance matrix $I(\theta_0)$. It is based on a self-normalization approach  used to build a statistic which depends on the true parameter $\theta_0$ and which is asymptotically distribution-free (see Theorem 1 of
\cite{shaox} for a reference in weak ARMA case). 
The idea comes from \cite{lobato} and has been already extended by \cite{BMS2018,kl2006,s2010JRSSBa,s2010JRSSBb,shaox} to more general frameworks.
See also \cite{s2016} for a review on some recent developments on the inference of time series data using the self-normalized approach.

Let us briefly explain the idea of the self-normalization.

By a Taylor expansion of the function $\partial Q_n(\cdot)/ \partial\theta $ around $\theta_0$, under {\bf (A3)}, we have
\begin{equation}\label{OOO}
0=\sqrt{n}\frac{\partial}{\partial\theta}Q_n(\hat{\theta}_n)=\sqrt{n}\frac{\partial}{\partial\theta}Q_n(\theta_0)+\left[\frac{\partial^2}{\partial\theta_i\partial\theta_j}Q_n\left( \theta^*_{n,i,j}\right) \right]\sqrt{n}\left( \hat{\theta}_n-\theta_0\right) ,
\end{equation}
where the $\theta^*_{n,i,j}$'s are between $\hat{\theta}_n$ and $\theta_0$.
Using the following
equation
\begin{align*}
\sqrt{n}\left(\frac{\partial}{\partial\theta}O_n(\theta_0)-\frac{\partial}{\partial\theta}Q_n(\theta_0)\right)=\sqrt{n}\frac{\partial}{\partial\theta}O_n(\theta_0)+\left\lbrace\left[ \frac{\partial^2}{\partial\theta_i\partial\theta_j}Q_n(\theta^{*}_{n,i,j})\right]-J(\theta_0) +J(\theta_0)\right\rbrace\sqrt{n}(\hat{\theta}_n-\theta_0),
\end{align*}we shall be able to prove that \eqref{OOO} implies that
\begin{align}\label{OO}
\sqrt{n}\frac{\partial}{\partial\theta}O_n(\theta_0)+J(\theta_0)\sqrt{n}(\hat{\theta}_n-\theta_0)=\mathrm{o}_{\mathbb{P}}\left(1\right).
\end{align}
This is due to the following technical properties:
\begin{itemize}
\item the convergence in probability of $\sqrt{n}\partial Q_n(\theta_0)/\partial\theta-\sqrt{n}\partial O_n(\theta_0)/\partial\theta$ to 0 (see Lemma \ref{ConProQversO} hereafter),
\item the convergence in probability of $[\partial^2 Q_n(\theta_{n,i,j}^{*})/\partial\theta_i\partial\theta_j]$ to $J(\theta_0)$ (see Lemma \ref{Conv_almost_sure_J_n_vers_J} hereafter),
\item the tightness of the sequence $(\sqrt{n}(\hat{\theta}_n-\theta_0))_{n\geq 1}$ (see Theorem \ref{n.asymptotique}) and
\item  the existence and invertibility of the matrix $J(\theta_0)$ (see Lemma \ref{matrixJ} hereafter).
\end{itemize}
%
Thus we obtain from \eqref{OO} that
\begin{align*}
\sqrt{n}(\hat{\theta}_n-\theta_0)=\frac{1}{\sqrt{n}}\sum_{t=1}^nU_t +\mathrm{o}_{\mathbb{P}}\left(1\right),
\end{align*}
where (remind \eqref{Ht_process})
\begin{align*}
U_t=-J^{-1}(\theta_0)H_t(\theta_0).
\end{align*}
At this stage, we do not rely on the classical method that would consist in estimating
the asymptotic covariance matrix $I(\theta_0)$. We rather try to apply Lemma 1 in \cite{lobato}. So
we need to check that a functional central limit theorem holds for the process $U:=(U_t)_{t\geq 1}$.
For that sake, we define the normalization matrix $P_{p+q+1,n}$ of $\mathbb{R}^{(p+q+1)\times (p+q+1)}$ by
\begin{equation}\label{matP}
P_{p+q+1,n}=\frac{1}{n^2}\sum_{t=1}^n\left(\sum_{j=1}^t( U_j-\bar U_n)\right)\left(\sum_{j=1}^t (U_j-\bar U_n)\right)^{'} , 
\end{equation}
where $\bar U_n = (1/n)\sum_{i=1}^n U_i$. To ensure the invertibility of the normalization matrix $P_{p+q+1,n}$ (it is the result stated in the next proposition), we need the following technical assumption on the distribution of $\epsilon_t$.
\begin{itemize}
\item[\hspace*{1em} {\bf (A5):}]
\hspace*{1em} The process $(\epsilon_t)_{t\in\mathbb{Z}}$ has a positive density on some neighborhood of zero.
\end{itemize}
\begin{prop}\label{inversibleP}
Under the assumptions of Theorem \ref{n.asymptotique} and {\bf (A5)}, the matrix $P_{p+q+1,n}$ is almost surely non singular.
 \end{prop}
The proof of this proposition is given in Subsection~\ref{sectinvP}.

Let $(B_m(r))_{r\geq 0}$ be a $m$-dimensional Brownian motion starting from 0. For $m\geq 1$, we denote by $\mathcal{U}_m$ the random variable defined by:
\begin{equation}\label{Uk}
\mathcal{U}_m=B_m^{'}(1)V_m^{-1}B_m(1),
\end{equation}
where
\begin{equation}\label{Vk}
V_m=\int_0^1\left(B_m(r)-rB_m(1)\right)\left(B_m(r)-rB_m(1)\right)^{'}dr.
\end{equation}
The critical values of $\mathcal{U}_m$ have been tabulated by \cite{lobato}.

The following theorem states the self-normalized asymptotic distribution of the random vector $\sqrt{n}(\hat{\theta}_n-\theta_0)$.
\begin{thm}\label{self-norm-estimator} Under the assumptions of Theorem $\ref{n.asymptotique}$ and \textbf{(A5)}, we have
$$n (\hat{\theta}_n-\theta_0)^{'}P_{p+q+1,n}^{-1}(\hat{\theta}_n-\theta_0)\xrightarrow[n\to\infty]{\text{in law}}\mathcal{U}_{p+q+1}.$$
\end{thm}
The proof of this theorem is given in Subsection~\ref{sectSN}.

Of course, the above theorem is useless for practical purpose  because the normalization matrix $P_{p+q+1,n}$ is not observable. This gap will be fixed below when one replaces the matrix $P_{p+q+1,n}$ by  its empirical or observable counterpart
\begin{equation}\label{matPchap}
\hat{P}_{p+q+1,n}=\frac{1}{n^2}\sum_{t=1}^n\left(\sum_{j=1}^t( \hat U_j- \mbox{$\frac{1}n \sum_{k=1}^n \hat U_k$})\right)\left(\sum_{j=1}^t( \hat U_j- \mbox{$\frac{1}n \sum_{k=1}^n \hat U_k$})\right)^{'}\text{ where }\hat{U}_j=-\hat{J}_n^{-1}\hat{H}_j.
\end{equation}
The above quantity is observable and we are able to state our Theorem which is the applicable version of Theorem \ref{self-norm-estimator}.
\begin{thm}\label{self-norm-estimator2} Under the assumptions of Theorem $\ref{n.asymptotique}$ and \textbf{(A5)}, we have
$$n (\hat{\theta}_n-\theta_0)^{'}\hat P_{p+q+1,n}^{-1}(\hat{\theta}_n-\theta_0)\xrightarrow[n\to\infty]{\text{in law}}\mathcal{U}_{p+q+1}.$$
\end{thm}
The proof of this theorem is given in Subsection~\ref{sectSN2}.

At the asymptotic level $\alpha$, a joint $100(1-\alpha)\%$ confidence region for the elements of $\theta_0$ is then given by the set of values of the vector $\theta$ which satisfy the following inequality:

$$n(\hat{\theta}_n-\theta)^{'}\hat{P}_{p+q+1,n}^{-1}(\hat{\theta}_n-\theta)\leq \mathcal{U}_{p+q+1,\alpha},$$
where $\mathcal{U}_{p+q+1,\alpha}$ is the quantile of order $1-\alpha$ for the distribution of $\mathcal{U}_{p+q+1}$.
\begin{cor} For any $1\le i\le p+q+1$, a $100(1-\alpha)\%$ confidence region for $\theta_0(i)$ is given by the following set:
$$\left\lbrace x\in\mathbb{R}\ {;} \ n\big (\hat{\theta}_n(i)-x\big )^{2}\hat{P}_{p+q+1,n}^{-1}(i,i)\leq \mathcal{U}_{1,\alpha}\right\rbrace ,$$
\end{cor}
where $\mathcal{U}_{1,\alpha}$ denotes the quantile of order $1-\alpha$ of the distribution for $\mathcal{U}_{1}$.

The proof of this corollary is similar to that of Theorem \ref{self-norm-estimator2} when one restricts ourselves to a one dimensional case.
\section{Numerical illustrations}\label{num-ill}
In this section, we investigate the finite sample
properties of the asymptotic results that we introduced in this work. For that sake we use Monte Carlo experiments. 
The numerical illustrations of this section are made  with the open source
statistical software R (see R Development Core Team, 2017) or (see http://cran.r-project.org/).
\subsection{Simulation studies and empirical sizes for confidence intervals}
We study numerically the behavior of the LSE for FARIMA models of the form
\begin{equation}\label{process-sim}
(1-L)^d\left(X_t-aX_{t-1}\right)=\epsilon_t-b\epsilon_{t-1},
\end{equation}
where the  unknown parameter is taken as $\theta_0=(a,b,d)=(-0.7,-0.2,0.4)$.  First we assume that in \eqref{process-sim} the innovation process $(\epsilon_t)_{t\in\mathbb{Z}}$ is an iid centered Gaussian process with common variance 1
which corresponds to the strong FARIMA case.
In two other experiments we consider that in \eqref{process-sim} the innovation processes $(\epsilon_t)_{t\in\mathbb{Z}}$ are defined respectively by
\begin{equation} \label{noise-sim}
\left\{\begin{array}{l}\epsilon_{t}=\sigma_t\eta_{t}\\
\sigma_t^2=0.04+0.12\epsilon_{t-1}^2 +0.85\sigma_{t-1}^2
\end{array}\right.
\end{equation}
and
\begin{align}
\label{PT}
\epsilon_{t}& =\eta_{t}^2\eta_{t-1},
\end{align}
where  $(\eta_t)_{t\ge 1}$ is a sequence of iid centered Gaussian random variables with variance 1.
Note that the innovation process in \eqref{PT} is not a martingale difference whereas it is the case of the noise defined in \eqref{noise-sim}.

%
We simulated $N=1,000$ independent trajectories of size $n=2,000$ of Model \eqref{process-sim} in the three following case: the strong Gaussian noise,
 the semi-strong noise \eqref{noise-sim} and the weak noise \eqref{PT}.

Figure~\ref{fig1}, Figure~\ref{fig2} and Figure~\ref{fig3} compare the distribution of the LSE in these three contexts. The distributions of $\hat{d}_n$ are similar in the three cases whereas the LSE $\hat{a}_n$  of $a$ is more accurate in the weak case than in the strong and semi-strong cases.
The distributions of  $\hat{b}_n$ are more accurate in the strong case than in the weak case. Remark that in the weak case the distributions of  $\hat{b}_n$ are more accurate  to the semi-strong ones.

Figure~\ref{fig4} compares standard estimator $\hat{\Omega}_S=2\hat\sigma_\epsilon^2\hat{J}_n^{-1}$ and the sandwich estimator $\hat{\Omega}=\hat{J}_n^{-1}\hat{I}^{\mathrm{SP}}_n\hat{J}_n^{-1}$ of the LSE asymptotic variance $\Omega$. We used the spectral estimator $\hat{I}^{\mathrm{SP}}_n$  defined in Theorem \ref{convergence_Isp}. The multivariate AR  order $r$ (see \eqref{AR_tronquee}) is automatically selected by AIC (we use the function  {\tt VARselect()} of the {\bf vars} R package). In the strong FARIMA case we know that the two estimators are consistent. In view of the two upper subfigures of Figure~\ref{fig4}, it seems that the sandwich estimator is less accurate in the strong case. This is not surprising because the sandwich estimator is more robust, in the sense that this estimator remains consistent in the semi-strong and weak FARIMA cases, contrary to the standard estimator (see the middle and bottom subfigures of Figure~\ref{fig4}). Figure~\ref{fig5} (resp. Figure~ \ref{fig6}) presents  a zoom of the left(right)-middle and left(right)-bottom panels of Figure~\ref{fig4}. It is clear that in the semi-strong or weak case $n(\hat{a}_n-a)^2$,  $n(\hat{b}_n-b)^2$ and $n(\hat{d}_n-d)^2$ are, respectively, better estimated by $\hat{J}_n^{-1}\hat{I}^{\mathrm{SP}}_n\hat{J}_n^{-1}(1,1)$, $\hat{J}_n^{-1}\hat{I}^{\mathrm{SP}}_n\hat{J}_n^{-1}(2,2)$ and $\hat{J}_n^{-1}\hat{I}^{\mathrm{SP}}_n\hat{J}_n^{-1}(3,3)$ (see Figure~\ref{fig6})
than by $2\hat\sigma_\epsilon^2\hat{J}_n^{-1}(1,1)$, $2\hat\sigma_\epsilon^2\hat{J}_n^{-1}(2,2)$ and $2\hat\sigma_\epsilon^2\hat{J}_n^{-1}(3,3)$ (see Figure~\ref{fig5}). The failure of the standard estimator of $\Omega$ in the weak FARIMA framework may have important consequences in terms of identification or hypothesis testing and validation.

Now we are interested in standard confidence interval and the modified versions proposed in Subsections \ref{ii} and \ref{sn}. 
Table \ref{tab1} displays the empirical sizes in the three previous different FARIMA cases. For the nominal level $\alpha=5\%$, the empirical  size over the $N=1,000$ independent replications should
vary between the significant limits 3.6\% and 6.4\% with probability 95\%. For the nominal level $\alpha=1\%$, the significant limits are 0.3\% and 1.7\%, and for the nominal level $\alpha=10\%$, they are
8.1\% and 11.9\%. When   the relative rejection frequencies  are outside the significant limits, they are displayed in bold type in Table \ref{tab1}. For the strong FARIMA model, all the relative rejection  frequencies are inside the significant limits for $n$ large. For the semi-strong FARIMA model, the relative rejection frequencies of the standard  confidence interval
are definitely outside the significant limits, contrary to  the modified versions proposed. For the weak FARIMA model, only
 the standard confidence interval of $\hat{b}_n$ is  outside the significant limits when $n$ increases.
As a conclusion,  Table \ref{tab1}  confirms  the comments done concerning Figure~\ref{fig4}.

\subsection{Application to real data}
We now consider an application to the daily returns of four stock market indices (CAC, DAX, Nikkei and S\&P 500).
The returns are defined by $r_t=\log(p_t/p_{t-1})$ where $p_t$ denotes the price index  of the stock market indices  at time $t$.
The observations cover the period from the starting date  of each index  to February 14, 2019. 

In Financial Econometrics the returns are often assumed to be  a white noise.
In view of the so-called volatility clustering, it is well known that the
strong white noise model is not adequate for these series
(see for instance \cite{FZ2010,lobatoNS2001,BMCF2012,BMS2018}). A long-range memory property of the stock market returns series was largely investigated by \cite{DGE1993} which shown that there are more correlation between power transformation of the absolute return $|r_t|^v$ ($v>0$) than returns themselves (see also \cite{Beran2013}, \cite{palma}, \cite{baillie1996} and \cite{ling-li}). We choose here the case where $v=2$ which corresponds to the squared returns
$(r_t^2)_{t\ge 1}$ process.
This process have significant positive autocorrelations at least up to lag 100 (see Figure~\ref{acf})  which confirm  the claim that stock market returns have long-term memory (see \cite{DGE1993}). 

We fit a FARIMA$(1,d,1)$ model to the squares of the 4 daily returns. As in \cite{Ling2003}, we denote by $(X_t)_{t\ge 1}$ the mean corrected
series of the squared returns and we adjust  the following model
\begin{eqnarray*}
(1-L)^d\left(X_t-aX_{t-1}\right)=\epsilon_t-b\epsilon_{t-1} .
\end{eqnarray*}
Figure~\ref{graph1} (resp. Figure~\ref{graph}) plots the closing prices (resp. the returns) of the four stock market indices. 
Figure~\ref{acf} shows that squared returns $(X_t)_{t\ge 1}$ are generally strongly autocorrelated.
Table~\ref{tabreal} displays the LSE of the parameter $\theta=(a,b,d)$ of each squared of daily returns.
The $p-$values  of the corresponding LSE, $\hat\theta_n=(\hat{a}_n,\hat{b}_n,\hat{d}_n)$ are given in parentheses. The last column presents the estimated residual variance. Note that for all series, the estimated coefficients $|\hat{a}_n|$ and $|\hat{b}_n|$ are smaller than one and this is in accordance with our Assumption {\bf (A1)}. We also observe that for all series  the estimated long-range dependence coefficients $\hat{d}_n$ are significant for any reasonable asymptotic level and are inside $]-0.5{,}0.5[$.
We thus think that the assumption {\bf (A3)} is  satisfied and thus our asymptotic normality theorem can be applied.
Table~\ref{tab1real} then presents for each serie the modified confidence interval at the asymptotic level $\alpha=5\%$ for the parameters estimated in Table~\ref{tabreal}.

\section{Conclusion}\label{conclusion}
Taking into account the possible lack of independence of the error terms, we show in this paper that we can  fit FARIMA representations of a wide class of
 nonlinear long memory times series. This is possible thanks to our theoretical results and it is illustrated in our real cases and simulations studies. 

This standard methodology (when the noise is supposed to be iid), in particular the significance tests on the parameters,
needs however to be adapted to take into account the possible lack of independence of the errors terms. A first step has been done thanks to our results on the confidence intervals. In future works, we intent to
study how the existing identification (see \cite{yac_jtsa},  \cite{BMK2016}) and diagnostic checking (see \cite{BMS2018}, \cite{frz}) procedures should be
adapted in the presence of long-range dependence framework and dependent noise.

\section{Proofs}\label{proofs}
In all our proofs, $K$ is a positive constant that may vary from line to line.
\subsection{Preliminary results}\label{prelim}
In this subsection, we shall give some results on estimations of the coefficient of formal power series that will arise in our study. Some of them are well know on some others are new to our knowledge. We will make some precise comments hereafter.

We begin by recalling the following properties on power series. If for $|z|\le R$, the power series $f(z)=\sum_{i\ge 0}a_iz^i$ and  $g(z)=\sum_{i\ge 0}b_iz^i$ are well defined, then one has $(fg)(z)= \sum_{i\ge 0} c_iz^i$ is also well defined for $|z|\le R$ with the sequence $(c_i)_{i\ge 0}$ which is given by $c=a\ast b$ where $\ast$ denotes the convolution product between $a$ and $b$ defined by $c_i=\sum_{k=0}^i a_kb_{i-k}=\sum_{k=0}^i a_{i-k}b_{k}$. We will make use of the Young inequality that states that if the sequence $a\in \ell^{r_1}$ and $b\in\ell^{r_2}$ and such that $\frac{1}{r_1}+\frac{1}{r_2}=1+\frac{1}r$ with $1\le r_1,r_2,r \le \infty$, then $$
\left \| a\ast b \right \|_{\ell^r} \le  \left \| a \right \|_{\ell^{r_1}} \times \left \|  b \right \|_{\ell^{r_2}} .$$

Now we come back to the power series that arise in our context. Remind that for the true value of the parameter,
\begin{equation}\label{FF}
a_{\theta_0}(L)(1-L)^{d_0}X_t=b_{\theta_0}(L)\epsilon_t.
\end{equation}
Thanks to the assumptions on the moving average polynomials $b_\theta$ and the autoregressive polynomials $a_\theta$, the power series $a_\theta^{-1}$ and $b_\theta^{-1}$ are well defined.

Thus the functions $\epsilon_t(\theta)$ defined in \eqref{FARIMA-th}  can be written as
\begin{align}\label{epsi}
\epsilon_t(\theta) &= b^{-1}_{\theta}(L) a_{\theta}(L)(1-L)^{d}X_t \\
\label{epsi-bis}
& =b^{-1}_{\theta}(L) a_{\theta}(L)(1-L)^{d-d_0}a^{-1}_{\theta_0}(L) b_{\theta_0}(L)\epsilon_t
\end{align}
and if we denote $\gamma(\theta)=(\gamma_i(\theta))_{i\ge 0}$ the sequence of coefficients of the power series $b^{-1}_{\theta}(z) a_{\theta}(z)(1-z)^{d}$,
 we may write for all $t\in\mathbb Z$:
\begin{align}\label{AR-Inf}
\epsilon_t(\theta)&=\sum_{i\geq 0}\gamma_i(\theta)X_{t-i}.
\end{align}
In the same way, by \eqref{epsi} one has
\begin{align*}
X_t & = (1-L)^{-d}a^{-1}_{\theta}(L) b_{\theta}(L)\epsilon_t(\theta)
\end{align*}
and if we denote $\eta(\theta)=(\eta_i(\theta))_{i\ge 0}$ the coefficients of the power series $(1-z)^{-d}a^{-1}_{\theta}(z) b_{\theta}(z)$ one has \begin{align}
\label{MA-Inf}
X_t & = \sum_{i\geq 0}\eta_i(\theta)\epsilon_{t-i}(\theta) \ .
\end{align}
We strength the fact that $\gamma_0(\theta)=\eta_0(\theta)=1$ for all $\theta$. 

For large $j$, \cite{hallin} have shown that uniformly in $\theta$ the sequences  $\gamma(\theta)$ and $\eta(\theta)$
satisfy
\begin{equation}\label{Coef-Gamma}
\frac{\partial^k\gamma_j(\theta)}{\partial\theta_{i_1}\cdots\partial\theta_{i_k}}=\mathrm{O}\left( j^{-1-d}\left\lbrace \log(j)\right\rbrace ^k\right),\text{ for }k=0,1,2,3,
\end{equation}
and
\begin{equation}\label{Coef-Eta}
\frac{\partial^k\eta_j(\theta)}{\partial\theta_{i_1}\cdots\partial\theta_{i_k}}=\mathrm{O}\left( j^{-1+d}\left\lbrace \log(j)\right\rbrace ^k\right), \text{ for }k=0,1,2,3.
\end{equation}

One difficulty that has to be addressed is that \eqref{AR-Inf} includes the infinite past $(X_{t-i})_{i\ge 0}$ whereas only a finite number of observations $(X_t)_{1\leq t\leq n}$ are available to compute the estimators defined in \eqref{theta_chap}.
The simplest solution is truncation which amounts to setting all unobserved values equal to zero. Thus, 
for all $\theta\in\Theta$ and $1\le t\le n$ one defines
\begin{equation}\label{epsilon-tilde--}
\tilde{\epsilon}_t(\theta)=\sum_{i=0}^{t-1}\gamma_i(\theta)X_{t-i}= \sum_{i\ge 0} \gamma_i^t(\theta) X_{t-i}
\end{equation}
where the truncated sequence $\gamma^t(\theta)= ( \gamma_i^t(\theta))_{i\ge 0}$ is defined by
$$\gamma_i^t(\theta)=\left\{
\begin{array}{rl}
 \gamma_i(\theta) &\text{ if } \ 0\leq i\leq t-1\ , \\
0& \text{ otherwise.}
\end{array}\right.$$
Since our assumptions are made on the noise in $(\ref{FARIMA})$, it will be useful to express the random variables $\epsilon_t(\theta)$ and its partial derivatives with respect to $\theta$, as a function of $(\epsilon_{t-i})_{i\ge 0}$.

From \eqref{epsi-bis}, there exists a sequence $\lambda(\theta)=(\lambda_i(\theta))_{i\geq0}$ such that
\begin{equation}\label{epsil-th}
\epsilon_t(\theta)=\sum_{i=0}^{\infty}\lambda_i\left( \theta\right) \epsilon_{t-i}
\end{equation}
where the sequence $\lambda(\theta)$ is given by the sequence of the coefficients of the power series $b^{-1}_{\theta}(z) a_{\theta}(z)(1-z)^{d-d_0}a^{-1}_{\theta_0}(z) b_{\theta_0}(z)$. Consequently $\lambda(\theta) = \gamma(\theta)\ast \eta(\theta_0)$ or, equivalently,
\begin{align}\label{Coef-lambda}
\lambda_i( \theta)& =\sum_{j=0}^i\gamma_j(\theta)\eta_{i-j}(\theta_0).
\end{align}
As in \cite{hualde2011}, it can be shown using Stirling's approximation that there exists a positive constant $K$ such that 
\begin{equation}\label{eqasym-lambda}
\sup_{\theta\in\Theta_{\delta}}\left| \lambda_i(\theta)\right| \leq K\sup_{d\in[d_1,d_2]} i^{-1-(d-d_0)}\leq K i^{-1-(d_1-d_0)} \ .
\end{equation}
\noindent Equations \eqref{epsil-th} and \eqref{eqasym-lambda} imply that for all $\theta\in\Theta$ the random variable $\epsilon_t(\theta)$ belongs to $\mathbb{L}^2$, that 
the sequence $(\epsilon_t(\theta))_t$ is an ergodic sequence and that for all $t\in\mathbb{Z}$ the function $\epsilon_t(\cdot)$ is a continuous function.
We proceed in the same way as regard to the derivatives of $\epsilon_t(\theta)$. More precisely, for any $\theta\in\Theta$, $t\in\mathbb{Z}$ and $1\le k,l \le p+q+1$ there exists sequences $\overset{\textbf{.}}{\lambda}_{k}(\theta)= (\overset{\textbf{.}}{\lambda}_{i,k}(\theta))_{i\geq1}$
and $\overset{\textbf{..}}{\lambda}_{k,l}(\theta)= (\overset{\textbf{..}}{\lambda}_{i,k,l}(\theta))_{i\geq1}$ such that
\begin{align}
\frac{\partial\epsilon_t(\theta)}{\partial\theta_k}&=\sum_{i=1}^{\infty}\overset{\textbf{.}}{\lambda}_{i,k}\left( \theta\right) \epsilon_{t-i} 
 \label{deriveesecepsil}\\
\frac{\partial^2\epsilon_t(\theta)}{\partial\theta_k\partial\theta_{l}}& =\sum_{i=1}^{\infty}\overset{\textbf{..}}{\lambda}_{i,k,l}\left( \theta\right) \epsilon_{t-i} .\label{deriveesecepsil1}
\end{align}
Of course it holds that $\overset{\textbf{.}}{\lambda}_{k}(\theta)=\frac{\partial\gamma(\theta)}{\partial\theta_k}\ast\eta(\theta_0)$ and
$\overset{\textbf{..}}{\lambda}_{k,l}( \theta)=\frac{\partial^2\gamma(\theta)}{\partial\theta_k\partial\theta_{l}}\ast \eta(\theta_0)$.

Similarly we have
\begin{align}\label{epsiltilde-th}
\tilde{\epsilon}_t(\theta)& =\sum_{i=0}^{\infty}\lambda_i^t\left( \theta\right) \epsilon_{t-i}, \\
\frac{\partial\tilde{\epsilon}_t(\theta)}{\partial\theta_k}& =\sum_{i=1}^{\infty}\overset{\textbf{.}}{\lambda}_{i,k}^t\left( \theta\right) \epsilon_{t-i},
\label{deriveesecepsiltilde} \\
\frac{\partial^2\tilde{\epsilon}_t(\theta)}{\partial\theta_k\partial\theta_{l}}& =\sum_{i=1}^{\infty}\overset{\textbf{..}}{\lambda}_{i,k,l}^t\left( \theta\right) \epsilon_{t-i},
\end{align}
where $\lambda^t(\theta) = \gamma^t(\theta)\ast \eta(\theta_0)$, $\overset{\textbf{.}}{\lambda}^t_{k}(\theta)=\frac{\partial\gamma^t(\theta)}{\partial\theta_k}\ast\eta(\theta_0)$ and
$\overset{\textbf{..}}{\lambda}^t_{k,l}( \theta)=\frac{\partial^2\gamma^t(\theta)}{\partial\theta_k\partial\theta_{l}}\ast \eta(\theta_0)$.

In order to handle the truncation error $\epsilon_t(\theta)-\tilde\epsilon_t(\theta)$, one needs information on the sequence $\lambda(\theta)-\lambda^t(\theta)$. This is the purpose of the following lemma.

\begin{lemme}\label{lemme_sur_les_ecarts_des_coef}
For  $2\leq r\leq \infty$ and $1\le k,l \le p+q+1 $, we have
\begin{equation*}
\parallel\lambda\left(\theta\right)-\lambda^t\left(\theta\right)\parallel_{\ell^r} \ =\mathrm{O}\left(t^{-1+\frac{1}{r}-(d-max(d_0,0))}\right),
\end{equation*}
\begin{equation*}
\parallel\overset{\textbf{.}}{\lambda}_k\left(\theta\right)-\overset{\textbf{.}}{\lambda}_k^t\left(\theta\right)\parallel_{\ell^r} \ =\mathrm{O}\left(t^{-1+\frac{1}{r}-(d-max(d_0,0))}\right)
\end{equation*}
and 
\begin{equation*}
\parallel\overset{\textbf{..}}{\lambda}_{k,l}\left(\theta\right)-\overset{\textbf{..}}{\lambda}_{k,l}^t\left(\theta\right)\parallel_{\ell^r} \ =\mathrm{O}\left(t^{-1+\frac{1}{r}-(d-max(d_0,0))}\right)
\end{equation*}
for any $\theta\in\Theta_{\delta}$ if $d_0\leq 0$ and for $\theta$ with non-negative memory parameter $d$ if $d_0>0$.

\end{lemme}

\begin{proof}

In view of \eqref{Coef-Eta}, $\eta(\theta_0)\in\ell^{r_2}$ for $r_2\geq 1$ when $d_0<0$. If $d_0=0$, $\eta(\theta_0)$ is the sequence of coefficients of the power series $a^{-1}_{\theta_0}(z) b_{\theta_0}(z)$, so it belongs to $\ell^{r_2}$ for all $r_2\geq 1$ since in this case $|\eta_j(\theta_0)|=\mathrm{O}(\rho^j)$ for some $0<\rho<1$ (see \cite{fz98}). Thanks to \eqref{Coef-Gamma}, when $d_0\leq 0$, Young's inequality for convolution yields that for all $r\geq 2$
\begin{align*}
\parallel\lambda\left(\theta\right)-\lambda^t\left(\theta\right)\parallel_{\ell^r}&\leq K\parallel\gamma(\theta)-\gamma^t(\theta)\parallel_{\ell^r}\\
&\leq K\left(\sum_{i=t}^{\infty}\left|\gamma_i(\theta)\right|^r\right)^{1/r}\\
&\leq K\left(\sum_{i=t}^\infty\frac{1}{i^{r+rd}}\right)^{1/r}\\
&\leq K\left(\int_{t}^{\infty}\frac{1}{x^{r+rd}}\mathrm{dx}+\frac{1}{t^{r+rd}}\right)^{1/r}\\
&\leq \frac{K}{t^{1-\frac{1}{r}+d}}.
\end{align*}

If $d_0>0$, the sequence $\eta(\theta_0)$ belongs to  $\ell^{r_2}$ for any $r_2>1/(1-d_0)$.
Young's inequality for convolution implies in this case that for all $r\ge 2$
\begin{align}\label{inegal_young}
\parallel\lambda\left(\theta\right)-\lambda^t\left(\theta\right)\parallel_{\ell^r} \ \leq \parallel \gamma(\theta)-\gamma^t(\theta)\parallel_{\ell^{r_1}}\parallel \eta(\theta_0)\parallel_{\ell^{r_2}}
\end{align}
with $r_2=(1-(d_0+\beta))^{-1}>1/(1-d_0)$ and $r_1=r/(1+r(d_0+\beta))$, for some $\beta>0$ sufficiently small. Thus there exists $K$ such that $\parallel \eta(\theta_0)\parallel_{\ell^{r_2}}\le K$. Similarly as before, we deduce when $d\geq 0$ that 
\begin{align*}
\parallel\lambda\left(\theta\right)-\lambda^t\left(\theta\right)\parallel_{\ell^r}&\leq K\parallel\gamma(\theta)-\gamma^t(\theta)\parallel_{\ell^{r_1}}\\
&\leq K\left(\sum_{i=t}^\infty\frac{1}{i^{r_1+r_1d}}\right)^{1/r_1}\\
&\leq \frac{K}{t^{1-\frac{1}{r_1}+d}}=\frac{K}{t^{1-\frac{1}{r}+(d-d_0)-\beta}},
\end{align*}
the conclusion follows by tending $\beta$ to 0.
The second and third points of the lemma are shown in the same way as the first. This is because from \eqref{Coef-Gamma}, the coefficients $\partial\gamma_j(\theta)/\partial\theta_k$ and $\partial^2\gamma_j(\theta)/\partial\theta_k\partial\theta_l$ are $\mathrm{O}(j^{-1-d+\zeta})$ for any small enough $\zeta> 0$.
The proof of the lemma is then complete.

\end{proof}

%

\begin{rmq}\label{impo}
The above lemma implies that the sequence $ \overset{\textbf{.}}{\lambda}_k\left(\theta_0\right)-\overset{\textbf{.}}{\lambda^t}_k\left(\theta_0\right)$ is bounded and more precisely there exists $K$ such that
\begin{align}\label{impoeq}
\sup_{j\ge 1} \left | \overset{\textbf{.}}{\lambda}_{j,k}\left(\theta_0\right)-\overset{\textbf{.}}{\lambda^t}_{j,k}\left(\theta_0\right) \right | & \le \frac{K}{t^{1+\min(d_0,0)}}
\end{align} for any $t\geq 1$ and any $1\le k\le p+q+1$.

\end{rmq}
\begin{rmq}\label{rmq:important}
In order to prove our asymptotic results, it will be convenient to give an upper bound for the norms of the sequences introduced in Lemma \ref{lemme_sur_les_ecarts_des_coef} valid for any $\theta\in\Theta_{\delta}$. Since $d_1-d_0>-1/2$, Estimation \eqref{eqasym-lambda} entails that for any $r\geq 2$,
\begin{align*}
\parallel\lambda\left(\theta\right)-\lambda^t\left(\theta\right)\parallel_{\ell^r} \ =\mathrm{O}\left(t^{-1+\frac{1}{r}-(d_1-d_0)}\right), \ \ \ \forall\theta\in\Theta_{\delta}.
\end{align*}
This can easily be seen since $\parallel\lambda(\theta)-\lambda^t(\theta)\parallel_{\ell^r}\leq K(\sum_{i\geq t}i^{-r-r(d_1-d_0)})^{1/r}\leq Kt^{-1+1/r-(d_1-d_0)}$. As in \cite{hallin}, the coefficients $\overset{\textbf{.}}{\lambda}_{j,k}(\theta)$ and $\overset{\textbf{..}}{\lambda}_{j,k,l}(\theta)$ are $\mathrm{O}(j^{-1-(d-d_0)+\zeta})$ for any small enough $\zeta>0$, so we have 
\begin{equation*}
\parallel\overset{\textbf{.}}{\lambda}_k\left(\theta\right)-\overset{\textbf{.}}{\lambda}_k^t\left(\theta\right)\parallel_{\ell^r} \ =\mathrm{O}\left(t^{-1+\frac{1}{r}-(d_1-d_0)+\zeta}\right)
\end{equation*}
and 
\begin{equation*}
\parallel\overset{\textbf{..}}{\lambda}_{k,l}\left(\theta\right)-\overset{\textbf{..}}{\lambda}_{k,l}^t\left(\theta\right)\parallel_{\ell^r} \ =\mathrm{O}\left(t^{-1+\frac{1}{r}-(d_1-d_0)+\zeta}\right)
\end{equation*}
for any $r\geq 2$, any $1\leq k,l\leq p+q+1$ and all $\theta\in\Theta_{\delta}$.
\end{rmq}
One shall also need the following lemmas.
\begin{lemme}\label{miss}
For  any $2\le r\le \infty$, $1\le k \le p+q+1 $ and $\theta\in\Theta$, there exists a constant $K$ such that we have
\begin{equation*}
\parallel\overset{\textbf{.}}{\lambda}_k^t\left(\theta\right)\parallel_{\ell^r} \le K.
\end{equation*}
\end{lemme}
\begin{proof}
The proof follows the same arguments as those developed in Remark \ref{rmq:important}.

\end{proof}
\begin{lemme}\label{miss2}
There exists a constant $K$ such that we have
\begin{equation}\label{eq-miss2}
 \left | \overset{\textbf{.}}{\lambda}_{i,k}\left(\theta_0\right) \right | \le  \frac{K}{i}.
\end{equation}
\end{lemme}
\begin{proof}
For  $1\le k \le p+q+1 $, the sequence $\overset{\textbf{.}}{\lambda}_{k}(\theta)= (\overset{\textbf{.}}{\lambda}_{i,k}(\theta))_{i\geq1}$ is in fact the sequence of  the coefficients in the power series of
$$\frac{\partial}{\partial \theta_k }\left (  b_\theta^{-1}(z) a_{\theta}(z)(1-z)^{d-d_0} a_{\theta_0}^{-1}(z)b_{\theta_0}(z)\right )\ .$$ Thus $\overset{\textbf{.}}{\lambda}_{i,k}\left(\theta_0\right)$ is the $i-$th coefficient taken in $\theta=\theta_0$. There are three cases.
\begin{itemize}
\item[$\diamond$] $k=1,\dots,p$:  \\
Since
$$\frac{\partial}{\partial \theta_k }\left (  b_\theta^{-1}(z) a_{\theta}(z)(1-z)^{d-d_0} a_{\theta_0}^{-1}(z)b_{\theta_0}(z) \right )= -b_\theta^{-1}(z) z^k (1-z)^{d-d_0} a_{\theta_0}^{-1}(z)b_{\theta_0}(z)\ , $$ we deduce that $\overset{\textbf{.}}{\lambda}_{i,k}\left(\theta_0\right)$ is the $i-$th coefficient of $-z^k a_{\theta_0}^{-1}(z)$ which satisfies
$\overset{\textbf{.}}{\lambda}_{i,k}\left(\theta_0\right)\le K \rho^i$ for some $0<\rho<1$ (see \cite{fz98} for example).
\item[$\diamond$] $k=p+1,\dots,p+q$:  \\
We have
$$\frac{\partial}{\partial \theta_k } \left (b_\theta^{-1}(z) a_{\theta}(z)(1-z)^{d-d_0} a_{\theta_0}(z)b_{\theta_0}(z)\right ) = \left  (\frac{\partial}{\partial \theta_k } b_\theta^{-1}(z)\right  ) a_{\theta}(z) (1-z)^{d-d_0} a_{\theta_0}^{-1}(z)b_{\theta_0}(z) $$ and consequently $\overset{\textbf{.}}{\lambda}_{i,k}\left(\theta_0\right)$ is the $i-$th coefficient of $(\frac{\partial}{\partial \theta_k } b_{\theta_0}^{-1}(z) ) b_{\theta_0}(z)$ which also satisfies
$\overset{\textbf{.}}{\lambda}_{i,k}\left(\theta_0\right)\le K \rho^i$ (see \cite{fz98}).

The last case will not be a consequence of the usual works on ARMA processes.
\item[$\diamond$] $k=p+q+1$:  \\
In this case, $\theta_k=d$ and so we have
$$\frac{\partial}{\partial \theta_k } \left (b_\theta^{-1}(z) a_{\theta}(z)(1-z)^{d-d_0} a_{\theta_0}^{-1}(z)b_{\theta_0}(z)\right ) = b_\theta^{-1}(z)a_{\theta}(z) \mathrm{ln}(1-z)(1-z)^{d-d_0} a_{\theta_0}^{-1}(z)b_{\theta_0}(z) $$ and consequently $\overset{\textbf{.}}{\lambda}_{i,k}\left(\theta_0\right)$ is the $i-$th coefficient of $\mathrm{ln}(1-z)$ which is equal to $-1/i$.
\end{itemize}
The three above cases imply the expected result.
\end{proof}

\subsection{Proof of Theorem \ref{convergence}}\label{sectconv}

Consider the random variable $\mathrm{W}_n(\theta)$ defined for any $\theta\in\Theta$ by
\begin{equation*}
\mathrm{W}_n(\theta)=\mathrm{V}(\theta)+Q_n(\theta_0)-Q_n(\theta),
\end{equation*}
where $V(\theta)=\mathbb{E}[O_n(\theta)]-\mathbb{E}[O_n(\theta_0)]$. For $\beta>0$, let $S_\beta=\{\theta: \|\theta-\theta_0\|\leq \beta\}$, $\overline{S}_\beta=\{\theta\in\Theta_\delta: \theta\notin S_\beta\}$. It can readily be shown that 
\begin{align}\label{std-conv}
\mathbb{P}\left(\left\|\hat{\theta}_n-\theta_0\right\|>\beta\right)&\leq \mathbb{P}\left(\hat{\theta}_n\in\overline{S}_\beta\right)\nonumber\\
&\leq \mathbb{P}\left(\inf_{\theta\in\overline{S}_\beta} \left\lbrace Q_n(\theta)-Q_n(\theta_0)\right\rbrace\leq 0\right)\nonumber\\
&\leq \mathbb{P}\left(\sup_{\theta\in\Theta_\delta}\left|\mathrm{W}_n(\theta)\right|\geq \inf_{\theta\in\overline{S}_\beta} \mathrm{V}(\theta)\right)\nonumber\\
&\leq\mathbb{P}\left(\sup_{\theta\in\Theta_\delta}\left|Q_n(\theta)-\mathbb{E}\left[O_n(\theta)\right]\right|\geq \frac{1}{2}\inf_{\theta\in\overline{S}_\beta} \mathrm{V}(\theta)\right).
\end{align}
Since $d_1-d_0>-1/2$, one has 
\begin{align}\label{ineq:momoent2}
\sup_{\theta\in\Theta_\delta}\mathbb{E}\left[\epsilon_t^2(\theta)\right]&=\sup_{\theta\in\Theta_\delta}\sum_{i=0}^\infty\sum_{j=0}^\infty\lambda_i(\theta)\lambda_j(\theta)\mathbb{E}\left[\epsilon_{t-i}\epsilon_{t-j}\right]=\sigma_{\epsilon}^2\sup_{\theta\in\Theta_\delta}\sum_{i=0}^\infty\lambda_i^2(\theta)\nonumber\\
&\leq \sigma_{\epsilon}^2+K\sigma_{\epsilon}^2\sum_{i=1}^\infty i^{-2-2(d_1-d_0)}<\infty.
\end{align}
We can therefore use the same arguments as those of \cite{fz98} to prove under {\bf (A1)} and {\bf (A2)} that for any $\overline{\theta}\in\Theta_\delta\setminus\{\theta_0\}$, there exists a neighbourhood $\mathrm{N}(\overline{\theta})$ of $\overline{\theta}$ such that $\mathrm{N}(\overline{\theta})\subset\Theta_\delta$ and
\begin{equation}\label{eq:fz98}
\liminf_{n\rightarrow\infty}\inf_{\theta\in \mathrm{N}(\overline{\theta})}O_n(\theta)>\sigma^2_{\epsilon}, \quad \text{ a.s}.
\end{equation}
Note that $\mathbb{E}[O_n(\theta_0)]=\sigma_\epsilon^2$. It follows from \eqref{eq:fz98} that 
$$\inf_{\theta\in\overline{S}_\beta}\mathrm{V}(\theta)>K$$
for some positive constant $K$. \\
In view of \eqref{std-conv}, it is then sufficient to show that the random variable $\sup_{\theta\in\Theta_\delta}|Q_n(\theta)-\mathbb{E}[O_n(\theta)]|$ converges in probability to zero to prove Theorem \ref{convergence}. We use Corollary 2.2 of \cite{Whitney-1991} to obtain this uniform convergence in probability. The set $\Theta_\delta$ is compact and $(\mathbb{E}[O_n(\theta)])_{n\geq 1}$ is a uniformly convergent sequence of continuous functions on a compact set so it is equicontinuous. We consequently need to show the following two points to complete the proof of the theorem:\\

\begin{itemize}
\item[$\bullet$] For each $\theta\in\Theta_\delta$, $Q_n(\theta)-\mathbb{E}[O_n(\theta)]=\mathrm{o}_{\mathbb{P}}(1).$\\

\item[$\bullet$] There is $B_n$ and $h: [0,\infty)\rightarrow [0,\infty)$ with $h(0)=0$ and $h$ continuous at zero such that $B_n=\mathrm{O}_{\mathbb{P}}(1)$ and for all $\theta_1, \theta_2\in \Theta_\delta$, $|Q_n(\theta_1)-Q_n(\theta_2)|\leq B_nh(\|\theta_1-\theta_2\|)$.
\end{itemize}

\subsubsection{Pointwise convergence in probability of $Q_n(\theta)-\mathbb{E}[O_n(\theta)]$ to zero}

For any $\theta\in\Theta_\delta$, Remark \ref{rmq:important}, the Cauchy-Schwarz inequality and \eqref{ineq:momoent2} yield that
\begin{align}\label{firstpart-conv}
\mathbb{E}\left|Q_n(\theta)-O_n(\theta)\right|&\leq \frac{1}{n}\sum_{t=1}^n\mathbb{E}\left|\tilde{\epsilon}_t^2(\theta)-\epsilon_t^2(\theta)\right|\nonumber\\
&\leq \frac{1}{n}\sum_{t=1}^n\mathbb{E}\left(\epsilon_t(\theta)-\tilde{\epsilon}_t(\theta)\right)^2+\frac{2}{n}\sum_{t=1}^n\mathbb{E}\left|\epsilon_t(\theta)-\tilde{\epsilon}_t(\theta)\right|\left|\epsilon_t(\theta)\right|\nonumber\\
&\leq \frac{\sigma_\epsilon^2}{n}\sum_{t=1}^n\left\|\lambda_i(\theta)-\lambda_i^t(\theta)\right\|_{\ell^2}^2+\frac{2\sigma_\epsilon}{n}\sum_{t=1}^n\left\|\lambda_i(\theta)-\lambda_i^t(\theta)\right\|_{\ell^2}\sqrt{\mathbb{E}\left[\epsilon_t^2(\theta)\right]}\nonumber\\
&\leq \frac{\sigma_\epsilon^2}{n}\sum_{t=1}^nt^{-1-2(d_1-d_0)}+\frac{2K\sigma_\epsilon}{n}\sum_{t=1}^nt^{-1/2-(d_1-d_0)}\xrightarrow[n\to \infty]{} 0.
\end{align}

We use the ergodic theorem and the continuous mapping theorem to obtain
\begin{align}\label{ergo}
\left|O_n(\theta)-\mathbb{E}\left[O_n(\theta)\right]\right|\xrightarrow[n\to \infty]{\mathrm{a.s.}} 0.
\end{align}
Combining the results in \eqref{firstpart-conv} and \eqref{ergo}, we deduce that for all $\theta\in\Theta_\delta$,
\begin{align*}
Q_n(\theta)-\mathbb{E}\left[O_n(\theta)\right]\xrightarrow[n\to \infty]{\mathrm{\mathbb{P}}} 0.
\end{align*}

\subsubsection{Tightness characterization}

Observe that for any $\theta_1, \theta_2\in\Theta_\delta$, there exists $\theta^\star$ between $\theta_1$ and $\theta_2$ such that
\begin{align*}
\left|Q_n(\theta_1)-Q_n(\theta_2)\right|\leq \left(\frac{1}{n}\sum_{t=1}^n\left\|\frac{\partial\tilde{\epsilon}_t^2(\theta^\star)}{\partial\theta}\right\|\right)\left\|\theta_1-\theta_2\right\|.
\end{align*}
As before, the uncorrelatedness of the innovation process $(\epsilon_t)_{t\in\mathbb{Z}}$ and Remark \ref{rmq:important} entail that
\begin{align*}
\mathbb{E}\left[\frac{1}{n}\sum_{t=1}^n\left\|\frac{\partial\tilde{\epsilon}_t^2(\theta^\star)}{\partial\theta}\right\|\right]&\leq \mathbb{E}\left[\frac{2}{n}\sum_{t=1}^n\sum_{i=1}^{p+q+1}\left|\tilde{\epsilon}_t(\theta^\star)\frac{\partial\tilde{\epsilon}_t(\theta^\star)}{\partial\theta_i}\right|\right]<\infty.
\end{align*}
Thanks to Markov's inequality, we conclude that
\begin{align*}
\frac{1}{n}\sum_{t=1}^n\left\|\frac{\partial\tilde{\epsilon}_t^2(\theta^\star)}{\partial\theta}\right\|=\mathrm{O}_{\mathbb{P}}(1).
\end{align*}
The proof of Theorem \ref{convergence} is then complete.
\subsection{Proof of Theorem \ref{n.asymptotique}}\label{sectna}
By a Taylor expansion of the function $\partial Q_n(\cdot)/ \partial\theta $ around $\theta_0$ and under {\bf (A3)}, we have
\begin{equation}\label{premierTaylor}
0=\sqrt{n}\frac{\partial}{\partial\theta}Q_n(\hat{\theta}_n)=\sqrt{n}\frac{\partial}{\partial\theta}Q_n(\theta_0)+\left[\frac{\partial^2}{\partial\theta_i\partial\theta_j}Q_n\left( \theta^*_{n,i,j}\right) \right] \sqrt{n}\left( \hat{\theta}_n-\theta_0\right),
\end{equation}
where the $\theta^*_{n,i,j}$'s are between $\hat{\theta}_n$ and $\theta_0$. 
The equation $\eqref{premierTaylor}$ can be rewritten in the form:
\begin{align}\label{Taylormodifie}
&\sqrt{n}\frac{\partial}{\partial\theta}O_n(\theta_0)-\sqrt{n}\frac{\partial}{\partial\theta}Q_n(\theta_0)= \sqrt{n}\frac{\partial}{\partial\theta}O_n(\theta_0)+\left[\frac{\partial^2}{\partial\theta_i\partial\theta_j}Q_n\left( \theta^*_{n,i,j}\right) \right] \sqrt{n}\left( \hat{\theta}_n-\theta_0\right).
\end{align}
Under the assumptions of Theorem~\ref{n.asymptotique}, it will be shown respectively in Lemma~\ref{ConProQversO} and Lemma~\ref{Conv_almost_sure_J_n_vers_J} that
\begin{equation*}
\sqrt{n}\frac{\partial}{\partial\theta}O_n(\theta_0)-\sqrt{n}\frac{\partial}{\partial\theta}Q_n(\theta_0)=\mathrm{o}_{\mathbb{P}}(1),
\end{equation*}
and
\begin{equation*}
\left[\frac{\partial^2}{\partial\theta_i\partial\theta_j}Q_n\left( \theta^*_{n,i,j}\right) \right]- J(\theta_0)=\mathrm{o}_{\mathbb{P}}(1).
\end{equation*}
As a consequence, the asymptotic normality of $\sqrt{n}( \hat{\theta}_n-\theta_0)$ will be a consequence of the one of $\sqrt{n}\partial/\partial\theta O_n(\theta_0)$.
\begin{lemme}\label{ConProQversO}
 For  $1\le k\le p+q+1$,  under the assumptions of Theorem~\ref{n.asymptotique}, we have
 \begin{equation}\label{der.o_p}
\sqrt{n}\left(\frac{\partial}{\partial\theta_k}Q_n(\theta_0)-\frac{\partial}{\partial\theta_k}O_n(\theta_0)\right)=\mathrm{o}_{\mathbb{P}}(1).
\end{equation}
\end{lemme}
\begin{proof}

Throughout this proof, $\theta=(\theta_1,...,\theta_{p+q},d)'\in\Theta_\delta$ is such that $\max(d_0,0)<d\leq d_2$ where $d_2$ is the upper bound of the support of the long-range dependence parameter $d_0$.

The proof is quite long so we divide it in several steps.
\paragraph{$\diamond$ Step 1: preliminaries} \ \\

For $1\le k \le p+q+1$ we have
\begin{align}\label{decompositionecartQO}
\sqrt{n}\frac{\partial}{\partial\theta_k}Q_n(\theta_0)&=\frac{2}{\sqrt{n}}\sum_{t=1}^n\tilde{\epsilon}_t(\theta_0)\frac{\partial}{\partial\theta_k}\tilde{\epsilon}_t(\theta_0)\nonumber \\
& =\frac{2}{\sqrt{n}}\sum_{t=1}^n\left(\tilde{\epsilon}_t(\theta_0)-\tilde{\epsilon}_t(\theta)\right)\frac{\partial}{\partial\theta_k}\tilde{\epsilon}_t(\theta_0)+\frac{2}{\sqrt{n}}\sum_{t=1}^n\left(\tilde{\epsilon}_t(\theta)-\epsilon_t(\theta)\right)\frac{\partial}{\partial\theta_k}\tilde{\epsilon}_t(\theta_0)\nonumber\\
&\qquad +\frac{2}{\sqrt{n}}\sum_{t=1}^n\left(\epsilon_t(\theta)-\epsilon_t(\theta_0)\right)\frac{\partial}{\partial\theta_k}\tilde{\epsilon}_t(\theta_0)
+\frac{2}{\sqrt{n}}\sum_{t=1}^n\epsilon_t(\theta_0)\left(\frac{\partial}{\partial\theta_k}\tilde{\epsilon}_t(\theta_0)-\frac{\partial}{\partial\theta_k}\epsilon_t(\theta_0)\right)\nonumber\\
& \qquad +\frac{2}{\sqrt{n}}\sum_{t=1}^n\epsilon_t(\theta_0)\frac{\partial}{\partial\theta_k}\epsilon_t(\theta_0)\nonumber\\
&=\Delta_{n,1}^k(\theta)+\Delta_{n,2}^k(\theta)+\Delta_{n,3}^k(\theta)+\Delta_{n,4}^k(\theta_0)+\sqrt{n}\frac{\partial}{\partial\theta_k}O_n(\theta_0),
\end{align}
where
\begin{align*}
 \Delta_{n,1}^k(\theta)&=\frac{2}{\sqrt{n}}\sum_{t=1}^n\left(\tilde{\epsilon}_t(\theta_0)-\tilde{\epsilon}_t(\theta)\right)\frac{\partial}{\partial\theta_k}\tilde{\epsilon}_t(\theta_0),\nonumber\\
 \Delta_{n,2}^k(\theta)&=\frac{2}{\sqrt{n}}\sum_{t=1}^n\left(\tilde{\epsilon}_t(\theta)-\epsilon_t(\theta)\right)\frac{\partial}{\partial\theta_k}\tilde{\epsilon}_t(\theta_0),\nonumber\\
 \Delta_{n,3}^k(\theta)&=\frac{2}{\sqrt{n}}\sum_{t=1}^n\left(\epsilon_t(\theta)-\epsilon_t(\theta_0)\right)\frac{\partial}{\partial\theta_k}\tilde{\epsilon}_t(\theta_0)\nonumber 
\end{align*}
and
\begin{align*}
\Delta_{n,4}^k(\theta_0)&=\frac{2}{\sqrt{n}}\sum_{t=1}^n\epsilon_t(\theta_0)\left(\frac{\partial}{\partial\theta_k}\tilde{\epsilon}_t(\theta_0)-\frac{\partial}{\partial\theta_k}\epsilon_t(\theta_0)\right).
\end{align*}
Using \eqref{deriveesecepsil} and \eqref{deriveesecepsiltilde}, the fourth term $\Delta_{n,4}^k(\theta_0)$ can be rewritten in the form:
\begin{equation}\label{delta1explicite}
\Delta_{n,4}^k(\theta_0)=\frac{2}{\sqrt{n}}\sum_{t=1}^n\sum_{j=1}^{\infty}\left\lbrace \overset{\textbf{.}}{\lambda}^t_{j,k}\left( \theta_0\right)-\overset{\textbf{.}}{\lambda}_{j,k}\left( \theta_0\right)\right\rbrace  \epsilon_t\epsilon_{t-j}.
\end{equation}
Therefore, if we prove that the three sequences of random variables $( \Delta_{n,1}^k(\theta)+\Delta_{n,3}^k(\theta))_{n\geq 1}$, $( \Delta_{n,2}^k(\theta))_{n\geq 1}$ and $( \Delta_{n,4}^k(\theta_0))_{n\geq 1}$ converge in probability to $0$, then  \eqref{der.o_p} will be true.
\paragraph{$\diamond$ Step 2: convergence in probability of $( \Delta_{n,4}^k(\theta_0))_{n\geq 1}$ to $0$} \ \\

For simplicity, we denote in the sequel by $\overset{\textbf{.}}{\lambda}_{j,k}$ the coefficient $\overset{\textbf{.}}{\lambda}_{j,k}(\theta_0)$ and by $\overset{\textbf{.}}{\lambda}_{j,k}^t$ the coefficient $\overset{\textbf{.}}{\lambda}_{j,k}^t(\theta_0)$.
Let $\varrho(\cdot,\cdot)$ be the function defined for $1\le t,s\le n$ by
$$\varrho(t,s)=\sum_{j_1=1}^{\infty}\sum_{j_2=1}^{\infty}\left\lbrace \overset{\textbf{.}}{\lambda}_{j_1,k}-\overset{\textbf{.}}{\lambda}_{j_1,k}^t\right\rbrace\left\lbrace \overset{\textbf{.}}{\lambda}_{j_2,k}-\overset{\textbf{.}}{\lambda}_{j_2,k}^s\right\rbrace\mathbb{E}\left[ \epsilon_t\epsilon_{t-j_1} \epsilon_s\epsilon_{s-j_2}\right]. $$
For all $\beta>0$, using the symmetry of the function $\varrho(t,s)$, we obtain that
\begin{align*}
\mathbb{P}\left( \left| \Delta_{n,4}^k(\theta_0)\right|\geq\beta\right)&\leq \frac{4}{n\beta^2}\mathbb{E}\left[\left(  \sum_{t=1}^n\sum_{j=1}^{\infty}\left\lbrace \overset{\textbf{.}}{\lambda}_{j,k}-\overset{\textbf{.}}{\lambda}_{j,k}^t\right\rbrace  \epsilon_t\epsilon_{t-j}\right) ^2\right] \\
&\leq \frac{4}{n\beta^2}\sum_{t=1}^n\sum_{s=1}^n\sum_{j_1=1}^{\infty}\sum_{j_2=1}^{\infty}\left\lbrace \overset{\textbf{.}}{\lambda}_{j_1,k}-\overset{\textbf{.}}{\lambda}_{j_1,k}^t\right\rbrace\left\lbrace \overset{\textbf{.}}{\lambda}_{j_2,k}-\overset{\textbf{.}}{\lambda}_{j_2,k}^s\right\rbrace\mathbb{E}\left[ \epsilon_t\epsilon_{t-j_1} \epsilon_s\epsilon_{s-j_2}\right] \\
&\leq \frac{8}{n\beta^2}\sum_{t=1}^n\sum_{s=1}^t\sum_{j_1=1}^{\infty}\sum_{j_2=1}^{\infty}\left\lbrace \overset{\textbf{.}}{\lambda}_{j_1,k}-\overset{\textbf{.}}{\lambda}_{j_1,k}^t\right\rbrace\left\lbrace \overset{\textbf{.}}{\lambda}_{j_2,k}-\overset{\textbf{.}}{\lambda}_{j_2,k}^s\right\rbrace\mathbb{E}\left[ \epsilon_t\epsilon_{t-j_1} \epsilon_s\epsilon_{s-j_2}\right].
\end{align*}
By the stationarity of $(\epsilon_t)_{t\in\mathbb{Z}}$ which is assumed in {\bf (A2)},
we have
\begin{align*}
 \mathbb{E}\left[ \epsilon_t\epsilon_{t-j_1} \epsilon_s\epsilon_{s-j_2}\right]  =
 \mathrm{cum}\left( \epsilon_0,\epsilon_{-j_1}, \epsilon_{s-t},\epsilon_{s-t-j_2}\right)
&+\mathbb{E}\left[\epsilon_0\epsilon_{-j_1}\right]\mathbb{E}\left[\epsilon_{s-t}\epsilon_{s-t-j_2}\right]+\mathbb{E}\left[\epsilon_0\epsilon_{s-t}\right]\mathbb{E}\left[\epsilon_{-j_1}\epsilon_{s-t-j_2}\right] \\
& +\mathbb{E}\left[\epsilon_0\epsilon_{s-t-j_2}\right]\mathbb{E}\left[\epsilon_{-j_1}\epsilon_{s-t}\right].
\end{align*}
Since the noise is not correlated, we deduce that  $\mathbb{E}\left[\epsilon_0\epsilon_{-j_1}\right]=0$ and $\mathbb{E}\left[\epsilon_0\epsilon_{s-t-j_2}\right]=0$ for $1\le j_1, j_2$ and $s\leq t$. Consequently we obtain
\begin{align}\label{delta_n4_prob}
\mathbb{P}\left( \left| \Delta_{n,4}^k(\theta_0)\right|\geq\beta\right)&\leq\frac{8}{n\beta^2}\sum_{t=1}^n\sum_{s=1}^t\sum_{j_1=1}^{\infty}\sum_{j_2=1}^{\infty}\sup_{j_1\ge 1}\left| \overset{\textbf{.}}{\lambda}_{j_1,k}-\overset{\textbf{.}}{\lambda}_{j_1,k}^t\right|\left|\overset{\textbf{.}}{\lambda}_{j_2,k}-\overset{\textbf{.}}{\lambda}_{j_2,k}^s\right| \left|\mathrm{cum}\left( \epsilon_0,\epsilon_{-j_1}, \epsilon_{s-t},\epsilon_{s-t-j_2}\right)\right|\nonumber\\
&+\frac{8}{n\beta^2}\sum_{t=1}^n\sum_{s=1}^t\sum_{j_1=1}^{\infty}\sum_{j_2=1}^{\infty}\left| \overset{\textbf{.}}{\lambda}_{j_1,k}-\overset{\textbf{.}}{\lambda}_{j_1,k}^t\right|\left|\overset{\textbf{.}}{\lambda}_{j_2,k}-\overset{\textbf{.}}{\lambda}_{j_2,k}^s\right| \left|\mathbb{E}\left[\epsilon_0\epsilon_{s-t}\right]\mathbb{E}\left[\epsilon_{-j_1}\epsilon_{s-t-j_2}\right]\right|.
\end{align}
If
\begin{equation}\label{cesaro_lemma}
\sum_{s=1}^t\sum_{j_1=1}^{\infty}\sum_{j_2=1}^{\infty}\sup_{j_1\ge 1}\left| \overset{\textbf{.}}{\lambda}_{j_1,k}-\overset{\textbf{.}}{\lambda}_{j_1,k}^t\right|\left|\overset{\textbf{.}}{\lambda}_{j_2,k}-\overset{\textbf{.}}{\lambda}_{j_2,k}^s\right| \left|\mathrm{cum}\left( \epsilon_0,\epsilon_{-j_1}, \epsilon_{s-t},\epsilon_{s-t-j_2}\right)\right|\xrightarrow[t\to\infty]{} 0,
\end{equation}
Ces\`{a}ro's Lemma implies that the first term in the right hand side of  \eqref{delta_n4_prob} tends to $0$.  Thanks to  Lemma \ref{lemme_sur_les_ecarts_des_coef} applied with $r=\infty$ (or see Remark \ref{impo}) and  Assumption {\bf (A4')} with $\tau=4$, we obtain that
\begin{align*}
\sum_{s=1}^t\sum_{j_1=1}^{\infty}\sum_{j_2=1}^{\infty}\sup_{j_1\ge 1}\left| \overset{\textbf{.}}{\lambda}_{j_1,k}-\overset{\textbf{.}}{\lambda}_{j_1,k}^t\right|&\left|\overset{\textbf{.}}{\lambda}_{j_2,k}-\overset{\textbf{.}}{\lambda}_{j_2,k}^s\right| \left|\mathrm{cum}\left( \epsilon_0,\epsilon_{-j_1}, \epsilon_{s-t},\epsilon_{s-t-j_2}\right)\right|\\
&\leq \frac{K}{t^{1+\min(d_0,0)}}\sum_{s=1}^t\sum_{j_1=1}^{\infty}\sum_{j_2=1}^{\infty} \left|\mathrm{cum}\left( \epsilon_0,\epsilon_{-j_1}, \epsilon_{s-t},\epsilon_{s-t-j_2}\right)\right|\\
&\leq \frac{K}{t^{1+\min(d_0,0)}}\sum_{s=-\infty}^{\infty}\sum_{j_1=-\infty}^{\infty}\sum_{j_2=-\infty}^{\infty} \left|\mathrm{cum}\left( \epsilon_0,\epsilon_{s}, \epsilon_{j_1},\epsilon_{j_2}\right)\right|\xrightarrow[t\to\infty]{} 0 \ ,
\end{align*}
hence $\eqref{cesaro_lemma}$ holds true.  Concerning the second term of right hand side of the inequality $\eqref{delta_n4_prob}$, we have
\begin{align*}
\frac{8}{n\beta^2}\sum_{t=1}^n\sum_{s=1}^t\sum_{j_1=1}^{\infty}\sum_{j_2=1}^{\infty} &\left| \overset{\textbf{.}}{\lambda}_{j_1,k}-\overset{\textbf{.}}{\lambda}_{j_1,k}^t\right|\left|\overset{\textbf{.}}{\lambda}_{j_2,k}-\overset{\textbf{.}}{\lambda}_{j_2,k}^s\right|  \left|\mathbb{E}\left[\epsilon_0\epsilon_{s-t}\right]\mathbb{E}\left[\epsilon_{-j_1}\epsilon_{s-t-j_2}\right]\right|\\
&\hspace{2cm}= \frac{8\sigma_{\epsilon}^2}{n\beta^2}\sum_{t=1}^n\sum_{j_1=1}^{\infty}\sum_{j_2=1}^{\infty}\left| \overset{\textbf{.}}{\lambda}_{j_1,k}-\overset{\textbf{.}}{\lambda}_{j_1,k}^t\right|\left|\overset{\textbf{.}}{\lambda}_{j_2,k}-\overset{\textbf{.}}{\lambda}_{j_2,k}^t\right| \left|\mathbb{E}\left[\epsilon_{-j_1}\epsilon_{-j_2}\right]\right|\\
&\hspace{2cm}= \frac{8\sigma_{\epsilon}^4}{n\beta^2}\sum_{t=1}^n\sum_{j_1=1}^{\infty}\left| \overset{\textbf{.}}{\lambda}_{j_1,k}-\overset{\textbf{.}}{\lambda}_{j_1,k}^t\right|^2 \\
&\hspace{2cm}=\frac{8\sigma_{\epsilon}^4}{n\beta^2}\sum_{t=1}^n\left\| \overset{\textbf{.}}{\lambda}_{k}-\overset{\textbf{.}}{\lambda}_{k}^t\right\|_{\ell^2}^2\\
&\hspace{2cm}\leq \frac{K}{\beta^2} \frac{1}n\sum_{t=1}^n\frac{1}{t^{1+2\min(d_0,0)}}\xrightarrow[n\to\infty]{}0
\end{align*}
where we have used the fact that the noise is not correlated, Lemma \ref{lemme_sur_les_ecarts_des_coef} with $r=2$ and Ces\`{a}ro's Lemma. This ends Step 2.
\paragraph{$\diamond$ Step 3: $(\Delta_{n,2}^k(\theta))_{n\geq 1}$ converges in probability to $0$}  \ \\

For all $\beta>0$, we have
\begin{align*}
\mathbb{P}\left( \left| \Delta_{n,2}^k(\theta)\right|\geq\beta\right)&
\leq \frac{2}{\beta\sqrt{n}}\sum_{t=1}^n\left\|\tilde{\epsilon}_t(\theta)-\epsilon_t(\theta)\right\|_{\mathbb{L}^2}\left\|\frac{\partial}{\partial\theta_k}\tilde{\epsilon}_t(\theta_0)\right\|_{\mathbb{L}^2}.
\end{align*}
First, using Lemma \ref{miss},  we have
\begin{align}\nonumber
\left\|\frac{\partial}{\partial\theta_k}\tilde{\epsilon}_t(\theta_0)\right\|^2_{\mathbb{L}^2}& =\mathbb{E}\left[ \left( \sum_{i=1}^{\infty}\overset{\textbf{.}}{\lambda}_{i,k}^t\left(\theta_0\right) \epsilon_{t-i}\right)^2 \right] \\ \nonumber
&=\sum_{i=1}^{\infty}\sum_{j=1}^{\infty}\overset{\textbf{.}}{\lambda}_{i,k}^t\left(\theta_0\right)\overset{\textbf{.}}{\lambda}_{j,k}^t\left(\theta_0\right) \mathbb{E}\left[\epsilon_{t-i}\epsilon_{t-j} \right] \\
& =\sigma_{\epsilon}^2\sum_{i=1}^{\infty}\left\lbrace \overset{\textbf{.}}{\lambda}_{i,k}^t\left(\theta_0\right)\right\rbrace ^2 \nonumber \\
&\le K . \label{dev-eps-tildeL2}
\end{align}

In view of \eqref{epsil-th}, \eqref{epsiltilde-th}  and \eqref{dev-eps-tildeL2}, we may write
\begin{align*}
\mathbb{P}\left( \left| \Delta_{n,2}^k(\theta)\right|\geq\beta\right)&\leq \frac{K}{\beta\sqrt{n}}\sum_{t=1}^n\left (\mathbb{E}\left[\left( \tilde{\epsilon}_t(\theta)-\epsilon_t(\theta)\right)^2\right] \right )^{1/2}\\
&\leq \frac{K}{\beta\sqrt{n}}\sum_{t=1}^n\left (\sum_{i\geq 0}\sum_{j\geq 0}\left(\lambda^t_i(\theta)-\lambda_i(\theta)\right)\left(\lambda^t_j(\theta)-\lambda_j(\theta)\right)\mathbb{E}\left[\epsilon_{t-i}\epsilon_{t-j}\right] \right )^{1/2}\\
&\leq \frac{\sigma_{\epsilon}K}{\beta\sqrt{n}}\sum_{t=1}^n\left (\sum_{i\geq 0}\left(\lambda^t_i(\theta)-\lambda_i(\theta)\right)^2\right )^{1/2} \\
&\leq \frac{\sigma_{\epsilon}K}{\beta\sqrt{n}}\sum_{t=1}^n\left\|\lambda(\theta)-\lambda^t(\theta)\right\|_{\ell^2}. 
\end{align*}
We use Lemma \ref{lemme_sur_les_ecarts_des_coef}, the fact that $d>\max(d_0,0)$ and the fractional version of Ces\`{a}ro's Lemma\footnote{Recall that the fractional version of Ces\`{a}ro's Lemma states that for $(h_t)_t$ a sequence of positive reals, $\kappa>0$ and $c\ge 0$ we have
$$\lim_{t\rightarrow\infty}h_tt^{1-\kappa}=\left|\kappa\right|c\Rightarrow\lim_{n\rightarrow\infty}\frac{1}{n^{\kappa}}\sum_{t=0}^n h_t=c.$$
} to obtain
\begin{align*}
\mathbb{P}\left( \left| \Delta_{n,2}^k(\theta)\right|\geq\beta\right)\leq \frac{\sigma_{\epsilon}K}{\beta}\, \frac{1}{\sqrt{n}}\sum_{t=1}^n\frac{1}{t^{1/2+(d-\max(d_0,0))}}\xrightarrow[n\to\infty]{} 0.
\end{align*}
This proves the expected convergence in probability.

\paragraph{$\diamond$ Step 4: convergence in probability of $(\Delta_{n,1}^k(\theta)+\Delta_{n,3}^k(\theta))_{n\ge 1}$ to 0} \ \\


Note that, for all $n\ge 1$, we have
\begin{align*}
\Delta_{n,1}^k(\theta)+\Delta_{n,3}^k(\theta)
&=\frac{2}{\sqrt{n}}\sum_{t=1}^n\Big \{ \left( \epsilon_t(\theta)-\tilde{\epsilon}_t(\theta)\right)-\left( \epsilon_t(\theta_0)-\tilde{\epsilon}_t(\theta_0)\right)\Big \} \frac{\partial}{\partial\theta_k}\tilde{\epsilon}_t(\theta_0).
\end{align*}
A Taylor expansion of the function $(\epsilon_t-\tilde{\epsilon}_t)(\cdot)$ around $\theta_0$ gives
\begin{align}\label{popo} \Big | ( \epsilon_t(\theta)-\tilde{\epsilon}_t(\theta))-( \epsilon_t(\theta_0)-\tilde{\epsilon}_t(\theta_0))\Big | & \le \left \| \frac{\partial ( \epsilon_t-\tilde{\epsilon}_t)}{\partial\theta}(\theta^\star)\right \|_{\mathbb R^{p+q+1}}  \| \theta-\theta_0\|_{\mathbb R^{p+q+1}}
\ ,
\end{align}
where $\theta^\star$ is between $\theta_0$ and $\theta$. Following the same method as in the previous step we obtain
\begin{align*}
\mathbb E \Big |\left( \epsilon_t(\theta)-\tilde{\epsilon}_t(\theta)\right)-\left( \epsilon_t(\theta_0)-\tilde{\epsilon}_t(\theta_0)\right)\Big |^2
& \le K\|\theta-\theta_0\|^2_{\mathbb R^{p+q+1}}\sum_{k=1}^{p+q+1}\mathbb E \left [\left | \frac{\partial ( \epsilon_t-\tilde{\epsilon}_t)}{\partial\theta_k}(\theta^\star)  \right |^2\right ] \nonumber\\
&\le K\|\theta-\theta_0\|^2_{\mathbb R^{p+q+1}}\sum_{k=1}^{p+q+1}\sigma_{\epsilon}^2\left\|(\overset{\textbf{.}}{\lambda_k}-\overset{\textbf{.}}{\lambda_k}^t)(\theta^\star)\right\|_{\ell^2}^2\nonumber.
\end{align*}
As in \cite{hallin}, it can be shown using Stirling's approximation and the fact that $d^\star >d_0$ that 
\begin{align*}
\left\|(\overset{\textbf{.}}{\lambda_k}-\overset{\textbf{.}}{\lambda_k}^t)(\theta^\star)\right\|_{\ell^2}\leq K\frac{1}{t^{1/2+(d^\star-d_0)-\zeta}}
\end{align*}
for any small enough $\zeta>0$. We then deduce that 
\begin{align}\label{ineq:diff}
\Big \|\left( \epsilon_t(\theta)-\tilde{\epsilon}_t(\theta)\right)-\left( \epsilon_t(\theta_0)-\tilde{\epsilon}_t(\theta_0)\right)\Big \|_{\mathbb{L}^2}&\le K\|\theta-\theta_0\|_{\mathbb R^{p+q+1}}\frac{1}{t^{1/2+(d^\star-d_0)-\zeta}}.
\end{align}
The expected convergence in probability follows from \eqref{dev-eps-tildeL2}, \eqref{ineq:diff} and the fractional version of Ces\`{a}ro's Lemma.

\end{proof}

We show in the following lemma the existence and invertibility of $J(\theta_0)$.
\begin{lemme}\label{matrixJ} Under Assumptions of Theorem $\ref{n.asymptotique}$, the matrix

$$J(\theta_0)=\lim_{n\rightarrow\infty}\left[ \frac{\partial^2}{\partial\theta_i\partial\theta_j}O_n(\theta_0)\right] $$
exists almost surely and is invertible.
\end{lemme}

\begin{proof} For all $1\le i,j \le p+q+1 $, we have
\begin{align*}
\frac{\partial^2}{\partial\theta_i\partial\theta_j}O_n(\theta_0)&=\frac{1}{n}\sum_{t=1}^n\frac{\partial^2}{\partial\theta_i\partial\theta_j}\epsilon_t^2(\theta_0)=\frac{2}{n}\sum_{t=1}^n\left\lbrace \frac{\partial}{\partial\theta_i}\epsilon_t(\theta_0)\frac{\partial}{\partial\theta_j}\epsilon_t(\theta_0)+\epsilon_t(\theta_0)\frac{\partial^2}{\partial\theta_i\partial\theta_j}\epsilon_t(\theta_0)\right\rbrace .
\end{align*}
Note that in view of \eqref{deriveesecepsil}, \eqref{deriveesecepsil1} and Remark \ref{rmq:important},  the first and second order derivatives of $\epsilon_t(\cdot)$ belong to $\mathbb{L}^2$. By using the ergodicity of $(\epsilon_t)_{t\in\mathbb{Z}}$ assumed in Assumption {\bf (A2)}, we deduce that
$$\frac{\partial^2}{\partial\theta_i\partial\theta_j}O_n(\theta_0)\underset{n\to\infty}{\overset{\text{a.s.}}{\longrightarrow}} 2\mathbb{E}\left[ \frac{\partial}{\partial\theta_i}\epsilon_t(\theta_0)\frac{\partial}{\partial\theta_j}\epsilon_t(\theta_0)\right]+2\mathbb{E}\left[ \epsilon_t(\theta_0)\frac{\partial^2}{\partial\theta_i\partial\theta_j}\epsilon_t(\theta_0)\right].  $$
By \eqref{epsil-th} and \eqref{deriveesecepsil}, $\epsilon_t$ and ${\partial\epsilon_t(\theta_0)}/{\partial\theta}$ are non correlated as well as $\epsilon_t$ and
${\partial^2\epsilon_t(\theta_0)}/{\partial\theta\partial\theta}$. Thus we have
\begin{equation}\label{matrixJexplicit}\frac{\partial^2}{\partial\theta_i\partial\theta_j}O_n(\theta_0)\underset{n\to\infty}{\overset{\text{a.s.}}{\longrightarrow}}J(\theta_0)(i,j):= 2\mathbb{E}\left[ \frac{\partial}{\partial\theta_i}\epsilon_t(\theta_0)\frac{\partial}{\partial\theta_j}\epsilon_t(\theta_0)\right].
\end{equation}
From \eqref{epsil-th} and \eqref{eq-miss2} we obtain that
\begin{align*}
\mathbb{E}\left[ \frac{\partial}{\partial\theta_i}\epsilon_t(\theta_0)\frac{\partial}{\partial\theta_j}\epsilon_t(\theta_0)\right]&=\mathbb{E}\left[ \left( \sum_{k_1\geq 1}\overset{\textbf{.}}{\lambda}_{k_1,i}\left(\theta_0\right)\epsilon_{t-k_1}\right) \left(\sum_{k_2\geq 1}\overset{\textbf{.}}{\lambda}_{k_2,j}\left(\theta_0\right)\epsilon_{t-k_2} \right) \right] \\
&=\sum_{k_1\geq 1}\sum_{k_2\geq 1}\overset{\textbf{.}}{\lambda}_{k_1,i}\left(\theta_0\right)\overset{\textbf{.}}{\lambda}_{k_2,j}\left(\theta_0\right)\mathbb{E}\left[ \epsilon_{t-k_1}\epsilon_{t-k_2}\right] \\
& \leq K\, \sigma_{\epsilon}^2\sum_{k_1\geq 1}\left( \frac{1}{k_1}\right) ^2
<\infty .
\end{align*}
Therefore $J(\theta_0)$ exists almost surely.

If the matrix $J(\theta_0)$ is not invertible, there exists some real constants $c_1,\dots,c_{p+q+1}$ not all equal to zero such that $\mathbf{c}^{'}J(\theta_0)\mathbf{c}=\sum_{i=1}^{p+q+1}\sum_{j=1}^{p+q+1}c_jJ(\theta_0)(j,i)c_i=0$, where $\mathbf{c}=(c_1,\dots,c_{p+q+1})^{'}$.
In view of \eqref{matrixJexplicit} we obtain that
\begin{align*}
\sum_{i=1}^{p+q+1}\sum_{j=1}^{p+q+1}\mathbb{E}\left[ \left( c_j\frac{\partial\epsilon_t(\theta_0)}{\partial\theta_j}\right) \left(  c_i\frac{\partial\epsilon_t(\theta_0)}{\partial\theta_i}\right) \right]=
\mathbb{E}\left[ \left( \sum_{k=1}^{p+q+1}c_k\frac{\partial\epsilon_t(\theta_0)}{\partial\theta_k}\right)^2 \right]&=0,
\end{align*}
which implies that
\begin{align}\label{hypInvJ}
\sum_{k=1}^{p+q+1}c_k\frac{\partial\epsilon_t(\theta_0)}{\partial\theta_k}&=0 \ \ \mathrm{a.s.}
\text{ or equivalenty }\ \ \mathbf{c}'\frac{\partial\epsilon_t(\theta_0)}{\partial\theta}=0 \ \ \mathrm{a.s.}
\end{align}
Differentiating the equation \eqref{FARIMA}, we obtain that
\begin{align*}
\mathbf{c}'\frac{\partial}{\partial\theta}\left\lbrace a_{\theta_0}(L)(1-L)^{d_0}\right\rbrace X_t=\mathbf{c}'\left\lbrace \frac{\partial}{\partial\theta}b_{\theta_0}(L)\right\rbrace \epsilon_t + b_{\theta_0}(L)\mathbf{c}'\frac{\partial}{\partial\theta}\epsilon_t(\theta_0).
\end{align*}
and by \eqref{hypInvJ} we may write that
\begin{align*}
\mathbf{c}^{'}\left(\frac{\partial}{\partial\theta}\left\lbrace a_{\theta_0}(L)(1-L)^{d_0}\right\rbrace X_t-\left\lbrace \frac{\partial}{\partial\theta}b_{\theta_0}(L)\right\rbrace \epsilon_t \right) =0 \ \ \mathrm{a.s.}
\end{align*}
It follows that  \eqref{FARIMA} can therefore be rewritten in the form:
\begin{align*}
\left( a_{\theta_0}(L)(1-L)^{d_0}+\mathbf{c}^{'}\frac{\partial}{\partial\theta}\left\lbrace a_{\theta_0}(L)(1-L)^{d_0}\right\rbrace\right)X_t = \left( b_{\theta_0}(L) +\mathbf{c}^{'}\frac{\partial}{\partial\theta}b_{\theta_0}(L)\right) \epsilon_t, \ \ \mathrm{a.s.}
\end{align*}
Under  Assumption {\bf (A1)}  the representation in \eqref{FARIMA} is unique (see \cite{Hosking}) so
\begin{align}\label{tt1}
 \mathbf{c}^{'}\frac{\partial}{\partial\theta}\left\lbrace a_{\theta_0}(L)(1-L)^{d_0}\right\rbrace & =0
\end{align}
and
\begin{align}
\mathbf{c}^{'}\frac{\partial}{\partial\theta}b_{\theta_0}(L) & =0. \label{tt2}
\end{align}
First, \eqref{tt2} implies that
$$ \sum_{k=p+1}^{p+q} c_k \frac{\partial}{\partial\theta_k}b_{\theta_0}(L)=\sum_{k=p+1}^{p+q} -c_k L^k =0$$
and thus $c_k = 0$ for $p+1\le k\le p+q$.

Similarly, \eqref{tt1} yields that
$$ \sum_{k=1}^{p} c_k \frac{\partial}{\partial\theta_k}a_{\theta_0}(L)(1-L)^{d_0}+c_{p+q+1} a_{\theta_0}(L)\frac{\partial (1-L)^d}{\partial d} (d_0) =0 \ .$$
Since $\partial (1-L)^d / \partial d= (1-L)^{d}\mathrm{ln}(1-L)$, it follows that
$$ -\sum_{k=1}^{p} c_k L^k +c_{p+q+1} \sum_{k\ge 0} e_k L^k =0 \ ,$$
where the sequence $(e_k)_{k\ge 1}$ is given by the coefficients of the power series $a_{\theta_0}(L)\mathrm{ln} (1-L)$. Since $e_0=0$ and $e_1=-1 $, we obtain that
\begin{align*}
c_1& =-c_{p+q+1}\\
c_k & =e_k c_{p+q+1} \ \text{for $k=2,\dots,p$} \\
0 & = e_k c_{p+q+1} \ \text{for $k\ge p+1$}.
\end{align*}
Since the polynomial $a_{\theta_0}$ is not the null polynomial, this implies that $c_{p+q+1}=0$ and then $c_k$ for $1\le k\le p$. Thus $\mathbf{c}=0$ which leads us to a contradiction. Hence $J(\theta_0)$ is invertible.

\end{proof}

%

\begin{lemme} \label{Conv_almost_sure_J_n_vers_J} For any $1\le i,j\le p+q+1 $ and under the assumptions of Theorem \ref{convergence}, we have
\begin{equation}
\frac{\partial^2}{\partial\theta_i\partial\theta_j}Q_n\left( \theta^*_{n,i,j}\right)-J(\theta_0)(i,j)=\mathrm{o}_{\mathbb{P}}(1),
\end{equation}
where $\theta^*_{n,i,j}$ is defined in \eqref{premierTaylor}.
\end{lemme}
\begin{proof}
For any $\theta\in\Theta_\delta$, let
\begin{align*}
{J}_n(\theta)&=\frac{\partial^2}{\partial\theta\partial\theta^{'}}Q_n\left( \theta\right)=\frac{2}{n}\sum_{t=1}^n\left\lbrace \frac{\partial}{\partial\theta}\tilde{\epsilon}_t\left(\theta\right)\right\rbrace \left\lbrace \frac{\partial}{\partial\theta^{'}}\tilde{\epsilon}_t\left(\theta\right)\right\rbrace+\frac{2}{n}\sum_{t=1}^n\tilde{\epsilon}_t(\theta)\frac{\partial^2}{\partial\theta\partial\theta^{'}}\tilde{\epsilon}_t(\theta)
\end{align*}
and
\begin{align*}
J_n^{*}(\theta)&=\frac{\partial^2}{\partial\theta\partial\theta^{'}}O_n\left( \theta\right)=\frac{2}{n}\sum_{t=1}^n\left\lbrace \frac{\partial}{\partial\theta}\epsilon_t\left(\theta\right)\right\rbrace \left\lbrace \frac{\partial}{\partial\theta^{'}}\epsilon_t\left(\theta\right)\right\rbrace+\frac{2}{n}\sum_{t=1}^n\epsilon_t(\theta)\frac{\partial^2}{\partial\theta\partial\theta^{'}}\epsilon_t(\theta).
\end{align*}
We have
\begin{align}\label{ineq_on_J}
\left|\frac{\partial^2}{\partial\theta_i\partial\theta_j}Q_n\left( \theta^*_{n,i,j}\right)-J(\theta_0)(i,j)\right|&\leq \left|{J}_{n}(\theta^*_{n,i,j})(i,j)-J_{n}^{*}(\theta^*_{n,i,j})(i,j)\right|\nonumber\\
&\quad+\left|J_{n}^{*}(\theta^*_{n,i,j})(i,j)-J_{n}^{*}(\theta_0)(i,j)\right|+\left|J_{n}^{*}(\theta_0)(i,j)-J(\theta_0)(i,j)\right|.
\end{align}
So it is enough to show that the three terms in the right hand side of \eqref{ineq_on_J} converge in probability to $0$ when $n$ tends to infinity.
Following the same arguments as the proof of Lemma \ref{matrixJ} and applying the ergodic theorem, we obtain that
\begin{equation*}
J_n^{*}(\theta_0)\underset{n\to\infty}{\overset{\text{a.s.}}{\longrightarrow}}2\mathbb{E}\left[ \frac{\partial}{\partial\theta}\epsilon_t(\theta_0)\frac{\partial}{\partial\theta^{'}}\epsilon_t(\theta_0)\right]=J(\theta_0).
\end{equation*}
Let us now show that the random variable $|J_{n}(\theta^*_{n,i,j})(i,j)-J_{n}^{*}(\theta^*_{n,i,j})(i,j)|$ converges in probability to 0. It can easily be seen that 

\begin{align*}
\frac{\partial\epsilon_t(\theta^*_{n,i,j})}{\partial\theta_i}\frac{\partial\epsilon_t(\theta^*_{n,i,j})}{\partial\theta_j}-&\frac{\partial\tilde{\epsilon}_t(\theta^*_{n,i,j})}{\partial\theta_i}\frac{\partial\tilde{\epsilon}_t(\theta^*_{n,i,j})}{\partial\theta_j}\\
&=\left(\frac{\partial\epsilon_t(\theta^*_{n,i,j})}{\partial\theta_i}-\frac{\partial\tilde{\epsilon}_t(\theta^*_{n,i,j})}{\partial\theta_i}\right)\frac{\partial\epsilon_t(\theta^*_{n,i,j})}{\partial\theta_j}+\frac{\partial\tilde{\epsilon}_t(\theta^*_{n,i,j})}{\partial\theta_i}\left(\frac{\partial\epsilon_t(\theta^*_{n,i,j})}{\partial\theta_j}-\frac{\partial\tilde{\epsilon}_t(\theta^*_{n,i,j})}{\partial\theta_j}\right).
\end{align*}
Hence, by the Cauchy-Schwarz inequality and Remark \ref{rmq:important} one has
\begin{align*}
\mathbb{E}\Big|\frac{\partial\epsilon_t(\theta^*_{n,i,j})}{\partial\theta_i}\frac{\partial\epsilon_t(\theta^*_{n,i,j})}{\partial\theta_j}&-\frac{\partial\tilde{\epsilon}_t(\theta^*_{n,i,j})}{\partial\theta_i}\frac{\partial\tilde{\epsilon}_t(\theta^*_{n,i,j})}{\partial\theta_j}\Big|\\
&\leq \left(\mathbb{E}\left(\frac{\partial\epsilon_t(\theta^*_{n,i,j})}{\partial\theta_i}-\frac{\partial\tilde{\epsilon}_t(\theta^*_{n,i,j})}{\partial\theta_i}\right)^2\mathbb{E}\left(\frac{\partial\epsilon_t(\theta^*_{n,i,j})}{\partial\theta_j}\right)^2\right)^{1/2}\\
&\qquad+\left(\mathbb{E}\left(\frac{\partial\epsilon_t(\theta^*_{n,i,j})}{\partial\theta_j}-\frac{\partial\tilde{\epsilon}_t(\theta^*_{n,i,j})}{\partial\theta_j}\right)^2\mathbb{E}\left(\frac{\partial\epsilon_t(\theta^*_{n,i,j})}{\partial\theta_i}\right)^2\right)^{1/2}\\
&\leq \left(\sup_{\theta\in\Theta_\delta}\mathbb{E}\left(\frac{\partial\epsilon_t(\theta)}{\partial\theta_i}-\frac{\partial\tilde{\epsilon}_t(\theta)}{\partial\theta_i}\right)^2\sup_{\theta\in\Theta_\delta}\mathbb{E}\left(\frac{\partial\epsilon_t(\theta)}{\partial\theta_j}\right)^2\right)^{1/2}\\
&\qquad+\left(\sup_{\theta\in\Theta_\delta}\mathbb{E}\left(\frac{\partial\epsilon_t(\theta)}{\partial\theta_j}-\frac{\partial\tilde{\epsilon}_t(\theta)}{\partial\theta_j}\right)^2\sup_{\theta\in\Theta_\delta}\mathbb{E}\left(\frac{\partial\epsilon_t(\theta)}{\partial\theta_i}\right)^2\right)^{1/2}\\
&\leq \sigma_\epsilon^2\left(\sup_{\theta\in\Theta_\delta}\left\|\overset{\textbf{.}}{\lambda}_{i}(\theta)-\overset{\textbf{.}}{\lambda}_{i}^t(\theta)\right\|_{\ell^2}^2\sup_{\theta\in\Theta_\delta}\sum_{k\geq 1}\left(\overset{\textbf{.}}{\lambda}_{k,j}(\theta)\right)^2\right)^{1/2}\\
&\qquad+\sigma_\epsilon^2\left(\sup_{\theta\in\Theta_\delta}\left\|\overset{\textbf{.}}{\lambda}_{j}(\theta)-\overset{\textbf{.}}{\lambda}_{j}^t(\theta)\right\|_{\ell^2}^2\sup_{\theta\in\Theta_\delta}\sum_{k\geq 1}\left(\overset{\textbf{.}}{\lambda}_{k,i}^t(\theta)\right)^2\right)^{1/2}\\
&\leq K\frac{1}{t^{1/2+(d_1-d_0)-\zeta}}\xrightarrow[t\to \infty]{} \, 0.
\end{align*}
Similar calculation can be done to obtain 
\begin{align*}
\mathbb{E}\left|\epsilon_t(\theta^*_{n,i,j})\frac{\partial^2\epsilon_t(\theta^*_{n,i,j})}{\partial\theta_i\partial\theta_j}-\tilde{\epsilon}_t(\theta^*_{n,i,j})\frac{\partial^2\tilde{\epsilon}_t(\theta^*_{n,i,j})}{\partial\theta_i\partial\theta_j}\right|\xrightarrow[t\to \infty]{} \, 0.
\end{align*}
It follows then using Ces\`{a}ro's Lemma that
\begin{align*}
\mathbb{E}\left|J_{n}(\theta^*_{n,i,j})(i,j)-J_{n}^{*}(\theta^*_{n,i,j})(i,j)\right|&\leq \frac{2}{n}\sum_{t=1}^n\mathbb{E}\left|\frac{\partial\epsilon_t(\theta^*_{n,i,j})}{\partial\theta_i}\frac{\partial\epsilon_t(\theta^*_{n,i,j})}{\partial\theta_j}-\frac{\partial\tilde{\epsilon}_t(\theta^*_{n,i,j})}{\partial\theta_i}\frac{\partial\tilde{\epsilon}_t(\theta^*_{n,i,j})}{\partial\theta_j}\right|\\
&\qquad+\frac{2}{n}\sum_{t=1}^n\mathbb{E}\left|\epsilon_t(\theta^*_{n,i,j})\frac{\partial^2\epsilon_t(\theta^*_{n,i,j})}{\partial\theta_i\partial\theta_j}-\tilde{\epsilon}_t(\theta^*_{n,i,j})\frac{\partial^2\tilde{\epsilon}_t(\theta^*_{n,i,j})}{\partial\theta_i\partial\theta_j}\right|\\
&\leq \frac{K}{n}\sum_{t=1}^n\frac{1}{t^{1/2+(d_1-d_0)-\zeta}}\xrightarrow[n\to \infty]{} \, 0,
\end{align*}
which entails the expected convergence in probability to 0 of $|J_{n}(\theta^*_{n,i,j})(i,j)-J_{n}^{*}(\theta^*_{n,i,j})(i,j)|$.

By a Taylor expansion of $J_n^{*}(\cdot)(i,j)$ around $\theta_0$, there exists $\theta_{n,i,j}^{**}$ between $\theta^*_{n,i,j}$ and $\theta_0$ such that
\begin{align}\label{J:3}
\left|J_n^{*}(\theta^*_{n,i,j})(i,j)-J_n^{*}(\theta_0)(i,j)\right|&=\left|\frac{\partial}{\partial\theta}J_n^{*}(\theta_{n,i,j}^{**})(i,j)\cdot(\theta^*_{n,i,j}-\theta_0)\right| \nonumber\\
& \leq \left\|\frac{\partial}{\partial\theta}J_n^{*}(\theta_{n,i,j}^{**})(i,j)\right\|\left\|\theta^*_{n,i,j}-\theta_0\right\|\nonumber\\
&\leq \frac{2}{n}\sum_{t=1}^n\left\|\left.\frac{\partial}{\partial\theta}\left\lbrace\frac{\partial}{\partial\theta_i}\epsilon_t(\theta)\frac{\partial}{\partial\theta_j}\epsilon_t(\theta) \right\rbrace\right|_{\theta=\theta_{n,i,j}^{**}}\right\|\left\|\theta^*_{n,i,j}-\theta_0\right\|\nonumber\\
&\qquad+\frac{2}{n}\sum_{t=1}^n\left\|\left.\frac{\partial}{\partial\theta}\left\lbrace\epsilon_t(\theta)\frac{\partial^2}{\partial\theta_i\partial\theta_j}\epsilon_t(\theta) \right\rbrace\right|_{\theta=\theta_{n,i,j}^{**}}\right\|\left\|\theta^*_{n,i,j}-\theta_0\right\| . 
\end{align}
Since $d_1-d_0>1/2$, it can easily be shown as before that 
\begin{align}\label{J:2}
\mathbb{E}\left\|\left.\frac{\partial}{\partial\theta}\left\lbrace\frac{\partial}{\partial\theta_i}\epsilon_t(\theta)\frac{\partial}{\partial\theta_j}\epsilon_t(\theta) \right\rbrace\right|_{\theta=\theta_{n,i,j}^{**}}\right\|<\infty
\end{align}
and 
\begin{align}\label{J:1}
\mathbb{E}\left\|\left.\frac{\partial}{\partial\theta}\left\lbrace\epsilon_t(\theta)\frac{\partial^2}{\partial\theta_i\partial\theta_j}\epsilon_t(\theta) \right\rbrace\right|_{\theta=\theta_{n,i,j}^{**}}\right\|<\infty.
\end{align}
We use \eqref{J:3}, \eqref{J:2}, \eqref{J:1}, the ergodic theorem and Theorem \ref{convergence} to deduce the convergence in probability of $|J_n^{*}(\theta^*_{n,i,j})(i,j)-J_n^{*}(\theta_0)(i,j)|$ to 0.

The proof of the lemma is then complete.
\end{proof}
The following lemma states the existence of the matrix $I(\theta_0)$.
\begin{lemme} \label{exit_I} Under the  assumptions of Theorem \ref{n.asymptotique}, the matrix
$$I(\theta_0)=\lim_{n\rightarrow\infty}Var\left\lbrace \sqrt{n}\frac{\partial}{\partial\theta}O_n(\theta_0)\right\rbrace$$
exists.
\end{lemme}

\begin{proof} By the stationarity of $(H_t(\theta_0))_{t\in\mathbb{Z}}$ (remind that this process is defined in \eqref{Ht_process}), we have
\begin{align*}
Var\left\lbrace \sqrt{n}\frac{\partial}{\partial\theta}O_n(\theta_0)\right\rbrace &=Var\left\lbrace \frac{1}{\sqrt{n}}\sum_{t=1}^nH_t(\theta_0)\right\rbrace \\
& =\frac{1}{n}\sum_{t=1}^n\sum_{s=1}^nCov\left\lbrace H_t(\theta_0),H_s(\theta_0)\right\rbrace \\
&=\frac{1}{n}\sum_{h=-n+1}^{n-1}\left( n-|h|\right)Cov\left\lbrace H_t(\theta_0),H_{t-h}(\theta_0)\right\rbrace .
\end{align*}
By the dominated convergence theorem, the matrix $I(\theta_0)$ exists and is given by
$$I(\theta_0)=\sum_{h=-\infty}^{\infty}Cov\left\lbrace H_t(\theta_0),H_{t-h}(\theta_0)\right\rbrace $$
whenever
\begin{equation}\label{conv_dominee}
\sum_{h=-\infty}^{\infty}\|Cov\left\lbrace H_t(\theta_0),H_{t-h}(\theta_0)\right\rbrace \|<\infty.
\end{equation}
For  $s\in\mathbb{Z}$ and $1\le k\ \le p+q+1$, we denote $H_{s,k}(\theta_0)=2\epsilon_s(\theta_0)\frac{\partial}{\partial\theta_k}\epsilon_s(\theta_0)$ the $k-$th
entry of $H_{s}(\theta_0)$.
In view of \eqref{deriveesecepsil} we have 
\begin{align*}
\left|Cov\left\lbrace H_{t,i}(\theta_0),H_{t-h,j}(\theta_0)\right\rbrace \right|&=4\left|Cov\left( \sum_{k_1\geq 1}\overset{\textbf{.}}{\lambda}_{k_1,i}\left(\theta_0\right)\epsilon_t\epsilon_{t-k_1},\sum_{k_2\geq 1}\overset{\textbf{.}}{\lambda}_{k_2,j}\left(\theta_0\right)\epsilon_{t-h}\epsilon_{t-h-k_2}\right) \right|
\\&\leq 4\sum_{k_1\geq 1}\sum_{k_2\geq 1}\left|\overset{\textbf{.}}{\lambda}_{k_1,i}\left(\theta_0\right)\right|\left|\overset{\textbf{.}}{\lambda}_{k_2,j}\left(\theta_0\right)\right|\left|\mathbb{E}\left[ \epsilon_t\epsilon_{t-k_1}\epsilon_{t-h}\epsilon_{t-h-k_2}\right]\right| \\
& \le  \sum_{k_1\geq 1}\sum_{k_2\geq 1}\frac{K}{k_1k_2}\left|\mathbb{E}\left[ \epsilon_t\epsilon_{t-k_1}\epsilon_{t-h}\epsilon_{t-h-k_2}\right]\right|
\end{align*}
where we have used Lemma \ref{miss2}. It follows that
\begin{align*}
\sum_{h=-\infty}^{\infty}\left|Cov\left\lbrace H_{t,i}(\theta_0),H_{t-h,j}(\theta_0)\right\rbrace \right|&\leq \sum_{h\in\mathbb Z\setminus\{0\}}\sum_{k_1\geq 1}\sum_{k_2\geq 1}\frac{K}{k_1k_2}\left|\mathrm{cum}\left( \epsilon_t,\epsilon_{t-k_1},\epsilon_{t-h},\epsilon_{t-h-k_2}\right)\right|\\ & \hspace{2cm}+\sum_{k_1\geq 1}\sum_{k_2\geq 1}\frac{K}{k_1k_2}\left|\mathbb{E}\left[ \epsilon_t\epsilon_{t-k_1}\epsilon_{t}\epsilon_{t-k_2}\right]\right| \ .
\end{align*}
Thanks to the stationarity of $(\epsilon_t)_{t\in\mathbb{Z}}$ and Assumption {\bf (A4')} with $\tau=4$ we deduce that
\begin{align*}
\sum_{h=-\infty}^{\infty}\left|Cov\left\lbrace H_{t,i}(\theta_0),H_{t-h,j}(\theta_0)\right\rbrace \right|&\leq
 \sum_{h\in\mathbb Z\setminus\{0\}}\sum_{k_1\geq 1}\sum_{k_2\geq 1}\frac{K}{k_1k_2}\left|\mathrm{cum}\left( \epsilon_0,\epsilon_{-k_1},\epsilon_{-h},\epsilon_{-h-k_2}\right)\right|\\ & \hspace{2cm}+\sum_{k_1\geq 1}\sum_{k_2\geq 1}\frac{K}{k_1k_2}\left|\mathbb{E}\left[ \epsilon_0\epsilon_{-k_1}\epsilon_{0}\epsilon_{-k_2}\right]\right| \\
& \le  K \sum_{h,k,l\in \mathbb Z}\left|\mathrm{cum}\left( \epsilon_0,\epsilon_{k},\epsilon_{h},\epsilon_{l}\right)\right|\\
& \hspace{2cm}+\sum_{k_1\geq 1}\sum_{k_2\geq 1}\frac{K}{k_1k_2}\Big ( \left|\mathrm{cum}\left( \epsilon_0,\epsilon_{-k_1},\epsilon_{0},\epsilon_{-k_2}\right)\right| \\ 
& \hspace{6cm}+\sigma_{\epsilon}^2\left|\mathbb{E}\left[ \epsilon_{-k_1}\epsilon_{-k_2}\right]\right|
\Big ) \\
& \le  K \sum_{h,k,l\in \mathbb Z}\left|\mathrm{cum}\left( \epsilon_0,\epsilon_{k},\epsilon_{h},\epsilon_{l}\right)\right|+ K \sigma_{\epsilon}^4\sum_{k_1\geq 1}\left(\frac{1}{k_1}\right)^2 \le K
\end{align*}
and we obtain the expected result.
\end{proof}
\begin{lemme}\label{normalite_score} Under Assumptions of Theorem~\ref{n.asymptotique}, the random vector $\sqrt{n}({\partial}/{\partial\theta})O_n(\theta_0)$ has a limiting normal distribution with mean $0$ and covariance matrix $I(\theta_0)$.
\end{lemme}

\begin{proof} 
Observe that for any $t\in\mathbb Z$
\begin{align}\label{zeromean}
& \mathbb{E}\left[\epsilon_t\frac{\partial}{\partial\theta}\epsilon_t(\theta_0)\right]=0
\end{align}
because ${\partial\epsilon_t(\theta_0)}/{\partial\theta}$ belongs to the Hilbert space  $\mathbf{H}_{\epsilon}(t-1)$, linearly generated by the family $(\epsilon_s)_{s\le t-1}$. Therefore we have
\begin{align*}
\lim_{n\rightarrow\infty}\mathbb{E}\left[ \sqrt{n}\frac{\partial}{\partial\theta}O_n(\theta_0)\right]&=
\lim_{n\rightarrow\infty}\frac{2}{\sqrt{n}}\sum_{t=1}^{n}\mathbb{E}\left[\epsilon_t\frac{\partial}{\partial\theta}\epsilon_t(\theta_0)\right]
=0.
\end{align*}
For $i\ge 1$, we denote by $ {\Lambda}_{i}\left(\theta_0\right)=(\overset{\textbf{.}}{\lambda}_{i,1}\left(\theta_0\right),\dots,
\overset{\textbf{.}}{\lambda}_{i,p+q+1}\left(\theta_0\right))'$ and we introduce for $r\ge 1$
$$H_{t,r}(\theta_0)=2\sum_{i= 1}^r{\Lambda}_{i}\left(\theta_0\right) \epsilon_t\epsilon_{t-i}\ \text{ and } \
G_{t,r}(\theta_0)=2\sum_{i\geq r+1}{\Lambda}_{k}\left(\theta_0\right) \epsilon_t\epsilon_{t-i}.$$
From \eqref{deriveesecepsil} we have
$$\sqrt{n}\frac{\partial}{\partial\theta}O_n(\theta_0)=\frac{1}{\sqrt{n}}\sum_{t=1}^{n}H_{t,r}(\theta_0)+\frac{1}{\sqrt{n}}\sum_{t=1}^{n}G_{t,r}(\theta_0) .$$
Since $H_{t,r}(\theta_0)$ is a function of finite number of values of the process $(\epsilon_t)_{t\in\mathbb{Z}}$, the stationary process $(H_{t,r}(\theta_0))_{t\in\mathbb{Z}}$ satisfies a mixing property  (see Theorem 14.1 in \cite{davidson1994}, p. 210) of the form {\bf (A4)}. The central limit theorem for strongly mixing processes (see \cite{herr})  implies that  $({1}/{\sqrt{n}})\sum_{t=1}^{n}H_{t,r}(\theta_0)$ has a limiting $\mathcal{N}(0,I_r(\theta_0))$ distribution with
$$I_r(\theta_0)=\lim_{n\rightarrow\infty}Var\left( \frac{1}{\sqrt{n}}\sum_{t=1}^{n}H_{t,r}(\theta_0)\right) .$$
Since $ \frac{1}{\sqrt{n}}\sum_{t=1}^{n}H_{t,r}(\theta_0)$ and $ \frac{1}{\sqrt{n}}\sum_{t=1}^{n}H_{t}(\theta_0)$ have zero expectation, we shall have
$$\lim_{r\rightarrow\infty}Var\left( \frac{1}{\sqrt{n}}\sum_{t=1}^{n}H_{t,r}(\theta_0)\right)=Var\left( \frac{1}{\sqrt{n}}\sum_{t=1}^{n}H_{t}(\theta_0)\right)=Var\left\lbrace \sqrt{n}\frac{\partial}{\partial\theta}O_n(\theta_0)\right\rbrace ,$$
 as soon as
\begin{equation}\label{conv_L2}
\lim_{r\rightarrow\infty}\mathbb{E}\left[ \left\|\frac{1}{\sqrt{n}}\sum_{t=1}^{n}H_{t}(\theta_0)-\frac{1}{\sqrt{n}}\sum_{t=1}^{n}H_{t,r}(\theta_0)\right\|^2\right] =0 .
\end{equation}
As a consequence we will have $\lim_{r\rightarrow\infty}I_r(\theta_0)=I(\theta_0)$.  The limit in $\eqref{conv_L2}$ is obtained as follows:
\begin{align*}
\mathbb{E}&\left[ \left\|\frac{1}{\sqrt{n}}\sum_{t=1}^{n}H_{t}(\theta_0)-\frac{1}{\sqrt{n}}\sum_{t=1}^{n}H_{t,r}(\theta_0)\right\|_{\mathbb R^{p+q+1}}^2\right] =\mathbb{E}\left[ \left\|\frac{1}{\sqrt{n}}\sum_{t=1}^{n}G_{t,r}(\theta_0)\right\|_{\mathbb R^{p+q+1}}^2\right] \\
&\hspace{4cm}\leq \frac{4}{n}\sum_{l=1}^{p+q+1}\mathbb{E}\left[ \left( \sum_{t=1}^{n}\sum_{k\geq r+1}\overset{\textbf{.}}{\lambda}_{k,l}\left(\theta_0\right) \epsilon_{t-k}\epsilon_t\right) ^2\right]\\
&\hspace{4cm}\leq \frac{4}{n}\sum_{l=1}^{p+q+1}\sum_{t=1}^{n}\sum_{s=1}^{n}\sum_{k\geq r+1}\sum_{j\geq r+1}\left|\overset{\textbf{.}}{\lambda}_{k,l}\left(\theta_0\right)\right|\left|\overset{\textbf{.}}{\lambda}_{j,l}\left(\theta_0\right)\right|\left|\mathbb{E}\left[ \epsilon_{t-k}\epsilon_t\epsilon_{s-j}\epsilon_{s}\right] \right|,
\end{align*}
We use successively the stationarity of $(\epsilon_t)_{t\in\mathbb{Z}}$, Lemma \ref{miss2} and  Assumption {\bf (A4')} with $\tau=4$ in order to obtain that
\begin{align*}
\mathbb{E}&\left[ \left\|\frac{1}{\sqrt{n}}\sum_{t=1}^{n}H_{t}(\theta_0)-\frac{1}{\sqrt{n}}\sum_{t=1}^{n}H_{t,r}(\theta_0)\right\|_{\mathbb R^{p+q+1}}^2\right] \\
&\hspace{2cm}\leq \frac{4}{n}\sum_{l=1}^{p+q+1}\sum_{h=1-n}^{n-1}\sum_{k\geq r+1}\sum_{j\geq r+1}\left|\overset{\textbf{.}}{\lambda}_{k,l}\left(\theta_0\right)\right|\left|\overset{\textbf{.}}{\lambda}_{j,l}\left(\theta_0\right)\right|\left(n-\left|h\right| \right) \left|\mathbb{E}\left[ \epsilon_{t-k}\epsilon_t\epsilon_{t-h-j}\epsilon_{t-h}\right] \right|\\
&\hspace{2cm}\leq 4\sum_{l=1}^{p+q+1}\sum_{h=-\infty}^{\infty}\sum_{k\geq r+1}\sum_{j\geq r+1}\left|\overset{\textbf{.}}{\lambda}_{k,l}\left(\theta_0\right)\right|\left|\overset{\textbf{.}}{\lambda}_{j,l}\left(\theta_0\right)\right|\left|\mathbb{E}\left[ \epsilon_{t-k}\epsilon_t\epsilon_{t-h-j}\epsilon_{t-h}\right] \right|\\
&\hspace{2cm}\leq \frac{K}{(r+1)^2}\sum_{h\neq 0}\sum_{k\geq r+1}\sum_{j=-\infty}^{\infty}\left|\mathrm{cum}\left(\epsilon_0,\epsilon_{-k},\epsilon_{-j},\epsilon_{-h} \right) \right|\\
&\hspace{4cm}+\frac{K}{(r+1)^2}\sum_{k\geq r+1}\sum_{j\geq r+1}\left|\mathrm{cum}\left(\epsilon_0,\epsilon_{-k},\epsilon_{-j},\epsilon_{0} \right) \right| +K\sigma_{\epsilon}^4\sum_{k\geq r+1}\left(\frac{1}{k}\right)^2
\end{align*}
and we obtain the convergence stated in \eqref{conv_L2} when $r\rightarrow\infty$.

Using Theorem 7.7.1 and Corollary 7.7.1 of Anderson (see \cite{anderson} pages 425-426), the lemma is proved once we have, uniformly in $n$,
$$Var\left( \frac{1}{\sqrt{n}}\sum_{t=1}^{n}G_{t,r}(\theta_0)\right)\xrightarrow[r\to\infty]{} 0 \ .$$

Arguing as before we may write
\begin{align*}
\left[ Var\left( \frac{1}{\sqrt{n}}\sum_{t=1}^{n}G_{t,r}(\theta_0) \right)\right] _{ij}&=\left[ Var\left( \frac{2}{\sqrt{n}}\sum_{t=1}^{n}\sum_{k\geq r+1}\Lambda_k(\theta_0)\epsilon_{t-k}\epsilon_t \right) \right] _{ij}\\
&=\frac{4}{n}\sum_{t=1}^{n}\sum_{s=1}^{n}\sum_{k_1\geq r+1}\sum_{k_2\geq r+1}\overset{\textbf{.}}{\lambda}_{k_1,i}\left(\theta_0\right)\overset{\textbf{.}}{\lambda}_{k_2,j}\left(\theta_0\right)\mathbb{E}\left[ \epsilon_{t-k_1}\epsilon_t\epsilon_{s-k_2}\epsilon_s\right] \\
&\leq 4\sum_{h=-\infty}^{\infty}\sum_{k_1,k_2\geq r+1} 
\left|\overset{\textbf{.}}{\lambda}_{k_1,i}\left(\theta_0\right)\overset{\textbf{.}}{\lambda}_{k_2,j}\left(\theta_0\right)\right|\left|\mathbb{E}\left[ \epsilon_{t-k_1}\epsilon_t\epsilon_{t-h-k_2}\epsilon_{t-h}\right]\right|.\\
\end{align*}
and we obtain that
\begin{align}\label{koko}
& \sup_n Var\left( \frac{1}{\sqrt{n}}\sum_{t=1}^{n}G_{t,r}(\theta_0) \right) \xrightarrow[r\to\infty]{} 0 ,
\end{align}
 which completes the proof.
\end{proof}
No we can end this quite long proof of the asymptotic normality result.
\paragraph*{Proof of Theorem \ref{n.asymptotique}} \ \\

In view of Lemma~\ref{ConProQversO}, the equation \eqref{Taylormodifie} can be rewritten in the form:
\begin{align*}
&\mathrm{o}_{\mathbb{P}}(1)= \sqrt{n}\frac{\partial}{\partial\theta}O_n(\theta_0)+\left[\frac{\partial^2}{\partial\theta_i\partial\theta_j}Q_n\left( \theta^*_{n,i,j}\right) \right] \sqrt{n}\left( \hat{\theta}_n-\theta_0\right).
\end{align*}
From Lemma \ref{normalite_score} $ [({\partial^2}/{\partial\theta_i\partial\theta_j})Q_n( \theta^*_{n,i,j})]\sqrt{n}( \widehat{\theta}_n-\theta_0) $ converges in distribution to $\mathcal{N}(0,I(\theta_0))$.
Using Lemma \ref{Conv_almost_sure_J_n_vers_J} and Slutsky's theorem we deduce that
$$\left(\left[\frac{\partial^2}{\partial\theta_i\partial\theta_j}Q_n\left( \theta^*_{n,i,j}\right) \right],\left[\frac{\partial^2}{\partial\theta_i\partial\theta_j}Q_n\left( \theta^*_{n,i,j}\right) \right]\sqrt{n}( \hat{\theta}_n-\theta_0)\right)$$ converges in distribution to $(J(\theta_0),Z)$ with $\mathbb{P}_{Z}=\mathcal{N}(0,I)$.
Consider now the function $h:\mathbb{R}^{(p+q+1)\times (p+q+1)}\times \mathbb{R}^{p+q+1}\rightarrow\mathbb{R}^{p+q+1}$ that maps $(A,X)$ to $A^{-1}X$. If $D_h$ denotes the set of discontinuity points of $h$, we have  $\mathbb{P}((J(\theta_0),Z)\in D_h)=0$. By the continuous mapping theorem $$h\left(\left[({\partial^2}/{\partial\theta_i\partial\theta_j})Q_n( \theta^*_{n,i,j}) \right],\left[({\partial^2}/{\partial\theta_i\partial\theta_j})Q_n( \theta^*_{n,i,j}) \right]\sqrt{n}( \hat{\theta}_n-\theta_0)\right)$$ converges in distribution to $h(J(\theta_0),Z)$ and thus $\sqrt{n}( \hat{\theta}_n-\theta_0)$ has a limiting normal distribution with mean $0$ and covariance matrix $J^{-1}(\theta_0)I(\theta_0)J^{-1}(\theta_0)$.
The proof of Theorem~\ref{n.asymptotique} is then completed.

\subsection{Proof of the convergence of the variance matrix estimator}\label{sectI}
We show in this section the convergence in probability of $\hat{\Omega}:=\hat{J}_n^{-1}\hat{I}_n^{SP}\hat{J}_n^{-1}$ to $\Omega$, which is an adaptation of the arguments used in \cite{BMCF2012}.\\
Using the same approach as that followed in Lemma \ref{Conv_almost_sure_J_n_vers_J}, we show that $\hat{J}_n$ converges in probability to $J$. We give below the proof of the convergence in probability of the estimator $\hat{I}_n^{SP}$, obtained using the approach of the spectral density, to $I$.\\
We recall that the matrix norm used is given by $\left\|A\right\|=\sup_{\left\|x\right\|\leq 1}\left\|Ax\right\|=\rho^{1/2}( A^{'}A)$, when $A$ is a $\mathbb{R}^{k_1\times k_2}$ matrix, $\|x\|^2=x^{'}x$ is the Euclidean norm of the vector $x\in\mathbb{R}^{k_2}$, and $\rho(\cdot)$ denotes the spectral radius. This norm satisfies
\begin{equation}\label{ineg_norme_matricielle}
\left\|A\right\|^2\leq \sum_{i=1}^{k_1}\sum_{j=1}^{k_2}a_{i,j}^2,
\end{equation}
with $a_{i,j}$ the entries of $A\in\mathbb R^{k_1\times k_2}$. The choice of the norm is crucial for the following results to hold (with e.g. the Euclidean norm, this result is not valid).

 We denote
$$\Sigma_{H,\underline{H_r}}=\mathbb{E}H_t\underline{H}_{r,t}^{'}, \ \ \Sigma_{H}=\mathbb{E}H_tH_t^{'}, \ \ \Sigma_{\underline{H_r}}=\mathbb{E}\underline{H}_{r,t}\underline{H}_{r,t}^{'} $$
where $H_t:=H_t(\theta_0)$ is definied in \eqref{Ht_process} and $\underline{{H}}_{r,t}=({H}_{t-1}^{'},\dots,{H}_{t-r}^{'}) ^{'}$.
For any $n\ge 1$, we have
\begin{align*}
\hat{I}_n^{\mathrm{SP}}&=\hat{\Phi}_r^{-1}(1)\hat{\Sigma}_{\hat{u}_r}\hat{\Phi}_r^{'-1}(1)
\\
&=\left( \hat{\Phi}_r^{-1}(1)-\Phi^{-1}(1)\right) \hat{\Sigma}_{\hat{u}_r}\hat{\Phi}_r^{'-1}(1)+\Phi^{-1}(1)\left( \hat{\Sigma}_{\hat{u}_r}-\Sigma_u\right) \hat{\Phi}_r^{'-1}(1)\\
&\hspace{2cm}+\Phi^{-1}(1)\Sigma_u\left( \hat{\Phi}_r^{'-1}(1)-\Phi^{'-1}(1)\right) +\Phi^{-1}(1)\Sigma_u\Phi^{'-1}(1).
\end{align*}
We then obtain
\begin{align}\label{converg_Isp_vers_I}
\left\|\hat{I}_n^{\mathrm{SP}}-I(\theta_0)\right\|&\leq \left\| \hat{\Phi}_r^{-1}(1)-\Phi^{-1}(1)\right\|\left\| \hat{\Sigma}_{\hat{u}_r}\right\|\left\|\hat{\Phi}_r^{'-1}(1)\right\|+\left\|\Phi^{-1}(1)\right\|\left\| \hat{\Sigma}_{\hat{u}_r}-\Sigma_u\right\|\left\| \hat{\Phi}_r^{'-1}(1)\right\| \nonumber\\
&\hspace{2cm}+\left\|\Phi^{-1}(1)\right\|\left\|\Sigma_u\right\|\left\| \hat{\Phi}_r^{'-1}(1)-\Phi^{'-1}(1)\right\|\nonumber\\
&\leq  \left\| \hat{\Phi}_r^{-1}(1)-\Phi^{-1}(1)\right\|\left( \left\| \hat{\Sigma}_{\hat{u}_r}\right\|\left\|\hat{\Phi}_r^{'-1}(1)\right\|+\left\|\Phi^{-1}(1)\right\|\left\|\Sigma_u\right\|\right)\nonumber\\
&\hspace{2cm} +\left\| \hat{\Sigma}_{\hat{u}_r}-\Sigma_u\right\|\left\| \hat{\Phi}_r^{'-1}(1)\right\|\left\|\Phi^{-1}(1)\right\|\nonumber\\
&\leq  \left\| \hat{\Phi}_r^{-1}(1)\right\|\left\|\Phi(1)- \hat{\Phi}_r(1)\right\|\left\|\Phi^{-1}(1)\right\|\left( \left\| \hat{\Sigma}_{\hat{u}_r}\right\|\left\|\hat{\Phi}_r^{'-1}(1)\right\|+\left\|\Phi^{-1}(1)\right\|\left\|\Sigma_u\right\|\right)\nonumber\\
&\hspace{2cm} +\left\| \hat{\Sigma}_{\hat{u}_r}-\Sigma_u\right\|\left\| \hat{\Phi}_r^{'-1}(1)\right\|\left\|\Phi^{-1}(1)\right\|.
\end{align}
In view of $\eqref{converg_Isp_vers_I}$, to prove the convergence in probability of $\hat{I}_n^{\mathrm{SP}}$ to $I(\theta_0)$, it suffices to show that $\hat{\Phi}_r(1)\rightarrow\Phi(1)$ and $\hat{\Sigma}_{\hat{u}_r}\rightarrow\Sigma_u$ in probability.
Let the $r\times 1$ vector $\mathbb{1}_r=(1,\dots,1)^{'}$ and the $r(p+q+1)\times (p+q+1)$ matrix $\mathbf{E}_r={I}_{p+q+1}\otimes \mathbb{1}_r$, where $\otimes$ denotes the matrix Kronecker product and ${I}_m$ the $m\times m$ identity matrix. Write $\underline{\Phi}^{*}_r=(\Phi_1,\dots,\Phi_r)$ where the $\Phi_i$'s are defined by $\eqref{AR_infty}$. We have
\begin{align}\label{conv_phichapeau_ver_phi}
\left\|\hat{\Phi}_r(1)-\Phi(1)\right\|&=\left\|\sum_{k=1}^r\hat{\Phi}_{r,k}-\sum_{k=1}^r\Phi_{r,k}+\sum_{k=1}^r\Phi_{r,k}-\sum_{k=1}^{\infty}\Phi_k\right\|\nonumber\\
&\leq \left\|\sum_{k=1}^r\left( \hat{\Phi}_{r,k}-\Phi_{r,k}\right) \right\|+\left\|\sum_{k=1}^r\left( \Phi_{r,k}-\Phi_{k}\right) \right\|+\left\|\sum_{k=r+1}^{\infty}\Phi_k\right\|\nonumber\\
&\leq \left\|\left( \underline{\hat{\Phi}}_r-\underline{\Phi}_r\right)\mathbf{E}_r\right\|+ \left\|\left( \underline{\Phi}^{*}_r-\underline{\Phi}_r\right)\mathbf{E}_r\right\|+\left\|\sum_{k=r+1}^{\infty}\Phi_k\right\|\nonumber\\
&\leq \sqrt{p+q+1}\sqrt{r}\left( \left\|\underline{\hat{\Phi}}_r-\underline{\Phi}_r\right\|+ \left\| \underline{\Phi}^{*}_r-\underline{\Phi}_r\right\|\right) +\left\|\sum_{k=r+1}^{\infty}\Phi_k\right\|.
\end{align}
Under the assumptions of Theorem $\ref{convergence_Isp}$ we have 
$$\left\|\sum_{k=r+1}^{\infty}\Phi_k\right\|\leq \sum_{k=r+1}^{\infty}\left\|\Phi_k\right\| \xrightarrow[n\to\infty]{} 0 .$$
Therefore it is enough to show that $\sqrt{r}\|\underline{\hat{\Phi}}_r-\underline{\Phi}_r\|$ and $\sqrt{r}\| \underline{\Phi}^{*}_r-\underline{\Phi}_r\|$ converge in probability towards $0$ in order to obtain the convergence in probability of $\hat{\Phi}_r(1)$ towards $\Phi(1)$. From \eqref{AR_tronquee} we have
\begin{equation}\label{ecriture_MATRICIELLE_ar_tronque}
H_t(\theta_0)=\underline{\Phi}_r\underline{H}_{r,t}(\theta_0)+u_{r,t},
\end{equation}
and thus
\begin{align*}
\Sigma_{u_r}=\mathrm{Var}(u_{r,t})
=\mathbb{E}\left[u_{r,t}\left( H_t(\theta_0)-\underline{\Phi}_r\underline{H}_{r,t}(\theta_0)\right)^{'}  \right].
\end{align*}
The vector $u_{r,t}$ is orthogonal to $\underline{H}_{r,t}(\theta_0)$. It follows that
\begin{align*}
\mathrm{Var}(u_{r,t})&=\mathbb{E}\left[\left( H_t(\theta_0)-\underline{\Phi}_r\underline{H}_{r,t}(\theta_0)\right)H_t^{'}(\theta_0)\right]\\
&=\Sigma_{H}-\underline{\Phi}_r\Sigma_{H,\underline{H}_r}^{'}.
\end{align*}
Consequently the least squares estimator of $\Sigma_{u_r}$ can be rewritten in the form:
\begin{equation}\label{estimateur_Sigma_U_r}
\hat{\Sigma}_{\hat{u}_r}=\hat{\Sigma}_{\hat{H}}-\underline{\hat{\Phi}}_r\hat{\Sigma}^{'}_{\hat{H},\underline{\hat{H}}_r},
\end{equation}
where
\begin{equation}\label{sigma_chap_de_h_chap}
\hat{\Sigma}_{\hat{H}}=\frac{1}{n}\sum_{t=1}^n\hat{H}_t\hat{H}_t^{'}.
\end{equation}
Similar arguments combined with \eqref{AR_infty} yield
\begin{align*}
\Sigma_u =\mathbb{E}\left[ u_t u_t^{'}\right] &=\mathbb{E}\left[ u_t H^{'}_t(\theta_0)\right] \\
&=\mathbb{E}\left[ H_t(\theta_0)H_t^{'}(\theta_0)\right] -\sum_{k=1}^{r}\Phi_k\mathbb{E}\left[ H_{t-k}(\theta_0)H^{'}_t(\theta_0)\right] -\sum_{k=r+1}^{\infty}\Phi_k\mathbb{E}\left[ H_{t-k}(\theta_0)H^{'}_t(\theta_0)\right]\\
&=\Sigma_{H}-\underline{\Phi}^{*}_r\Sigma^{'}_{H,\underline{H}_r}-\sum_{k=r+1}^{\infty}\Phi_k\mathbb{E}\left[ H_{t-k}(\theta_0)H^{'}_t(\theta_0)\right].
\end{align*}
By \eqref{estimateur_Sigma_U_r} we obtain
\begin{align}\label{conv_Sigmachap_vers_Sigma}
\left\|\hat{\Sigma}_{\hat{u}_r}-\Sigma_u\right\|&=\left\|\hat{\Sigma}_{\hat{H}}-\underline{\hat{\Phi}}_r\hat{\Sigma}^{'}_{\hat{H},\underline{\hat{H}}_r}-\Sigma_{H}+\underline{\Phi}^{*}_r\Sigma^{'}_{H,\underline{H}_r}+\sum_{k=r+1}^{\infty}\Phi_k\mathbb{E}\left[ H_{t-k}(\theta_0)H^{'}_t(\theta_0)\right]\right\|\nonumber\\
&=\left\|\hat{\Sigma}_{\hat{H}}-\Sigma_{H}-\left( \underline{\hat{\Phi}}_r-\underline{\Phi}^{*}_r\right) \hat{\Sigma}^{'}_{\hat{H},\underline{\hat{H}}_r}-\underline{\Phi}^{*}_r\left( \hat{\Sigma}^{'}_{\hat{H},\underline{\hat{H}}_r}-\Sigma^{'}_{H,\underline{H}_r}\right) +\sum_{k=r+1}^{\infty}\Phi_k\mathbb{E}\left[ H_{t-k}(\theta_0)H^{'}_t(\theta_0)\right]\right\|\nonumber\\
&\leq \left\|\hat{\Sigma}_{\hat{H}}-\Sigma_{H}\right\|+\left\|\left( \underline{\hat{\Phi}}_r-\underline{\Phi}^{*}_r\right) \left( \hat{\Sigma}^{'}_{\hat{H},\underline{\hat{H}}_r}-\Sigma^{'}_{H,\underline{H}_r}\right) \right\|+\left\|\left( \underline{\hat{\Phi}}_r-\underline{\Phi}^{*}_r\right)\Sigma^{'}_{H,\underline{H}_r} \right\|\nonumber\\
&\quad+\left\|\underline{\Phi}^{*}_r\left( \hat{\Sigma}^{'}_{\hat{H},\underline{\hat{H}}_r}-\Sigma^{'}_{H,\underline{H}_r}\right)\right\|+\left\|\sum_{k=r+1}^{\infty}\Phi_k\mathbb{E}\left[ H_{t-k}(\theta_0)H^{'}_t(\theta_0)\right]\right\|.
\end{align}
From Lemma~\ref{exit_I} and under Assumptions of Theorem~\ref{convergence_Isp} we deduce that
\begin{align*}
\left\|\sum_{k=r+1}^{\infty}\Phi_k\mathbb{E}\left[ H_{t-k}(\theta_0)H^{'}_t(\theta_0)\right]\right\|&\leq \sum_{k=r+1}^{\infty}\left\|\Phi_k\right\|\left\|\mathbb{E}\left[ H_{t-k}(\theta_0)H^{'}_t(\theta_0)\right]\right\| \\
&\leq K \sum_{k=r+1}^{\infty}\frac{1}{k^2}\xrightarrow[n\to\infty]{}  0.
\end{align*}
Observe also that
$$\left\|\underline{\Phi}^{*}_r\right\|^2\leq \sum_{k\geq 1}\mathrm{Tr}\left(\Phi_k\Phi_k^{'}\right) <\infty.$$
Therefore  the convergence $\hat{\Sigma}_{\hat{u}_r}$ to $\Sigma_u$ will be a consequence of the four following properties:
\begin{itemize}
\item $\|\hat{\Sigma}_{\hat{H}}-\Sigma_{H}\|=\mathrm{o}_{\mathbb{P}}(1)$,
\item $\mathbb P-\lim_{n\to\infty}\| \underline{\hat{\Phi}}_r-\underline{\Phi}^{*}_r\|=0$,
\item $\mathbb P-\lim_{n\to\infty}\|\hat{\Sigma}^{'}_{\hat{H},\underline{\hat{H}}_r}-\Sigma^{'}_{H,\underline{H}_r} \|=0$ and
\item $\|\Sigma^{'}_{H,\underline{H}_r} \|=\mathrm{O}(1)$.
\end{itemize}
The above properties will be proved thanks to several lemmas that are stated and proved hereafter. This ends the proof of Theorem \ref{convergence_Isp}.
For this, consider the following lemmas:
\begin{lemme}\label{sur_existence_des_matrices}
Under the assumptions of Theorem~\ref{convergence_Isp}, we have
$$\sup_{r\geq 1}\max\left\lbrace \left\|\Sigma_{H,\underline{H}_r}\right\|, \left\|\Sigma_{\underline{H}_r}\right\|,\left\|\Sigma_{\underline{H}_r}^{-1}\right\|\right\rbrace <\infty.$$
\end{lemme}
\begin{proof} See Lemma 1 in the supplementary material of \cite{BMCF2012}.

\end{proof}

\begin{lemme}\label{lemme_cov_des_H} Under the assumptions of Theorem $\ref{convergence_Isp}$ there exists a finite positive constant $K$ such that, for $1\le r_1,r_2\le r$ and $1 \le m_1,m_2\le p+q+1$ we have
$$\sup_{t\in\mathbb{Z}}\sum_{h=-\infty}^{\infty}\left|\mathrm{Cov}\left\lbrace H_{t-r_1,m_1}(\theta_0)H_{t-r_2,m_2}(\theta_0),H_{t-r_1-h,m_1}(\theta_0)H_{t-r_2-h,m_2}(\theta_0) \right\rbrace  \right|<K.$$
\end{lemme}

\begin{proof}  We denote in the sequel by $\overset{\textbf{.}}{\lambda}_{j,k}$ the coefficient $\overset{\textbf{.}}{\lambda}_{j,k}(\theta_0)$ defined in \eqref{Coef-lambda}.\\
Using the fact that the process $(H_t(\theta_0))_{t\in\mathbb{Z}}$ is centered and taking into consideration the strict stationarity of $(\epsilon_t)_{t\in\mathbb{Z}}$ we obtain that for any $t\in\mathbb{Z}$
\begin{align*}
\sum_{h=-\infty}^{\infty}\Big |&\mathrm{Cov}\big ( H_{t-r_1,m_1}(\theta_0)H_{t-r_2,m_2}(\theta_0) ,H_{t-r_1-h,m_1}(\theta_0)H_{t-r_2-h,m_2}(\theta_0) \big ) \Big |\\
&= \sum_{h=-\infty}^{\infty}\Big |\mathbb{E}\left[  H_{t-r_1,m_1}(\theta_0)H_{t-r_2,m_2}(\theta_0)H_{t-r_1-h,m_1}(\theta_0)H_{t-r_2-h,m_2}(\theta_0)\right]  \\
&\hspace{2cm}-\mathbb{E}\left[  H_{t-r_1,m_1}(\theta_0)H_{t-r_2,m_2}(\theta_0)\right]\mathbb{E}\left[H_{t-r_1-h,m_1}(\theta_0)H_{t-r_2-h,m_2}(\theta_0)\right]  \Big |\\
&\leq\sum_{h=-\infty}^{\infty} \Big|\mathrm{cum}\big ( H_{t-r_1,m_1}(\theta_0),H_{t-r_2,m_2}(\theta_0),H_{t-r_1-h,m_1}(\theta_0),H_{t-r_2-h,m_2}(\theta_0)\big ) \Big |\\
&\hspace{2cm}+\sum_{h=-\infty}^{\infty}\left|\mathbb{E}\left[  H_{t-r_1,m_1}(\theta_0)H_{t-r_1-h,m_1}(\theta_0)\right]\right|\left|\mathbb{E}\left[H_{t-r_2,m_2}(\theta_0)H_{t-r_2-h,m_2}(\theta_0)\right]  \right|\\
&\hspace{2cm}+\sum_{h=-\infty}^{\infty}\left|\mathbb{E}\left[  H_{t-r_1,m_1}(\theta_0)H_{t-r_2-h,m_2}(\theta_0)\right]\right|\left|\mathbb{E}\left[H_{t-r_2,m_2}(\theta_0)H_{t-r_1-h,m_1}(\theta_0)\right]  \right|\\
&\leq \sum_{h=-\infty}^{\infty}\sum_{i_1,j_1,k_1,\ell_1\geq 1}\left|\overset{\textbf{.}}{\lambda}_{i_1,m_1}\overset{\textbf{.}}{\lambda}_{j_1,m_2}\overset{\textbf{.}}{\lambda}_{k_1,m_1}\overset{\textbf{.}}{\lambda}_{\ell_1,m_2}\right|\left|\mathrm{cum}\left(\epsilon_{0}\epsilon_{-i_1},\epsilon_{r_1-r_2}\epsilon_{r_1-r_2-j_1},\epsilon_{-h}\epsilon_{-h-k_1},\epsilon_{r_1-r_2-h}\epsilon_{r_1-r_2-h-\ell_1}\right)\right| \\
&\hspace{2cm} +T_{r1,m_1,r_2,m_2}^{(1)}+T_{r1,m_1,r_2,m_2}^{(2)},
\end{align*}
 where
\begin{align*}
T_{r1,m_1,r_2,m_2}^{(1)}&=\sum_{h=-\infty}^{\infty}\left|\mathbb{E}\left[  H_{t-r_1,m_1}(\theta_0)H_{t-r_1-h,m_1}(\theta_0)\right]\right|\left|\mathbb{E}\left[H_{t-r_2,m_2}(\theta_0)H_{t-r_2-h,m_2}(\theta_0)\right]  \right|
\end{align*}
and
\begin{align*}
T_{r1,m_1,r_2,m_2}^{(2)}&=\sum_{h=-\infty}^{\infty}\left|\mathbb{E}\left[  H_{t-r_1,m_1}(\theta_0)H_{t-r_2-h,m_2}(\theta_0)\right]\right|\left|\mathbb{E}\left[H_{t-r_2,m_2}(\theta_0)H_{t-r_1-h,m_1}(\theta_0)\right]  \right|.
\end{align*}
Thanks to Lemma \ref{miss2} one may use the product theorem for the joint cumulants  (\cite{Brillinger1975}) as in the proof of Lemma A.3. in \cite{shao2011} in order to obtain that
\begin{align*}
&\sum_{h=-\infty}^{\infty}\sum_{i_1,j_1,k_1,\ell_1\geq 1}\left|\overset{\textbf{.}}{\lambda}_{i_1,m_1}\overset{\textbf{.}}{\lambda}_{j_1,m_2}\overset{\textbf{.}}{\lambda}_{k_1,m_1}\overset{\textbf{.}}{\lambda}_{\ell_1,m_2}\right|\left|\mathrm{cum}\left(\epsilon_{0}\epsilon_{-i_1},\epsilon_{r_1-r_2}\epsilon_{r_1-r_2-j_1},\epsilon_{-h}\epsilon_{-h-k_1},\epsilon_{r_1-r_2-h}\epsilon_{r_1-r_2-h-\ell_1}\right)\right|\\
&<\infty
\end{align*}
where we have used the absolute summability of the $k$-th $(k=2,\dots,8)$ cumulants assumed in {\bf (A4')} with $\tau=8$.

Observe now that
\begin{align*}
T_{r1,m_1,r_2,m_2}^{(1)}&=\sum_{h=-\infty}^{\infty}\left|\mathbb{E}\left[  H_{t-r_1,m_1}(\theta_0)H_{t-r_1-h,m_1}(\theta_0)\right]\right|\left|\mathbb{E}\left[H_{t-r_2,m_2}(\theta_0)H_{t-r_2-h,m_2}(\theta_0)\right]  \right|\\
&\leq\sup_{h\in\mathbb{Z}}\left|\mathbb{E}\left[  H_{t-r_1,m_1}(\theta_0)H_{t-r_1-h,m_1}(\theta_0)\right]\right|\sum_{h=-\infty}^{\infty}\left|\mathbb{E}\left[H_{t-r_2,m_2}(\theta_0)H_{t-r_2-h,m_2}(\theta_0)\right]  \right|.
\end{align*}
For any $h\in\mathbb{Z}$, from \eqref{epsil-th} we have
\begin{align*}
\left|\mathbb{E}\left[  H_{t-r_1,m_1}(\theta_0)H_{t-r_1-h,m_1}(\theta_0)\right]\right|& \leq\sum_{i,j\geq 1}\left|\overset{\textbf{.}}{\lambda}_{i,m_1}\right|\left|\overset{\textbf{.}}{\lambda}_{j,m_1}\right|\left|\mathrm{cum}\left(\epsilon_0,\epsilon_{-i},\epsilon_{-h},\epsilon_{-h-j}\right)\right|\\
&\hspace{-2cm}+\sum_{i,j\geq 1}\left|\overset{\textbf{.}}{\lambda}_{i,m_1}\right|\left|\overset{\textbf{.}}{\lambda}_{j,m_1}\right|\Bigg\{  \left|\mathbb{E}\left[ \epsilon_0\epsilon_{-i}\right]\mathbb{E}\left[ \epsilon_{-h}\epsilon_{-h-j}\right] \right| \\
&\hspace{-2cm}\hspace{1.5cm}+\left|\mathbb{E}\left[ \epsilon_0\epsilon_{-h}\right]\mathbb{E}\left[ \epsilon_{-i}\epsilon_{-h-j}\right] \right|+\left|\mathbb{E}\left[ \epsilon_0\epsilon_{-h-j}\right]\mathbb{E}\left[ \epsilon_{-i}\epsilon_{-h}\right] \right|\Bigg \} \\
&\hspace{-2cm}\leq \sum_{i,j\geq 1}\left|\mathrm{cum}\left(\epsilon_0,\epsilon_{-i},\epsilon_{-h},\epsilon_{-h-j}\right)\right|+\sigma_{\epsilon}^4\sum_{i\geq 1}\left|\overset{\textbf{.}}{\lambda}_{i,m_1}\right|^2 .
\end{align*}
Under Assumption {\bf (A4')} with $\tau=4$ and in view of Lemma \ref{miss2} we may write that
\begin{align*}
\sup_{h\in\mathbb{Z}}\left|\mathbb{E}\left[  H_{t-r_1,m_1}(\theta_0)H_{t-r_1-h,m_1}(\theta_0)\right]\right|&\leq\sup_{h\in\mathbb{Z}}\sum_{i,j\geq 1}\left|\mathrm{cum}\left(\epsilon_0,\epsilon_{-i},\epsilon_{-h},\epsilon_{-h-j}\right)\right|+\sigma_{\epsilon}^4\sum_{i\geq 1}\left|\overset{\textbf{.}}{\lambda}_{i,m_1}\right|^2 <\infty . 
\end{align*}
Similarly, we obtain
\begin{align*}
\sum_{h=-\infty}^{\infty}\left|\mathbb{E}\left[H_{t-r_2,m_2}(\theta_0)H_{t-r_2-h,m_2}(\theta_0)\right]  \right|&\leq\sum_{h=-\infty}^{\infty}\sum_{i,j\geq 1}\left|\mathrm{cum}\left(\epsilon_0,\epsilon_{-i},\epsilon_{-h},\epsilon_{-h-j}\right)\right|+\sigma_{\epsilon}^4\sum_{i\geq 1}\left|\overset{\textbf{.}}{\lambda}_{i,m_1}\right|^2\\
&<\infty.
\end{align*}
Consequently $T_{r1,m_1,r_2,m_2}^{(1)}<\infty$ and the same approach yields that $
T_{r1,m_1,r_2,m_2}^{(2)}<\infty$ and the lemma is proved.
\end{proof}
 Let $\hat{\Sigma}_{\underline{H}_r}$, $\hat{\Sigma}_{H}$ and $\hat{\Sigma}_{H,\underline{H}_r}$ be the matrices obtained by replacing $\hat{H}_t$ by $H_t(\theta_0)$ in  $\hat{\Sigma}_{\underline{\hat{H}}_r}$, $\hat{\Sigma}_{\hat{H}}$ and $\hat{\Sigma}_{\hat{H},\underline{\hat{H}}_r}$.
\begin{lemme}\label{con_prob_Sigmachap_H} Under the assumptions of Theorem \ref{convergence_Isp}, $\sqrt{r}\|\hat{\Sigma}_{\underline{H}_r}-\Sigma_{\underline{H}_r}\|$,  $\sqrt{r}\|\hat{\Sigma}_{H,\underline{H}_r}-\Sigma_{H,\underline{H}_r}\|$ and  $\sqrt{r}\|\hat{\Sigma}_{H}-\Sigma_{H}\|$ tend to zero in probability as $n\rightarrow\infty$ when $r=\mathrm{o}(n^{1/3})$.
\end{lemme}
\begin{proof} For $1\le m_1,m_2 \le p+q+1 $ and $1\le r_1,r_2 \le r $, the $( \lbrace (r_1-1)(p+q+1)+m_1\rbrace,\lbrace (r_2-1)(p+q+1)+m_2\rbrace)-$th element of $\hat{\Sigma}_{\underline{H}_r}$ is given by: $$\frac{1}{n}\sum_{t=1}^nH_{t-r_1,m_1}(\theta_0)H_{t-r_2,m_2}(\theta_0).$$
For all $\beta>0$, we use \eqref{ineg_norme_matricielle} and we obtain
\begin{align*}
\mathbb{P}\left( \sqrt{r}\left\|\hat{\Sigma}_{\underline{H}_r}-\Sigma_{\underline{H}_r}\right\|   \geq \beta\right) & \leq \frac{r}{\beta^2}\mathbb{E}\left\|\hat{\Sigma}_{\underline{H}_r}-\Sigma_{\underline{H}_r}\right\|^2\\
& \leq \frac{r}{\beta^2}\mathbb{E}\left\|\frac{1}{n}\sum_{t=1}^n\underline{H}_{r,t}\underline{H}_{r,t}^{'}-\mathbb{E}\left[ \underline{H}_{r,t}\underline{H}_{r,t}^{'}\right] \right\|^2\\
&\leq \frac{r}{\beta^2}\sum_{r_1=1}^r\sum_{r_2=1}^r\sum_{m_1=1}^{p+q+1}\sum_{m_2=1}^{p+q+1}\mathbb{E}\Bigg(\frac{1}{n}\sum_{t=1}^nH_{t-r_1,m_1}(\theta_0)H_{t-r_2,m_2}(\theta_0)  \\
& \hspace{5cm} -\mathbb{E}\left[ H_{t-r_1,m_1}(\theta_0)H_{t-r_2,m_2}(\theta_0)\right] \Bigg)^2.
\end{align*}
The stationarity of the process $\left( H_{t-r_1,m_1}(\theta_0)H_{t-r_2,m_2}(\theta_0)\right) _{t\in\mathbb{Z}}$ and Lemma~\ref{lemme_cov_des_H} imply
\begin{align*}
\mathbb{P}&\left( \sqrt{r}\left\|\hat{\Sigma}_{\underline{H}_r}-\Sigma_{\underline{H}_r}\right\|\geq \beta\right)
\\&\leq\frac{r}{\beta^2}\sum_{r_1=1}^r\sum_{r_2=1}^r\sum_{m_1=1}^{p+q+1}\sum_{m_2=1}^{p+q+1}\mathrm{Var}\left(\frac{1}{n}\sum_{t=1}^nH_{t-r_1,m_1}(\theta_0)H_{t-r_2,m_2}(\theta_0) \right)\\
&\leq\frac{r}{(n\beta)^2}\sum_{r_1=1}^r\sum_{r_2=1}^r\sum_{m_1=1}^{p+q+1}\sum_{m_2=1}^{p+q+1}\sum_{t=1}^n\sum_{s=1}^n\mathrm{Cov}\left(H_{t-r_1,m_1}(\theta_0)H_{t-r_2,m_2}(\theta_0),H_{s-r_1,m_1}(\theta_0)H_{s-r_2,m_2}(\theta_0) \right) \\
&\leq\frac{r}{(n\beta)^2}\sum_{r_1=1}^r\sum_{r_2=1}^r\sum_{m_1=1}^{p+q+1}\sum_{m_2=1}^{p+q+1}\sum_{h=1-n}^{n-1}(n-|h|)\mathrm{Cov}\left(H_{t-r_1,m_1}(\theta_0)H_{t-r_2,m_2}(\theta_0),H_{t-h-r_1,m_1}(\theta_0)H_{t-h-r_2,m_2}(\theta_0) \right) \\
&\leq \frac{r}{n\beta^2}\sum_{r_1=1}^r\sum_{r_2=1}^r\sum_{m_1=1}^{p+q+1}\sum_{m_2=1}^{p+q+1}\sup_{t\in\mathbb{Z}}\sum_{h=-\infty}^{\infty}\left|\mathrm{Cov}\left(H_{t-r_1,m_1}(\theta_0)H_{t-r_2,m_2}(\theta_0),H_{t-h-r_1,m_1}(\theta_0)H_{t-h-r_2,m_2}(\theta_0) \right)\right|\\
&\leq \frac{C(p+q+1)^2r^3}{n\beta^2}.
\end{align*}
Consequently we have
\begin{align*}
\mathbb{E}\left[ r\left\|\hat{\Sigma}_{H}-\Sigma_{H}\right\|^2\right]&  \leq \mathbb{E}\left[ r\left\|\hat{\Sigma}_{H,\underline{H}_r}-\Sigma_{H,\underline{H}_r}\right\|^2\right]
\\
& \leq \mathbb{E}\left[ r\left\|\hat{\Sigma}_{\underline{H}_r}-\Sigma_{\underline{H}_r}\right\|^2\right]\\
& \leq \frac{C(p+q+1)^2r^3}{n}\xrightarrow[n\to\infty]{} 0
\end{align*}
 when $r=\mathrm{o}(n^{1/3})$.
%
The conclusion follows.
\end{proof}

We show in the following lemma that the previous lemma remains valid when we replace $H_t(\theta_0)$ by $\hat{H}_t$.
\begin{lemme}\label{con_prob_Sigmachap_Hchap} Under the assumptions of Theorem \ref{convergence_Isp}, $\sqrt{r}\|\hat{\Sigma}_{\underline{\hat{H}}_r}-\Sigma_{\underline{H}_r}\|$,  $\sqrt{r}\|\hat{\Sigma}_{\hat{H},\underline{\hat{H}}_r}-\Sigma_{H,\underline{H}_r}\|$ and  $\sqrt{r}\|\hat{\Sigma}_{\hat{H}}-\Sigma_{H}\|$ tend to zero in probability as $n\rightarrow\infty$ when $r=\mathrm{o}(n^{(1-2(d_0-d_1))/5})$.
\end{lemme}

\begin{proof}
As mentioned in the end of the proof of the previous lemma, we only have to deal with the term $\sqrt{r}\|\hat{\Sigma}_{\underline{\hat{H}}_r}-\Sigma_{\underline{H}_r}\|$.

We denote $\hat{\Sigma}_{\underline{H}_{r,n}}$ the matrix obtained by replacing $\tilde{\epsilon}_t(\hat{\theta}_n)$ by $\epsilon_t(\hat{\theta}_n)$ in $\hat{\Sigma}_{\underline{\hat{H}}_{r}}$. We have
\begin{align*}
\sqrt{r}\left\|\hat{\Sigma}_{\underline{\hat{H}}_r}-\Sigma_{\underline{H}_r}\right\|&\leq \sqrt{r}\left\|\hat{\Sigma}_{\underline{\hat{H}}_r}-\hat{\Sigma}_{\underline{H}_{r,n}}\right\|+\sqrt{r}\left\|\hat{\Sigma}_{\underline{H}_{r,n}}-\hat{\Sigma}_{\underline{H}_r}\right\|
+\sqrt{r}\left\|\hat{\Sigma}_{\underline{H}_r}-\Sigma_{\underline{H}_r}\right\|.
\end{align*}
By Lemma~\ref{con_prob_Sigmachap_H}, the term $\sqrt{r}\|\hat{\Sigma}_{\underline{H}_r}-\Sigma_{\underline{H}_r}\|$  converges in probability.
The lemma will be proved as soon as we show that
\begin{align}\label{tr1}
\sqrt{r}\left\|\hat{\Sigma}_{\underline{\hat{H}}_r}-\hat{\Sigma}_{\underline{H}_{r,n}}\right\|&=\mathrm{o}_{\mathbb{P}}(1)\quad\text{ and } \\
 \label{tr2}
\sqrt{r}\left\|\hat{\Sigma}_{\underline{H}_{r,n}}-\hat{\Sigma}_{\underline{H}_r}\right\|&=\mathrm{o}_{\mathbb{P}}(1),
\end{align}
 when $r=\mathrm{o}(n^{(1-2(d_0-d_1))/5})$. This is done in two separate steps.

\paragraph*{Step 1: proof of \eqref{tr1}.}\ \\
For all $\beta>0$, we have
\begin{align*}
\mathbb{P}\left( \sqrt{r}\left\|\hat{\Sigma}_{\hat{\underline{H}}_r}-\hat{\Sigma}_{\underline{H}_{r,n}}\right\|\geq \beta\right) & \leq \frac{\sqrt{r}}{\beta}\mathbb{E}\left\|\hat{\Sigma}_{\hat{\underline{H}}_r}-\hat{\Sigma}_{\underline{H}_{r,n}}\right\| \\
& \leq \frac{\sqrt{r}}{\beta}\mathbb{E}\left\|\frac{1}{n}\sum_{t=1}^n\hat{\underline{H}}_{r,t}\hat{\underline{H}}_{r,t}^{'}-\frac{1}{n}\sum_{t=1}^n \underline{H}_{r,t}^{(n)}\underline{H}_{r,t}^{(n) '} \right\|\\
&\leq \frac{K\sqrt{r}}{\beta}\sum_{r_1=1}^r\sum_{r_2=1}^r\sum_{m_1=1}^{p+q+1}\sum_{m_2=1}^{p+q+1}\mathbb{E}\left|\frac{1}{n}\sum_{t=1}^n\hat{H}_{t-r_1,m_1}\hat{H}_{t-r_2,m_2}-\frac{1}{n}\sum_{t=1}^nH^{(n)}_{t-r_1,m_1}H_{t-r_2,m_2}^{(n)}\right|,
\end{align*}
where
$$H_{t,m}^{(n)}=2\epsilon_t(\hat{\theta}_n)\frac{\partial}{\partial\theta_m}\epsilon_t(\hat{\theta}_n)\quad\text{ and }\quad \underline{H}_{r,t}^{(n)}=\left(H_{t-1}^{(n) '},\dots,H_{t-r}^{(n) '} \right) ^{'} \ .$$
It is follow that
\begin{align}\label{trou}
\mathbb{P}\left( \sqrt{r}\left\|\hat{\Sigma}_{\hat{\underline{H}}_r}-\hat{\Sigma}_{\underline{H}_{r,n}}\right\|\geq \beta\right) \nonumber
& \leq \frac{4K\sqrt{r}}{n\beta}\sum_{r_1=1}^r\sum_{r_2=1}^r\sum_{m_1=1}^{p+q+1}\sum_{m_2=1}^{p+q+1}  \nonumber\\
&\hspace{2cm}\mathbb{E}\Bigg |\sum_{t=1}^n\tilde{\epsilon}_{t-r_1}(\hat{\theta}_n)\frac{\partial}{\partial\theta_{m_1}}\tilde{\epsilon}_{t-r_1}(\hat{\theta}_n)\tilde{\epsilon}_{t-r_2}(\hat{\theta}_n)\frac{\partial}{\partial\theta_{m_2}}\tilde{\epsilon}_{t-r_2}(\hat{\theta}_n)  \nonumber\\
&\hspace{3cm}-\epsilon_{t-r_1}(\hat{\theta}_n)\frac{\partial}{\partial\theta_{m_1}}\epsilon_{t-r_1}(\hat{\theta}_n)\epsilon_{t-r_2}(\hat{\theta}_n)\frac{\partial}{\partial\theta_{m_2}}\epsilon_{t-r_2}(\hat{\theta}_n)\Bigg |.
\end{align}
Observe now that
\begin{align*}
&\tilde{\epsilon}_{t-r_1}(\hat{\theta}_n)\frac{\partial}{\partial\theta_{m_1}}\tilde{\epsilon}_{t-r_1}(\hat{\theta}_n)\tilde{\epsilon}_{t-r_2}(\hat{\theta}_n)\frac{\partial}{\partial\theta_{m_2}}\tilde{\epsilon}_{t-r_2}(\hat{\theta}_n)-\epsilon_{t-r_1}(\hat{\theta}_n)\frac{\partial}{\partial\theta_{m_1}}\epsilon_{t-r_1}(\hat{\theta}_n)\epsilon_{t-r_2}(\hat{\theta}_n)\frac{\partial}{\partial\theta_{m_2}}\epsilon_{t-r_2}(\hat{\theta}_n)\\
&\hspace{3cm}=\left(\tilde{\epsilon}_{t-r_1}(\hat{\theta}_n)-\epsilon_{t-r_1}(\hat{\theta}_n)\right)\frac{\partial}{\partial\theta_{m_1}}\tilde{\epsilon}_{t-r_1}(\hat{\theta}_n)\tilde{\epsilon}_{t-r_2}(\hat{\theta}_n)\frac{\partial}{\partial\theta_{m_2}}\tilde{\epsilon}_{t-r_2}(\hat{\theta}_n)\\
&\hspace{3.5cm}+\epsilon_{t-r_1}(\hat{\theta}_n)\left(\frac{\partial}{\partial\theta_{m_1}}\tilde{\epsilon}_{t-r_1}(\hat{\theta}_n)-\frac{\partial}{\partial\theta_{m_1}}\epsilon_{t-r_1}(\hat{\theta}_n)\right)\tilde{\epsilon}_{t-r_2}(\hat{\theta}_n)\frac{\partial}{\partial\theta_{m_2}}\tilde{\epsilon}_{t-r_2}(\hat{\theta}_n)\\
&\hspace{3.5cm}+\epsilon_{t-r_1}(\hat{\theta}_n)\frac{\partial}{\partial\theta_{m_1}}\epsilon_{t-r_1}(\hat{\theta}_n)\left(\tilde{\epsilon}_{t-r_2}(\hat{\theta}_n)-\epsilon_{t-r_2}(\hat{\theta}_n)\right)\frac{\partial}{\partial\theta_{m_2}}\tilde{\epsilon}_{t-r_2}(\hat{\theta}_n)\\
&\hspace{3.5cm}+\epsilon_{t-r_1}(\hat{\theta}_n)\frac{\partial}{\partial\theta_{m_1}}\epsilon_{t-r_1}(\hat{\theta}_n)\epsilon_{t-r_2}(\hat{\theta}_n)\left(\frac{\partial}{\partial\theta_{m_2}}\tilde{\epsilon}_{t-r_2}(\hat{\theta}_n)-\frac{\partial}{\partial\theta_{m_2}}\epsilon_{t-r_2}(\hat{\theta}_n)\right).
\end{align*}
We replace the above identity in \eqref{trou} and we obtain by H\"{o}lder's inequality that
\begin{align}\label{trou2}
\mathbb{P}\left( \sqrt{r}\left\|\hat{\Sigma}_{\hat{\underline{H}}_r}-\hat{\Sigma}_{\underline{H}_{r,n}}\right\|\geq \beta\right)
&\leq  \frac{4K\sqrt{r}}{n\beta}\sum_{r_1=1}^r\sum_{r_2=1}^r\sum_{m_1=1}^{p+q+1}\sum_{m_2=1}^{p+q+1}\left(T_{n,1}+T_{n,2}+T_{n,3}+T_{n,4}\right)
\end{align}
where
\begin{align*}
T_{n,1}&=\sum_{t=1}^n\left\|\tilde{\epsilon}_{t-r_1}(\hat{\theta}_n)-\epsilon_{t-r_1}(\hat{\theta}_n)\right\|_{\mathbb{L}^2}\left\|\frac{\partial}{\partial\theta_{m_1}}\tilde{\epsilon}_{t-r_1}(\hat{\theta}_n)\right\|_{\mathbb{L}^6}\left\|\tilde{\epsilon}_{t-r_2}(\hat{\theta}_n)\right\|_{\mathbb{L}^6}\left\|\frac{\partial}{\partial\theta_{m_2}}\tilde{\epsilon}_{t-r_2}(\hat{\theta}_n)\right\|_{\mathbb{L}^6},\\
T_{n,2}&=\sum_{t=1}^n\left\|\epsilon_{t-r_1}(\hat{\theta}_n)\right\|_{\mathbb{L}^6}\left\|\frac{\partial}{\partial\theta_{m_1}}\tilde{\epsilon}_{t-r_1}(\hat{\theta}_n)-\frac{\partial}{\partial\theta_{m_1}}\epsilon_{t-r_1}(\hat{\theta}_n)\right\|_{\mathbb{L}^2}\left\|\tilde{\epsilon}_{t-r_2}(\hat{\theta}_n)\right\|_{\mathbb{L}^6}\left\|\frac{\partial}{\partial\theta_{m_2}}\tilde{\epsilon}_{t-r_2}(\hat{\theta}_n)\right\|_{\mathbb{L}^6},\\
T_{n,3}&=\sum_{t=1}^n\left\|\epsilon_{t-r_1}(\hat{\theta}_n)\right\|_{\mathbb{L}^6}\left\|\frac{\partial}{\partial\theta_{m_1}}\epsilon_{t-r_1}(\hat{\theta}_n)\right\|_{\mathbb{L}^6}\left\|\tilde{\epsilon}_{t-r_2}(\hat{\theta}_n)-\epsilon_{t-r_2}(\hat{\theta}_n)\right\|_{\mathbb{L}^2}\left\|\frac{\partial}{\partial\theta_{m_2}}\tilde{\epsilon}_{t-r_2}(\hat{\theta}_n)\right\|_{\mathbb{L}^6}, \\
T_{n,4}&=\sum_{t=1}^n\left\|\epsilon_{t-r_1}(\hat{\theta}_n)\right\|_{\mathbb{L}^6}\left\|\frac{\partial}{\partial\theta_{m_1}}\epsilon_{t-r_1}(\hat{\theta}_n)\right\|_{\mathbb{L}^6}\left\|\epsilon_{t-r_2}(\hat{\theta}_n)\right\|_{\mathbb{L}^6}\left\|\frac{\partial}{\partial\theta_{m_2}}\tilde{\epsilon}_{t-r_2}(\hat{\theta}_n)-\frac{\partial}{\partial\theta_{m_2}}\epsilon_{t-r_2}(\hat{\theta}_n)\right\|_{\mathbb{L}^2}.
\end{align*}
For all $\theta\in\Theta_{\delta}$ and $t\in\mathbb{Z}$, in view of \eqref{epsil-th} and Remark \ref{rmq:important}, we have
\begin{align*}
\left\|\tilde{\epsilon}_{t}(\hat{\theta}_n)-\epsilon_{t}(\hat{\theta}_n)\right\|_{\mathbb{L}^2}&=\left(\mathbb{E}\left[\left\lbrace \sum_{j\geq 0}\left(\lambda^{t}_j(\hat{\theta}_n)-\lambda_j(\hat{\theta}_n)\right)\epsilon_{t-j}\right\rbrace ^2\right]\right)^{1/2}\\
&\leq\sup_{\theta\in\Theta_{\delta}}\left(\mathbb{E}\left[\left\lbrace \sum_{j\geq 0}\left(\lambda^{t}_j({\theta})-\lambda_j({\theta})\right)\epsilon_{t-j}\right\rbrace ^2\right]\right)^{1/2}\\
&\le  \sigma_{\epsilon}\sup_{\theta\in\Theta_{\delta}}\left\|\lambda(\theta)-\lambda^{t}(\theta)\right\|_{\ell^2}\\
&\leq  K\frac{1}{t^{1/2+(d_1-d_0)}}.
\end{align*}
It is not difficult to prove that $\tilde{\epsilon}_t(\theta)$ and $\partial\tilde{\epsilon}_t(\theta)/\partial\theta$ belong to $\mathbb L^6$. The fact that $\epsilon_t(\theta)$ and $\partial\epsilon_t(\theta)/\partial\theta$ have moment of order $6$ 
can be proved using the same method than in Lemma \ref{lemme_cov_des_H} using the absolute summability of the $k$-th $(k=2,\dots,8)$ cumulants assumed in {\bf (A4')} with $\tau=8$.
We deduce that
\begin{align*}
T_{n,1}\leq K\sum_{t=1}^n\left\|\tilde{\epsilon}_{t-r_1}(\hat{\theta}_n)-\epsilon_{t-r_1}(\hat{\theta}_n)\right\|_{\mathbb{L}^2}&\leq K \sum_{t=1-r}^0\left\|\epsilon_{t}(\hat{\theta}_n)\right\|_{\mathbb{L}^2}+K\sum_{t=1}^{n}\left\|\tilde{\epsilon}_{t}(\hat{\theta}_n)-\epsilon_{t}(\hat{\theta}_n)\right\|_{\mathbb{L}^2}\\
&\leq K\left(r+\sum_{t=1}^{n}\frac{1}{t^{1/2+(d_1-d_0)}}\right).
\end{align*}
Then  we obtain
\begin{align}\label{T1}
T_{n,1}&\leq K\left(r+ n^{1/2-(d_1-d_0)}\right).
\end{align}
The same calculations hold for the terms $T_{n,2}$, $T_{n,3}$ and $T_{n,4}$. Thus
\begin{equation}\label{T1_T4}
T_{n,1}+T_{n,2}+T_{n,3}+T_{n,4}\leq K\left(r+ n^{1/2-(d_1-d_0)}\right)
\end{equation}
and reporting this estimation in \eqref{trou2} implies that
\begin{align*}
\mathbb{P}\left( \sqrt{r}\left\|\hat{\Sigma}_{\hat{\underline{H}}_r}-\hat{\Sigma}_{\underline{H}_{r,n}}\right\|\geq \beta\right) & \leq\frac{Kr^{5/2}(p+q+1)^2}{n\beta}\left(r+n^{1/2-(d_1-d_0)}\right)\\
&\leq K\left(\frac{r^{7/2}}{n}+\frac{r^{5/2}}{n^{1/2+(d_1-d_0)}}\right).
\end{align*}
Since $2/7>(1+2(d_1-d_0))/5$, the sequence $\sqrt{r}\left\|\hat{\Sigma}_{\hat{\underline{H}}_r}-\hat{\Sigma}_{\underline{H}_{r,n}}\right\|$ converges in probability to $0$ as $n\rightarrow\infty$ when $r=r(n)=\mathrm{o}(n^{(1-2(d_0-d_1))/5})$.
\paragraph*{Step 2: proof of \eqref{tr2}.}\ \\
First we follow the same approach than in the previous step.
We have
\begin{align*}
\left\|\hat{\Sigma}_{\underline{H}_{r,n}}-\hat{\Sigma}_{\underline{H}_r}\right\|^2&=\left\|\frac{1}{n}\sum_{t=1}^n \underline{H}_{r,t}^{(n)}\underline{H}_{r,t}^{(n) '}-\frac{1}{n}\sum_{t=1}^n\underline{H}_{r,t}\underline{H}_{r,t}^{'} \right\|^2\\
&\leq \sum_{r_1=1}^r\sum_{r_2=1}^r\sum_{m_1=1}^{p+q+1}\sum_{m_2=1}^{p+q+1}\left(\frac{1}{n}\sum_{t=1}^nH^{(n)}_{t-r_1,m_1}H_{t-r_2,m_2}^{(n)}-\frac{1}{n}\sum_{t=1}^nH_{t-r_1,m_1}H_{t-r_2,m_2}\right)^2\\
& \leq 16\sum_{r_1=1}^r\sum_{r_2=1}^r\sum_{m_1=1}^{p+q+1}\sum_{m_2=1}^{p+q+1}\left(\frac{1}{n}\sum_{t=1}^n\epsilon_{t-r_1}(\hat{\theta}_n)\frac{\partial}{\partial\theta_{m_1}}\epsilon_{t-r_1}(\hat{\theta}_n)\epsilon_{t-r_2}(\hat{\theta}_n)\frac{\partial}{\partial\theta_{m_2}}\epsilon_{t-r_2}(\hat{\theta}_n)\right.\\
&\hspace{5cm}\left. -\epsilon_{t-r_1}(\theta_0)\frac{\partial}{\partial\theta_{m_1}}\epsilon_{t-r_1}(\theta_0)\epsilon_{t-r_2}(\theta_0)\frac{\partial}{\partial\theta_{m_2}}\epsilon_{t-r_2}(\theta_0)\right)^2.
\end{align*}
Since
\begin{align*}
&\epsilon_{t-r_1}(\hat{\theta}_n)\frac{\partial}{\partial\theta_{m_1}}\epsilon_{t-r_1}(\hat{\theta}_n)\epsilon_{t-r_2}(\hat{\theta}_n)\frac{\partial}{\partial\theta_{m_2}}\epsilon_{t-r_2}(\hat{\theta}_n)-\epsilon_{t-r_1}(\theta_0)\frac{\partial}{\partial\theta_{m_1}}\epsilon_{t-r_1}(\theta_0)\epsilon_{t-r_2}(\theta_0)\frac{\partial}{\partial\theta_{m_2}}\epsilon_{t-r_2}(\theta_0)\\
&\hspace{2cm}=\left(\epsilon_{t-r_1}(\hat{\theta}_n)-\epsilon_{t-r_1}(\theta_0)\right)\frac{\partial}{\partial\theta_{m_1}}\epsilon_{t-r_1}(\hat{\theta}_n)\epsilon_{t-r_2}(\hat{\theta}_n)\frac{\partial}{\partial\theta_{m_2}}\epsilon_{t-r_2}(\hat{\theta}_n)\\
&\hspace{4cm}+\epsilon_{t-r_1}(\theta_0)\left(\frac{\partial}{\partial\theta_{m_1}}\epsilon_{t-r_1}(\hat{\theta}_n)-\frac{\partial}{\partial\theta_{m_1}}\epsilon_{t-r_1}(\theta_0)\right)\epsilon_{t-r_2}(\hat{\theta}_n)\frac{\partial}{\partial\theta_{m_2}}\epsilon_{t-r_2}(\hat{\theta}_n)\\
&\hspace{4cm}+\epsilon_{t-r_1}(\theta_0)\frac{\partial}{\partial\theta_{m_1}}\epsilon_{t-r_1}(\theta_0)\left(\epsilon_{t-r_2}(\hat{\theta}_n)-\epsilon_{t-r_2}(\theta_0)\right)\frac{\partial}{\partial\theta_{m_2}}\epsilon_{t-r_2}(\hat{\theta}_n)\\
&\hspace{4cm}+\epsilon_{t-r_1}(\theta_0)\frac{\partial}{\partial\theta_{m_1}}\epsilon_{t-r_1}(\theta_0)\epsilon_{t-r_2}(\theta_0)\left(\frac{\partial}{\partial\theta_{m_2}}\epsilon_{t-r_2}(\hat{\theta}_n)-\frac{\partial}{\partial\theta_{m_2}}\epsilon_{t-r_2}(\theta_0)\right),
\end{align*}
one has
\begin{align}\label{cu}
\left\|\hat{\Sigma}_{\underline{H}_{r,n}}-\hat{\Sigma}_{\underline{H}_r}\right\|^2& \leq16\sum_{r_1=1}^r\sum_{r_2=1}^r\sum_{m_1=1}^{p+q+1}\sum_{m_2=1}^{p+q+1}\left(U_{n,1}+U_{n,2}+U_{n,3}+U_{n,4}\right)^2
\end{align}
where
\begin{align*}
U_{n,1}&=\frac{1}{n}\sum_{t=1}^n\left|\epsilon_{t-r_1}(\hat{\theta}_n)-\epsilon_{t-r_1}(\theta_0)\right|\left|\frac{\partial}{\partial\theta_{m_1}}\epsilon_{t-r_1}(\hat{\theta}_n)\right|\left|\epsilon_{t-r_2}(\hat{\theta}_n)\right|\left|\frac{\partial}{\partial\theta_{m_2}}\epsilon_{t-r_2}(\hat{\theta}_n)\right|,\\
U_{n,2}&=\frac{1}{n}\sum_{t=1}^n\left|\epsilon_{t-r_1}(\theta_0)\right|\left|\frac{\partial}{\partial\theta_{m_1}}\epsilon_{t-r_1}(\hat{\theta}_n)-\frac{\partial}{\partial\theta_{m_1}}\epsilon_{t-r_1}(\theta_0)\right|\left|\epsilon_{t-r_2}(\hat{\theta}_n)\right|\left|\frac{\partial}{\partial\theta_{m_2}}\epsilon_{t-r_2}(\hat{\theta}_n)\right|,\\
U_{n,3}&=\frac{1}{n}\sum_{t=1}^n\left|\epsilon_{t-r_1}(\theta_0)\right|\left|\frac{\partial}{\partial\theta_{m_1}}\epsilon_{t-r_1}(\theta_0)\right|\left|\epsilon_{t-r_2}(\hat{\theta}_n)-\epsilon_{t-r_2}(\theta_0)\right|\left|\frac{\partial}{\partial\theta_{m_2}}\epsilon_{t-r_2}(\hat{\theta}_n)\right| \\
U_{n,4}&=\frac{1}{n}\sum_{t=1}^n\left|\epsilon_{t-r_1}(\theta_0)\right|\left|\frac{\partial}{\partial\theta_{m_1}}\epsilon_{t-r_1}(\theta_0)\right|\left|\epsilon_{t-r_2}(\theta_0)\right|\left|\frac{\partial}{\partial\theta_{m_2}}\epsilon_{t-r_2}(\hat{\theta}_n)-\frac{\partial}{\partial\theta_{m_2}}\epsilon_{t-r_2}(\theta_0)\right|.
\end{align*}
Taylor expansions around $\theta_0$ yield that there exists $\underline{\theta}$ and $\overline{\theta}$ between $\hat{\theta}_n$ and $\theta_0$ such that
\begin{align*}
\left|\epsilon_{t}(\hat{\theta}_n)-\epsilon_{t}(\theta_0)\right|& \leq w_t\left\|\hat{\theta}_n-\theta_0\right\|
\end{align*}
and
\begin{align*}
\left|\frac{\partial}{\partial\theta_{m}}\epsilon_{t}(\hat{\theta}_n)-\frac{\partial}{\partial\theta_{m}}\epsilon_{t}(\theta_0)\right|& \leq q_t\left\|\hat{\theta}_n-\theta_0\right\|
\end{align*}
 with $w_t=\left\|{\partial\epsilon_t(\underline{\theta})}/{\partial\theta^{'}}\right\|$ and $q_t=\left\|{\partial^2\epsilon_t(\overline{\theta})}/{\partial\theta^{'}\partial\theta_m}\right\|$.
Using the fact that $$\mathbb{E}\left| w_{t-r_1}\frac{\partial}{\partial\theta_{m_1}}\epsilon_{t-r_1}(\hat{\theta}_n)\epsilon_{t-r_2}(\hat{\theta}_n)\frac{\partial}{\partial\theta_{m_2}}\epsilon_{t-r_2}(\hat{\theta}_n)\right| <\infty$$ and that $(\sqrt{n}( \hat{\theta}_n-\theta_0))_n$ is a tight sequence (which implies that $\|\hat{\theta}_n-\theta_0\|=\mathrm{O}_{\mathbb{P}}(1/\sqrt{n})$), we deduce that
\begin{align*}
U_{n,1}=\mathrm{O}_{\mathbb{P}}\left(\frac{1}{\sqrt{n}}\right).
\end{align*}
The same arguments are valid for $U_{n,2}$, $U_{n,3}$ and $U_{n,4}$. Consequently $
U_{n,1}+U_{n,2}+U_{n,3}+U_{n,4}=\mathrm{O}_{\mathbb{P}}(1/\sqrt{n})$ and \eqref{cu} yields
\begin{align*}
\left\|\hat{\Sigma}_{\underline{H}_{r,n}}-\hat{\Sigma}_{\underline{H}_r}\right\|^2&
=\mathrm{O}_{\mathbb{P}}\left(\frac{r^2}{n}\right).
\end{align*}
When $r=\mathrm{o}(n^{1/3})$ we finally obtain $\sqrt{r}\|\hat{\Sigma}_{\underline{H}_{r,n}}-\hat{\Sigma}_{\underline{H}_r}\|=\mathrm{o}_{\mathbb{P}}(1)$.

\end{proof}
\begin{lemme}\label{lemPhi} Under the assumptions of Theorem~\ref{convergence_Isp}, we have
$$\sqrt{r}\left\|\underline{\Phi}_r^{*}-\underline{\Phi}_r\right\|=\mathrm{o}_{\mathbb{P}}(1)\quad\text{ as }r\rightarrow\infty.$$
\end{lemme}
\begin{proof} Recall that by \eqref{AR_infty} and \eqref{ecriture_MATRICIELLE_ar_tronque} we have
$$H_t(\theta_0)=\underline{\Phi}_r\underline{H}_{r,t}+u_{r,t}=\underline{\Phi}_r^{*}\underline{H}_{r,t}+\sum_{k=r+1}^{\infty}\Phi_kH_{t-k}(\theta_0)+u_t:=\underline{\Phi}_r^{*}\underline{H}_{r,t}+u_{r,t}^{*}.$$
By the orthogonality conditions in \eqref{AR_infty} and \eqref{ecriture_MATRICIELLE_ar_tronque}, one has
\begin{align*}
\Sigma_{u_r^{*},\underline{H}_r}:=\mathbb{E}\left[ u_{r,t}^{*}\underline{H}_{r,t}^{'}\right] &=\mathbb{E}\left[ \left(H_t(\theta_0)-\underline{\Phi}_r^{*}\underline{H}_{r,t} \right) \underline{H}_{r,t}^{'}\right]
=\mathbb{E}\left[ \left(\underline{\Phi}_r\underline{H}_{r,t}+u_{r,t}-\underline{\Phi}_r^{*}\underline{H}_{r,t} \right) \underline{H}_{r,t}^{'}\right]\\
&=\left(\underline{\Phi}_r- \underline{\Phi}_r^{*}\right) \Sigma_{\underline{H}_r} ,
\end{align*}
and consequently
\begin{equation}\label{ecart_phi_etoile_phi}
\underline{\Phi}_r^{*}-\underline{\Phi}_r=-\Sigma_{u_r^{*},\underline{H}_r}\Sigma_{\underline{H}_r}^{-1}.
\end{equation}
Using Lemma~\ref{sur_existence_des_matrices} and Lemma~\ref{lemme_cov_des_H}, \eqref{ecart_phi_etoile_phi} implies that
\begin{align*}
\mathbb{P}\left( \sqrt{r}\left\|\underline{\Phi}_r^{*}-\underline{\Phi}_r\right\|\geq \beta\right)&\leq \frac{\sqrt{r}}{\beta} \left\|\Sigma_{u_r^{*},\underline{H}_r}\right\|\left\|\Sigma_{\underline{H}_r}^{-1}\right\| \\ %
&\leq\frac{K\sqrt{r}}{\beta}\left\|\mathbb{E}\left[ \left( \sum_{k\geq r+1}\Phi_kH_{t-k}(\theta_0)+u_t\right)\underline{H}_{r,t}^{'} \right] \right\|\\
&\leq \frac{K\sqrt{r}}{\beta} \sum_{k\geq r+1}\left\|\Phi_k\right\|\left\|\mathbb{E}\left[H_{t-k}(\theta_0)\underline{H}_{r,t}^{'} \right] \right\|\\
&\leq \frac{K\sqrt{r}}{\beta} \sum_{\ell\geq 1}\left\|\Phi_{\ell+r}\right\|\left\|\mathbb{E}\left[H_{t-\ell-r}(\theta_0)\left( H^{'}_{t-1}(\theta_0),\dots, H^{'}_{t-r}(\theta_0)\right)  \right] \right\|\\
&\leq \frac{K\sqrt{r}}{\beta} \sum_{\ell\geq 1}\left\|\Phi_{\ell+r}\right\|\left ( \sum_{j=1}^{p+q+1}\sum_{k=1}^{p+q+1}\sum_{r_1=1}^{r}\left|\mathbb{E}\left[ H_{t-r-\ell,j}(\theta_0) H_{t-r_1,k}(\theta_0) \right] \right|^2\right )^{1/2}\\
&\leq \frac{K\sqrt{r}}{\beta} \sum_{\ell\geq 1}\left\|\Phi_{\ell+r}\right\|\left (\sum_{j=1}^{p+q+1}\sum_{k=1}^{p+q+1}\sum_{r_1=1}^{r}\mathbb{E}\left[ H_{t-r-\ell,j}^2(\theta_0) \right]\mathbb{E}\left[H_{t-r_1,k}^2(\theta_0) \right]\right )^{1/2}\\
&\leq \frac{K(p+q+1)r}{\beta} \sum_{\ell\geq 1}\left\|\Phi_{\ell+r}\right\|.
\end{align*}
Under Assumptions of Theorem~\ref{convergence_Isp}, $r\sum_{\ell\geq 1}\left\|\Phi_{\ell+r}\right\|=\mathrm{o}(1)$ as $r\rightarrow\infty$. The proof of the lemma follows.
\end{proof}
\begin{lemme}\label{lemme_inter} Under the assumptions of Theorem $\ref{convergence_Isp}$, we  have
$$\sqrt{r}\left\|\hat{\Sigma}^{-1}_{\underline{\hat{H}}_r}-\Sigma^{-1}_{\underline{H}_r}\right\|=\mathrm{o}_{\mathbb{P}}(1)$$
as $n\rightarrow\infty$ when $r=\mathrm{o}(n^{(1-2(d_0-d_1))/5})$ and $r\rightarrow\infty$.
\end{lemme}

\begin{proof} We have
\begin{align*}
\left\|\hat{\Sigma}^{-1}_{\underline{\hat{H}}_r}-\Sigma^{-1}_{\underline{H}_r}\right\|&
\leq \left( \left\|\hat{\Sigma}^{-1}_{\underline{\hat{H}}_r}-\Sigma^{-1}_{\underline{H}_r}\right\|+\left\|\Sigma^{-1}_{\underline{H}_r}\right\|\right)\left\|\Sigma_{\underline{H}_r}- \hat{\Sigma}_{\underline{\hat{H}}_r}\right\|\left\| \Sigma_{\underline{H}_r}^{-1}\right\| ,
\end{align*}
and by induction we obtain
$$\left\|\hat{\Sigma}^{-1}_{\underline{\hat{H}}_r}-\Sigma^{-1}_{\underline{H}_r}\right\|\leq\left\| \Sigma_{\underline{H}_r}^{-1}\right\|\sum_{k=1}^{\infty}\left\|\Sigma_{\underline{H}_r}- \hat{\Sigma}_{\underline{\hat{H}}_r}\right\|^k\left\| \Sigma_{\underline{H}_r}^{-1}\right\|^k.$$
We have
\begin{align*}
\mathbb{P}\Big ( \sqrt{r}\big\|\hat{\Sigma}^{-1}_{\underline{\hat{H}}_r}-& \Sigma^{-1}_{\underline{H}_r}\big\|>\beta\Big ) \\
 & \leq \mathbb{P}\left(\sqrt{r}\left\| \Sigma_{\underline{H}_r}^{-1}\right\|\sum_{k=1}^{\infty}\left\|\Sigma_{\underline{H}_r}- \hat{\Sigma}_{\underline{\hat{H}}_r}\right\|^k\left\| \Sigma_{\underline{H}_r}^{-1}\right\|^k>\beta\right) \\
&\leq \mathbb{P}\left(\sqrt{r}\left\| \Sigma_{\underline{H}_r}^{-1}\right\|\sum_{k=1}^{\infty}\left\|\Sigma_{\underline{H}_r}- \hat{\Sigma}_{\underline{\hat{H}}_r}\right\|^k\left\| \Sigma_{\underline{H}_r}^{-1}\right\|^k>\beta \text{ and } \left\|\Sigma_{\underline{H}_r}- \hat{\Sigma}_{\underline{\hat{H}}_r}\right\|\left\| \Sigma_{\underline{H}_r}^{-1}\right\|<1\right)\\
&+\mathbb{P}\left( \sqrt{r}\left\| \Sigma_{\underline{H}_r}^{-1}\right\|\sum_{k=1}^{\infty}\left\|\Sigma_{\underline{H}_r}- \hat{\Sigma}_{\underline{\hat{H}}_r}\right\|^k\left\| \Sigma_{\underline{H}_r}^{-1}\right\|^k>\beta \text{ and } \left\|\Sigma_{\underline{H}_r}- \hat{\Sigma}_{\underline{\hat{H}}_r}\right\|\left\| \Sigma_{\underline{H}_r}^{-1}\right\|\geq 1\right)\\
&\leq  \mathbb{P}\left( \sqrt{r}\frac{\left\| \Sigma_{\underline{H}_r}^{-1}\right\|^2\left\|\Sigma_{\underline{H}_r}- \hat{\Sigma}_{\underline{\hat{H}}_r}\right\|}{1-\left\|\Sigma_{\underline{H}_r}- \hat{\Sigma}_{\underline{\hat{H}}_r}\right\|\left\| \Sigma_{\underline{H}_r}^{-1}\right\|}>\beta\right)+\mathbb{P}\left(\sqrt{r} \left\|\Sigma_{\underline{H}_r}- \hat{\Sigma}_{\underline{\hat{H}}_r}\right\|\left\| \Sigma_{\underline{H}_r}^{-1}\right\|\geq 1\right)\\
&\leq \mathbb{P}\left( \sqrt{r}\left\|\Sigma_{\underline{H}_r}- \hat{\Sigma}_{\underline{\hat{H}}_r}\right\|>\frac{\beta}{\left\| \Sigma_{\underline{H}_r}^{-1}\right\|^2+\beta r^{-1/2}\left\| \Sigma_{\underline{H}_r}^{-1}\right\|}\right) \\
& \hspace{4.5cm}+\mathbb{P}\left(\sqrt{r} \left\|\Sigma_{\underline{H}_r}- \hat{\Sigma}_{\underline{\hat{H}}_r}\right\|\geq \left\| \Sigma_{\underline{H}_r}^{-1}\right\|^{-1}\right) .
\end{align*}
Lemmas~\ref{sur_existence_des_matrices} and \ref{con_prob_Sigmachap_Hchap} imply the result.
\end{proof}

\begin{lemme} Under the assumptions of Theorem~\ref{convergence_Isp}, we have

$$\sqrt{r}\left\|\underline{\hat{\Phi}}_r-\underline{\Phi}_r\right\|=\mathrm{o}_{\mathbb{P}}(1)
\quad\text{ as }r\rightarrow\infty\text{ and }r=\mathrm{o}(n^{(1-2(d_0-d_1))/5}).$$
\end{lemme}

\begin{proof} Lemmas~\ref{sur_existence_des_matrices} and \ref{lemme_inter} yield
\begin{equation}\label{eq_inter}
\left\|\hat{\Sigma}^{-1}_{\underline{\hat{H}}_r}\right\|\leq \left\|\hat{\Sigma}^{-1}_{\underline{\hat{H}}_r}-\Sigma^{-1}_{\underline{H}_r}\right\|+\left\|\Sigma^{-1}_{\underline{H}_r}\right\|=\mathrm{O}_{\mathbb{P}}(1).
\end{equation}
By \eqref{ecriture_MATRICIELLE_ar_tronque}, we have
\begin{align*}
0&=\mathbb{E}\left[u_{r,t}\underline{H}_{r,t}^{'} \right] =\mathbb{E}\left[\left(H_t(\theta_0)-\underline{\Phi}_r\underline{H}_{r,t} \right) \underline{H}_{r,t}^{'} \right] =\Sigma_{H,\underline{H}_r}-\underline{\Phi}_r\Sigma_{\underline{H}_r},
\end{align*}
and so we have  $\underline{\Phi}_r=\Sigma_{H,\underline{H}_r}\Sigma_{\underline{H}_r}^{-1}$.
Lemmas~\ref{sur_existence_des_matrices},  \ref{con_prob_Sigmachap_Hchap},  \ref{lemme_inter} and \eqref{eq_inter} imply
\begin{align*}
\sqrt{r}\left\|\underline{\hat{\Phi}}_r-\underline{\Phi}_r\right\| & =\sqrt{r}\left\|\hat{\Sigma}_{\hat{H},\underline{\hat{H}}_r}\hat{\Sigma}_{\underline{\hat{H}}_r}^{-1}-\Sigma_{H,\underline{H}_r}\Sigma_{\underline{H}_r}^{-1}\right\| \\
& =\sqrt{r}\left\|\left( \hat{\Sigma}_{\hat{H},\underline{\hat{H}}_r}-\Sigma_{H,\underline{H}_r}\right) \hat{\Sigma}_{\underline{\hat{H}}_r}^{-1}+\Sigma_{H,\underline{H}_r}\left(\hat{\Sigma}_{\underline{\hat{H}}_r}^{-1}-\Sigma_{\underline{H}_r}^{-1}\right) \right\|\\
&=\mathrm{o}_{\mathbb{P}}(1) ,
\end{align*}
and the lemma is proved.
\end{proof}
\subsection*{Proof of Theorem \ref{convergence_Isp}}
Since by Lemma~\ref{con_prob_Sigmachap_Hchap} we have $\|\hat{\Sigma}_{\hat{H}}-\Sigma_{H}\|=\mathrm{o}_{\mathbb{P}}(r^{-1/2})=\mathrm{o}_{\mathbb{P}}(1)$ and $\|\hat{\Sigma}_{\hat{H},\underline{\hat{H}}_r}-\Sigma_{H,\underline{H}_r} \|=\mathrm{o}_{\mathbb{P}}(r^{-1/2})=\mathrm{o}_{\mathbb{P}}(1)$,  and by Lemma~\ref{lemPhi} $\| \underline{\hat{\Phi}}_r-\underline{\Phi}^{*}_r\|=\mathrm{o}_{\mathbb{P}}(r^{-1/2})=\mathrm{o}_{\mathbb{P}}(1)$,  Theorem~\ref{convergence_Isp} is then proved.

\subsection{Invertibility of the normalization matrix $P_{p+q+1,n}$}\label{sectinvP}
The following proofs are quite technical and are adaptations of the arguments used in \cite{BMS2018}.\\
To prove Proposition~\ref{inversibleP}, we need to introduce the following notation.\\
We denote $S_t$ the vector of $\mathbb{R}^{p+q+1}$ defined by
$$S_t=\sum_{j=1}^tU_j=\sum_{j=1}^t-2J^{-1}H_j=-2J^{-1}\sum_{j=1}^t\epsilon_j\frac{\partial}{\partial\theta}\epsilon_j(\theta_0)$$
and $S_t(i)$ its $i-$th component. We have 
\begin{align}\label{recurrence_Ft}
S_{t-1}(i)&=S_t(i)-U_t(i).
\end{align}
If the matrix $P_{p+q+1,n}$ is not invertible, there exists some real constants $d_1,\dots,d_{p+q+1}$ not all equal to zero, such that $\mathbf{d}^{'}P_{p+q+1,n}\mathbf{d}=0$, where $\mathbf{d}=(d_1,\dots,d_{p+q+1})^{'}$. Thus we may write that $\sum_{i=1}^{p+q+1}\sum_{j=1}^{p+q+1}d_jP_{p+q+1,n}(j,i)d_i=0$ or equivalently 
\begin{align*}
\frac{1}{n^2}\sum_{t=1}^n\sum_{i=1}^{p+q+1}\sum_{j=1}^{p+q+1}d_j\left(\sum_{k=1}^t( U_k(j)-\bar U_n(j))\right)
\left(\sum_{k=1}^t( U_k(i)-\bar U_n(i))\right)d_i&=0 .
\end{align*}
Then 
\begin{align*}
\sum_{t=1}^n\left(\sum_{i=1}^{p+q+1}d_i\left(\sum_{k=1}^t( U_k(i)-\bar U_n(i))\right)\right) ^2&=0,
\end{align*}
which implies that for all $t\ge 1$ 
\begin{align*}
\sum_{i=1}^{p+q+1}d_i\left(\sum_{k=1}^t( U_k(i)-\bar U_n(i))\right)=\sum_{i=1}^{p+q+1}d_i\left(S_t(i)-\frac{t}{n}S_n(i)\right)&=0 .
\end{align*}
So we have 
\begin{align}\label{plo}
\frac{1}{t}\sum_{i=1}^{p+q+1}d_iS_t(i)&=\sum_{i=1}^{p+q+1}d_i\left(\frac{1}{n}S_n(i)\right) . 
\end{align}
We apply the ergodic theorem and we use the orthogonality of $\epsilon_t$ and $(\partial/\partial\theta)\epsilon_t(\theta_0)$ in order to obtain that 
$$\sum_{i=1}^{p+q+1}d_i\left(\frac{1}{n}\sum_{k=1}^nU_k(i)\right)\xrightarrow[n\to\infty]{\text{a.s.}}
\sum_{i=1}^{p+q+1}d_i\mathbb{E}\left[U_k(i) \right]=-2\sum_{i,j=1}^{p+q+1}d_iJ^{-1}(i,j)\mathbb{E}\left[\epsilon_k\frac{\partial\epsilon_k}{\partial\theta_j} \right]=0  \ .$$
Reporting this convergence in \eqref{plo} implies that $\sum_{i=1}^{p+q+1}d_iS_t(i)=0$ a.s. for all $t\ge 1$. 
By \eqref{recurrence_Ft}, we deduce that 
\begin{align*}
\sum_{i=1}^{p+q+1}d_iU_t(i)=-2\sum_{i=1}^{p+q+1}d_i\sum_{j=1}^{p+q+1}J^{-1}(i,j)\left(\epsilon_t\frac{\partial\epsilon_t}{\partial\theta_j} \right)=0, \ \ \  \mathrm{ a.s.}
\end{align*}
Thanks to Assumption {\bf (A5)}, $(\epsilon_t)_{t\in\mathbb{Z}}$ has a positive density in some neighborhood of zero and then
$\epsilon_t\neq0$ almost-surely. So we would have $\mathbf {d}^{'}J^{-1}\frac{\partial\epsilon_t}{\partial\theta}=0$ a.s. 
Now we can follow the same arguments that we developed in the proof of  the invertibility of $J$ (see Proof of Lemma \ref{matrixJ} and more precisely \eqref{hypInvJ}) and this leads us to a contradiction. We deduce that the matrix $P_{p+q+1,n}$ is non singular. 
\subsection{Proof of Theorem \ref{self-norm-estimator}}\label{sectSN}
The arguments follows the one \cite{BMS2018} in a simpler context.\\
Recall that the Skorohod space $\mathbb{D}^{\ell}[ 0{,}1]$  is the set of $\mathbb{R}^{\ell}-$valued functions on $[ 0{,}1] $ which are right-continuous and has left limits everywhere. It is endowed with the Skorohod topology and the weak convergence on $\mathbb{D}^{\ell} [ 0{,}1]$ is mentioned by $\xrightarrow[]{\mathbb{D}^{\ell}}$. The integer part of $x$ will be denoted by $\lfloor x\rfloor$.\\
The goal at first is to show that there exists a lower triangular matrix $T$ with nonnegative diagonal entries such that
\begin{equation}\label{conv-Skorohod-1-pr}
\frac{1}{\sqrt{n}}\sum_{t=1}^{\lfloor nr\rfloor}U_t  \xrightarrow[n\to\infty]{\mathbb{D}^{p+q+1}} \ (TT^{'})^{1/2} B_{p+q+1}(r),
\end{equation}
where $(B_{p+q+1}(r))_{r\geq 0}$ is a $(p+q+1)-$dimensional standard Brownian motion.
Using \eqref{epsil-th}, $U_t$ can be rewritten as
$$U_t=\left( -2\left\lbrace \sum_{i=1}^{\infty}\overset{\textbf{.}}{\lambda}_{i,1}\left(\theta_0\right) \epsilon_t\epsilon_{t-i},\dots,\sum_{i=1}^{\infty}\overset{\textbf{.}}{\lambda}_{i,p+q+1}\left(\theta_0\right) \epsilon_t\epsilon_{t-i}\right\rbrace J^{-1 '} \right) ^{'}.$$
The non-correlation between $\epsilon_t$'s implies that the process $(U_t)_{t\in\mathbb{Z}}$ is centered. In order to apply the functional central limit theorem for strongly mixing process, we need to identify the asymptotic covariance matrix in the classical central limit theorem for the sequence $(U_t)_{t\in\mathbb{Z}}$. It is proved in Theorem \ref{n.asymptotique} that
$$\frac{1}{\sqrt{n}}\sum_{t=1}^nU_t\xrightarrow[n\to\infty]{\text{in law}} \mathcal{N}\left(0,\Omega=:2\pi f_U(0)\right),$$
where $f_U(0)$ is the spectral density of the stationary process $(U_t)_{t\in\mathbb{Z}}$ evaluated at frequency 0. The existence of the matrix $\Omega$ has already been discussed (see the proofs of lemmas \ref{matrixJ} and \ref{exit_I}).

Since the matrix $\Omega$ is symmetric positive definite, it can be factorized as $\Omega=TT^{'}$ where the $(p+q+1)\times(p+q+1)$ lower triangular matrix $T$ has real positive diagonal entries. Therefore, we have
$$\frac{1}{\sqrt{n}}\sum_{t=1}^n(TT^{'})^{-1/2}F_t\xrightarrow[n\to\infty]{\text{in law}}  \ \mathcal{N}\left(0,I_{p+q+1}\right),$$
where $I_{p+q+1}$ is the identity matrix of order $p+q+1$.

As in the proof of the asymptotic normality of $(\sqrt{n}(\hat{\theta}_n-\theta_0))_{n\geq 1}$, the distribution  of $n^{-1/2}\sum_{t=1}^nU_t$ when $n$ tends to infinity is obtained by introducing the random vector $U_t^k$ defined for any positive integer $k$ by
$$U_t^k=\left( -2\left\lbrace \sum_{i=1}^{k}\overset{\textbf{.}}{\lambda}_{i,1}\left(\theta_0\right) \epsilon_t\epsilon_{t-i},\dots,\sum_{i=1}^{k}\overset{\textbf{.}}{\lambda}_{i,p+q+1}\left(\theta_0\right) \epsilon_t\epsilon_{t-i}\right\rbrace J^{-1 '} \right) ^{'}.$$
Since $U_t^k$ depends on a finite number of values of the noise process $(\epsilon_t)_{t\in\mathbb{Z}}$, it also satisfies a mixing property (see Theorem 14.1 in \cite{davidson1994}, p. 210).
The central limit theorem for strongly mixing process  of \cite{herr} shows  that its asymptotic distribution is normal with zero mean and variance matrix $\Omega_k$ that converges when $k$ tends to infinity to $\Omega$ (see the proof of Lemma \ref{normalite_score}): 
$$\frac{1}{\sqrt{n}}\sum_{t=1}^nU_t^k \xrightarrow[n\to\infty]{\text{in law}}  \ \mathcal{N}\left(0,\Omega_k\right).$$
The above arguments also apply to matrix $\Omega_k$ with some matrix $T_k$ which is defined analogously as $T$. Consequently, we obtain
$$\frac{1}{\sqrt{n}}\sum_{t=1}^n(T_kT_k^{'})^{-1/2}U_t^k \xrightarrow[n\to\infty]{\text{in law}} \mathcal{N}(0,I_{p+q+1}).$$
Now we are able to apply the functional central limit theorem for strongly mixing process  of \cite{herr} and we obtain that 
$$\frac{1}{\sqrt{n}}\sum_{t=1}^{\lfloor nr\rfloor}(T_kT_k^{'})^{-1/2}U_t^k  \xrightarrow[n\to\infty]{\mathbb{D}^{p+q+1}} B_{p+q+1}(r).$$
Since
$$(TT^{'})^{-1/2}U_t^k=\left((TT^{'})^{-1/2}-(T_kT_k^{'})^{-1/2}\right)U_t^k+(T_kT_k^{'})^{-1/2}U_t^k,$$
we may use the same approach as in the proof of Lemma \ref{normalite_score} in order to prove that $n^{-1/2}\sum_{t=1}^{n}((TT^{'})^{-1/2}-(T_kT_k^{'})^{-1/2})U_t^k$ converges in distribution to 0.
Consequently we obtain that 
$$\frac{1}{\sqrt{n}}\sum_{t=1}^{\lfloor nr\rfloor}(TT^{'})^{-1/2}U_t^k  \xrightarrow[n\to\infty]{\mathbb{D}^{p+q+1}}B_{p+q+1}(r).$$

In order to conclude that \eqref{conv-Skorohod-1-pr} is true, it remains to observe that uniformly with respect to $n$ it holds that 
\begin{equation}\label{conv-skorokh-2-pr}
\tilde{Y}_n^k(r):=\frac{1}{\sqrt{n}}\sum_{t=1}^{\lfloor nr\rfloor}(TT^{'})^{-1/2}\tilde{Z}_t^k  \xrightarrow[n\to\infty]{\mathbb{D}^{p+q+1}} \ 0,
\end{equation}
where
$$\tilde{Z}_t^k=\left( -2\left\lbrace \sum_{i=k+1}^{\infty}\overset{\textbf{.}}{\lambda}_{i,1}\left(\theta_0\right) \epsilon_t\epsilon_{t-i},\dots,\sum_{i=k+1}^{\infty}\overset{\textbf{.}}{\lambda}_{i,p+q+1}\left(\theta_0\right) \epsilon_t\epsilon_{t-i}\right\rbrace J^{-1 '} \right) ^{'}.$$
By \eqref{koko}, one has
$$\sup_{n}\mathrm{Var}\left(\frac{1}{\sqrt{n}}\sum_{t=1}^{n}\tilde{Z}_t^k\right)\xrightarrow[n\to\infty]{} 0$$
and since $\lfloor nr\rfloor\leq n$,
$$\sup_{0\leq r\leq 1}\sup_n\left\lbrace \left\|\tilde{Y}_n^k(r)\right\|\right\rbrace  \xrightarrow[n\to\infty]{} \ 0.$$
Thus \eqref{conv-skorokh-2-pr} is true and the proof of \eqref{conv-Skorohod-1-pr} is achieved.

By \eqref{conv-Skorohod-1-pr} we deduce that 
\begin{align}\label{convsko}
& \frac{1}{\sqrt{n}}\left(\sum_{j=1}^{\lfloor nr \rfloor} (U_j-\bar U_n)\right)  \xrightarrow[n\to\infty]{\mathbb{D}^{p+q+1}}(TT^{'})^{1/2}\left(B_{p+q+1}(r)-rB_{p+q+1}(1)\right) . 
\end{align}
One remarks that the continuous mapping theorem on the Skorohod space yields 
\begin{align*}
P_{p+q+1,n}  \xrightarrow[n\to\infty]{\text{in law}}  
&(TT^{'})^{1/2}\left[ \int_0^1\left\lbrace  B_{p+q+1}(r)-rB_{p+q+1}(1)\right\rbrace\left\lbrace B_{p+q+1}(r)-rB_{p+q+1}(1)\right\rbrace^{'}\mathrm{dr}\right] (TT^{'})^{1/2}\\
=&(TT^{'})^{1/2}V_{p+q+1}(TT^{'})^{1/2}.
\end{align*}
Using \eqref{conv-Skorohod-1-pr}, \eqref{convsko} and the continuous mapping theorem on the Skorohod space, one finally obtains
\begin{align*}
& n \big (\hat{\theta}_{n}-\theta_0\big )^{'}P_{p+q+1,n}^{-1}\big (\hat{\theta}_{n}-\theta_0\big ) \\ 
& \hspace{2.5cm} \xrightarrow[n\to\infty]{\mathbb{D}^{p+q+1}} \left\lbrace (TT^{'})^{1/2}B_{p+q+1}(1)\right\rbrace ^{'}\left\lbrace (TT^{'})^{1/2}V_{p+q+1}(TT^{'})^{1/2}\right\rbrace^{-1} \left\lbrace (TT^{'})^{1/2}B_{p+q+1}(1)\right\rbrace\\
& \hspace{4cm} =B_{p+q+1}^{'}(1)V_{p+q+1}^{-1}B_{p+q+1}(1):=\mathcal{U}_{p+q+1}.
\end{align*}
The proof of Theorem \ref{self-norm-estimator} is then complete.

\subsection{Proof of Theorem \ref{self-norm-estimator2}}\label{sectSN2}
In view of \eqref{matP} and \eqref{matPchap},  we write $\hat P_{p+q+1,n}=P_{p+q+1,n}+ Q_{p+q+1,n}$ where
\begin{align*}
Q_{p+q+1,n}  &=
\big ( J(\theta_0)^{-1}-{\hat J}_n^{-1} \big  )\frac{1}{n^2}\sum_{t=1}^n\left(\sum_{j=1}^t( H_j-\mbox{$\frac{1}n \sum_{k=1}^n H_k$})\right)\left(\sum_{t=1}^n( H_j-\mbox{$\frac{1}n \sum_{k=1}^n H_k$})\right)^{'}
\\ 
&\hspace{2cm}+ {\hat J}_n^{-1} \frac{1}{n^2} \sum_{t=1}^n \Bigg \{ \left(\sum_{j=1}^t( H_j-\mbox{$\frac{1}n \sum_{k=1}^n H_k$})\right)\left(\sum_{t=1}^n( H_j-\mbox{$\frac{1}n \sum_{k=1}^n H_k$})\right)^{'} \\ 
& \hspace{5cm} - \left(\sum_{j=1}^t( \hat H_j-\mbox{$\frac{1}n \sum_{k=1}^n \hat H_k$})\right)\left(\sum_{t=1}^n( \hat H_j-\mbox{$\frac{1}n \sum_{k=1}^n \hat H_k$})\right)^{'} \Bigg \} . 
\end{align*}
Using the same approach as in Lemma \ref{Conv_almost_sure_J_n_vers_J}, $\hat{J}_n$ converges in probability to $J$. Thus we deduce that the first term in the right hand side of the above equation tends to zero in probability. 

The second term is a sum composed of the following terms 
$$q_{s,t}^{i,j,k,l} = \epsilon_s(\theta_0)\epsilon_{t}(\theta_0)\frac{\partial {\epsilon}_{s}(\theta_0)}{\partial \theta_i}
\frac{\partial \epsilon_{t}(\theta_0)}{\partial \theta_{j}}  - \tilde{\epsilon}_s(\hat\theta_n)\tilde{\epsilon}_{t}(\hat\theta_n)\frac{\partial \tilde{\epsilon}_{s}(\hat\theta_n)}{\partial \theta_k}
\frac{\partial \tilde{\epsilon}_{t}(\hat\theta_n)}{\partial \theta_l}   \ .$$
Using similar arguments done before (see for example the use of Taylor's expansion in Subsection \ref{sectI}, we have $q_{s,t}^{i,j,k,l} = \mathrm{o}_{\mathbb P}(1)$ as $n$ goes to infinity and thus  $Q_{p+q+1,n}= \mathrm{o}_{\mathbb P}(1)$.
So one may find a matrix $Q^\ast_{p+q+1,n}$ that tends to the null matrix in probability and such  that
\begin{align*}
n\, \left(\hat{\theta}_{n}-\theta_0\right)^{'}\hat P_{p+q+1,n}^{-1}\left(\hat{\theta}_{n}-\theta_0\right) &  = n\,  \left(\hat{\theta}_{n}-\theta_0\right)^{'}
 \left(P_{p+q+1,n}+ Q_{p+q+1,n}\right)^{-1}\left(\hat{\theta}_{n}-\theta_0\right) \\&= n\, \left(\hat{\theta}_{n}-\theta_0\right)^{'}P_{p+q+1,n}^{-1}\left(\hat{\theta}_{n}-\theta_0\right)\\&\qquad + n\, \left(\hat{\theta}_{n}-\theta_0\right)^{'}Q_{p+q+1,n}^\ast \left(\hat{\theta}_{n}-\theta_0\right).
\end{align*}
Thanks to the arguments developed in the proof of Theorem \ref{self-norm-estimator}, $n (\hat{\theta}_{n}-\theta_0)^{'} P_{p+q+1,n}^{-1}(\hat{\theta}_{n}-\theta_0)$ converges in distribution. So $n (\hat{\theta}_{n}-\theta_0)^{'}Q_{p+q+1,n}^\ast (\hat{\theta}_{n}-\theta_0)$ tends to zero in distribution, hence in probability.  Then $n (\hat{\theta}_{n}-\theta_0)^{'}\hat P_{p+q+1,n}^{-1}(\hat{\theta}_{n}-\theta_0)$ and  $n (\hat{\theta}_{n}-\theta_0)^{'} P_{p+q+1,n}^{-1}(\hat{\theta}_{n}-\theta_0)$ have the same limit in distribution and the result is proved.

\bibliographystyle{smfalpha2}
\bibliography{biblio-youss}
\newpage
\section{Figures and tables}\label{fig}
\begin{figure}[H]
\includegraphics[width=12cm,height=10cm]{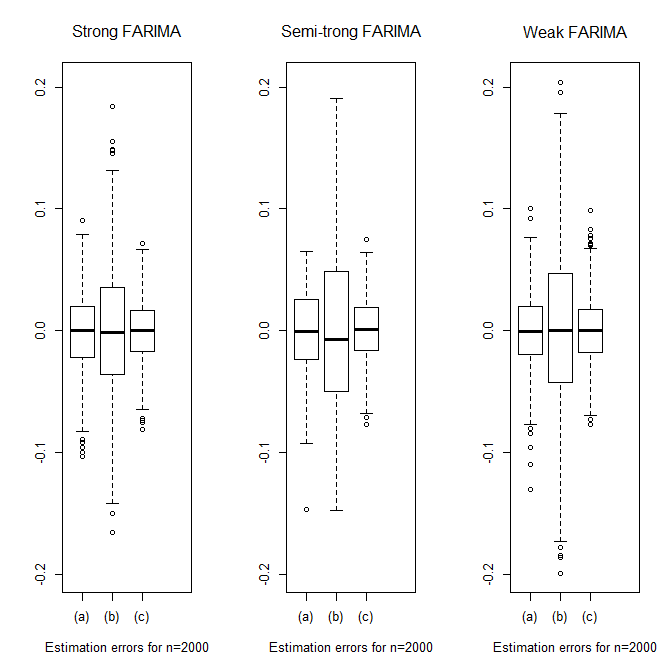}
\caption{LSE of $N=1,000$ independent simulations of the FARIMA$(1,d,1)$ model \eqref{process-sim} with size $n=2,000$ and unknown parameter $\theta_0=(a,b,d)=(-0.7,-0.2,0.4)$, when the noise is strong (left panel), when the noise is semi-strong \eqref{noise-sim} (middle panel) and when the noise is weak of the form \eqref{PT} (right panel). Points (a)-(c), in the box-plots, display the distribution of the estimation error $\hat{\theta}_n(i)-\theta_0(i)$ for $i=1,2,3$.  }
\label{fig1}
\end{figure}
\begin{figure}[H]
\includegraphics[width=12cm,height=10cm]{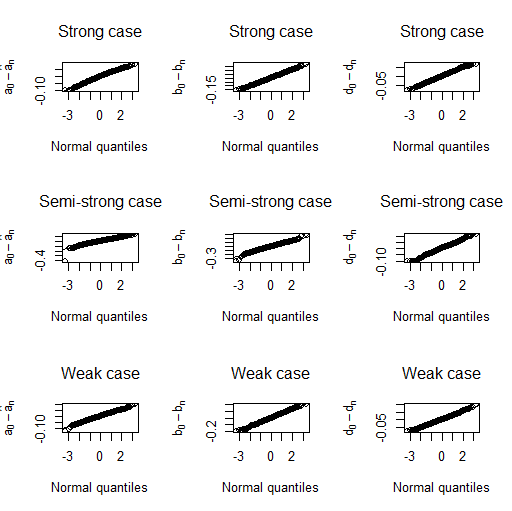}
\caption{LSE of $N=1,000$ independent simulations of the FARIMA$(1,d,1)$ model \eqref{process-sim} with size $n=2,000$ and unknown parameter $\theta_0=(a,b,d)=(-0.7,-0.2,0.4)$.
The top panels present respectively, from left to right, the Q-Q plot of the estimates $\hat{a}_n$, $\hat{b}_n$ and $\hat{d}_n$ of $a$, $b$ and $d$ in the strong case.
Similarly the middle and the bottom panels present respectively, from left to right, the Q-Q plot of the estimates $\hat{a}_n$, $\hat{b}_n$ and $\hat{d}_n$ of $a$, $b$ and $d$ in the semi-strong and weak cases. }
\label{fig2}
\end{figure}
\begin{figure}[H]
\includegraphics[width=12cm,height=10cm]{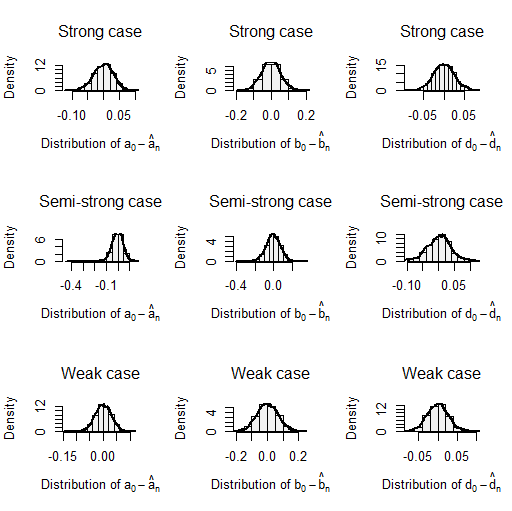}
\caption{LSE of $N=1,000$ independent simulations of the FARIMA$(1,d,1)$ model \eqref{process-sim} with size $n=2,000$ and unknown parameter $\theta_0=(a,b,d)=(-0.7,-0.2,0.4)$. The top panels present respectively, from left to right, the distribution of the estimates $\hat{a}_n$, $\hat{b}_n$ and $\hat{d}_n$ of $a$, $b$ and $d$ in the strong case. Similarly the middle and the bottom panels present respectively, from left to right, the distribution of the estimates $\hat{a}_n$, $\hat{b}_n$ and $\hat{d}_n$ of $a$, $b$ and $d$ in the semi-strong and weak cases. The kernel density estimate is displayed in full line, and the centered Gaussian density with the same variance is plotted in dotted line. }
\label{fig3}
\end{figure}
\begin{figure}[H]
\includegraphics[width=12cm,height=10cm]{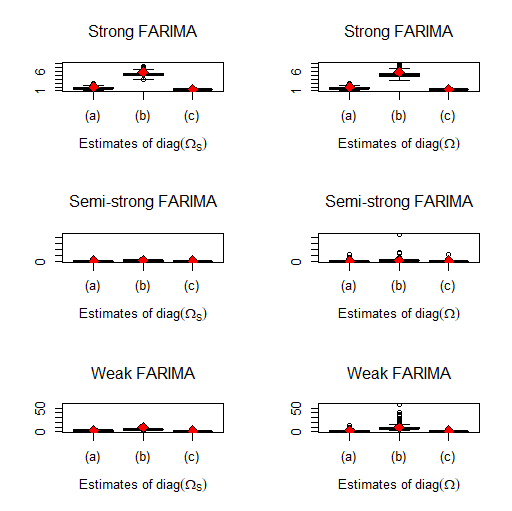}
\caption{Comparison of standard and modified estimates of the asymptotic variance $\Omega$ of the LSE, on the simulated models presented in Figure~\ref{fig1}.
The diamond symbols represent the mean, over $N=1,000$ replications, of the standardized errors $n(\hat{a}_n+0.7)^2$ for (a) (1.90 in the strong case and 4.32
(resp. 1.80) in the semi-strong case (resp.  in the weak case)), $n(\hat{b}_n+0.2)^2$ for (b) (5.81 in the strong case and  11.33
(resp. 8.88) in the semi-strong case (resp.  in the weak case)) and $n(\hat{d}_n-0.4)^2$ for (c) (1.28 in the strong case and 2.65 (resp. 1.40) in the semi-strong case (resp.  in the weak case)).}
\label{fig4}
\end{figure}
\begin{figure}[H]
\includegraphics[width=12cm,height=10cm]{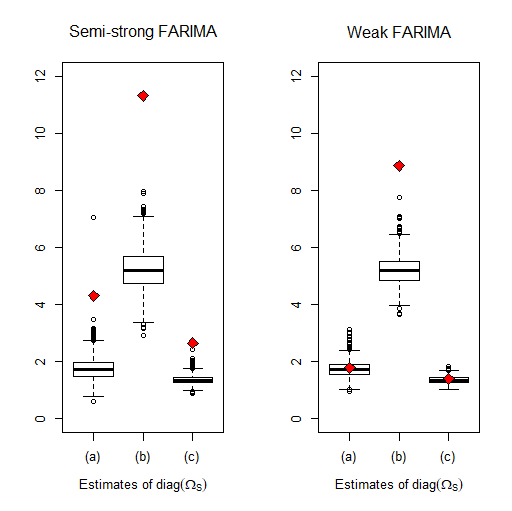}
\caption{A zoom of the left-middle and left-bottom panels of Figure~\ref{fig4}}
\label{fig5}
\end{figure}
\begin{figure}[H]
\includegraphics[width=12cm,height=10cm]{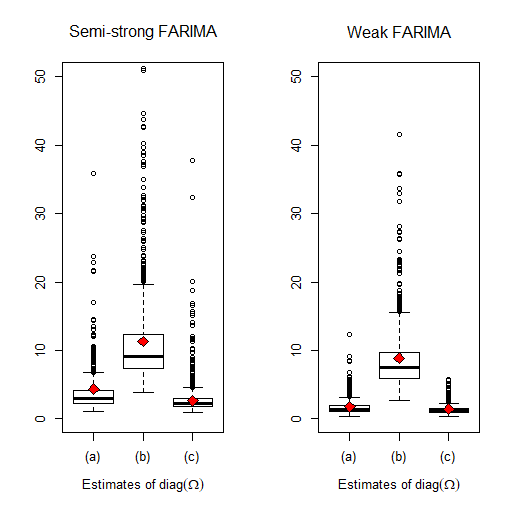}
\caption{A zoom of the right-middle and right-bottom panels of Figure~\ref{fig4}}
\label{fig6}
\end{figure}
\begin{table}[H]
 \caption{\small{Empirical size of standard and modified confidence interval: relative frequencies (in \%)
 of rejection. Modified SN stands for the self-normalized approach. In Modified we use the sandwich estimator
 of the asymptotic variance $\Omega$ of the LSE  while in Standard we use  $\hat\Omega_S$.
The number of replications is $N=1000$. }}
\begin{center}
{\scriptsize
\begin{tabular}{lll rrr rrr rrr}
\hline\hline
Model& Length $n$ & Level & \multicolumn{3}{c}{Standard}& \multicolumn{3}{c}{Modified}& \multicolumn{3}{c}{Modified SN}\\
&&&$\hat{a}_n$&$\hat{b}_n$&$\hat{d}_n$
&$\quad$ $\hat{a}_n$&$\hat{b}_n$&$\hat{d}_n$&$\quad$ $\hat{a}_n$&$\hat{b}_n$&$\hat{d}_n$\vspace*{0.1cm}\\
&& $\alpha=1\%$&{\bf 2.8}  &{\bf 2.7} &{\bf 2.1} &$\quad${\bf 3.7} &{\bf 3.1}& {\bf 2.5}&$\quad${\bf 2.5} &{\bf 2.5}&{\bf 2.0} \\
Strong FARIMA&$n=200$& $\alpha=5\%$&{\bf  7.1 }&{\bf 7.3} &5.2 &$\quad${\bf 8.2 }&{\bf 8.0}&5.4&$\quad${\bf 8.4} &{\bf 6.9}& 5.6\\
 && $\alpha=10\%$& 11.8  & 11.2 &8.3 &$\quad${\bf 12.8} &{\bf 12.4}&9.5&$\quad$ {\bf 14.5}&11.5&10.6
 \\
\cline{2-12}
&& $\alpha=1\%$&1.1  &1.6 &  0.7&$\quad$1.3 & 1.6&1.0&$\quad$1.6 &1.0& 0.8 \\
Strong FARIMA&$n=2,000$& $\alpha=5\%$& 5.8 &{\bf 6.9 }&5.1 &$\quad$6.1 &{\bf 6.8}&5.3&$\quad$5.6 &6.4& 3.8\\
 && $\alpha=10\%$&10.9  &{\bf 13.0} & 9.5  & $\quad$ 11.4  &{\bf 12.8 }& 9.5&$\quad$ 10.3&{\bf 12.4}& 9.0
 \\
\cline{2-12}
&& $\alpha=1\%$&1.3  &1.2 &  0.7&$\quad$1.2 & 1.2&0.8&$\quad$0.8&1.3& 1.2 \\
Strong FARIMA&$n=5,000$& $\alpha=5\%$& 5.3 &4.9 &5.2&$\quad$5.7 &5.1&5.4&$\quad$5.1 &4.9&5.3\\
 && $\alpha=10\%$&10.6  &10.3 & 11.4 & $\quad$ 10.7  &10.2 &11.7&$\quad$ 11.8&10.8& 11.4
 \\
\hline
&& $\alpha=1\%$&{\bf 5.3 } &{\bf 3.7} &{\bf 2.3 }&$\quad${\bf 4.8} &{\bf 3.4}& {\bf 3.2}&$\quad${\bf 4.0} &{\bf 2.8}&1.4 \\
Semi-strong FARIMA&$n=200$& $\alpha=5\%$&{\bf  11.2 }&{\bf 9.5 }&6.1&$\quad${\bf 10.7 }&{\bf 9.1}&5.9&$\quad${\bf 11.1 }&{\bf 8.5}& 5.7\\
 && $\alpha=10\%$&{\bf 16.8}  & {\bf 14.5} &8.6 &$\quad${\bf 16.7}  &{\bf 14.7}&10.0&$\quad$ {\bf 17.1}&{\bf 13.9}&10.9
 \\
\cline{2-12}
&& $\alpha=1\%$&{\bf 6.8 } &{\bf 7.7} & {\bf  7.8}&$\quad$ 1.7 & 0.9&1.4&$\quad${\bf 2.3} &{\bf 2.0}& 0.8\\
Semi-strong FARIMA&$n=2,000$& $\alpha=5\%$& {\bf 19.5}
 &{\bf 17.7} &{\bf 14.9 }&$\quad${\bf 6.5} &5.8&5.5&$\quad${\bf 8.1} &{\bf 6.8}&{\bf  6.5}\\
 && $\alpha=10\%$&{\bf 26.5}  &{\bf 26.7} &{\bf  21.5}  & $\quad$ {\bf 13.5} &11.0 & 9.9&$\quad$ {\bf 14.6}&{\bf 12.8}& {\bf 12.5}
  \\
\cline{2-12}
&& $\alpha=1\%$&{\bf 11.2 } &{\bf 9.8} &{\bf 9.4}&$\quad$ 1.6 & 1.5&1.1&$\quad$1.6&1.3& 1.2\\
Semi-strong FARIMA&$n=5,000$& $\alpha=5\%$& {\bf 20.8}
 &{\bf 20.2} &{\bf 20.9 }&$\quad$6.4 &5.7&5.3&$\quad$5.7 &6.2&{\bf 7.2}\\
 && $\alpha=10\%$&{\bf 28.2}  &{\bf 28.4} &{\bf  28.4}  & $\quad$ {\bf 12.2} &9.8 & 11.3&$\quad$ {\bf 12.0}&{\bf 13.0}& {\bf 13.9}
\\
\hline
&& $\alpha=1\%$&{\bf 2.6}  &{\bf 4.4} &1.2 &$\quad${\bf 6.2 }&{\bf 6.9}&{\bf 4.3}&$\quad${\bf 4.1} &{\bf 4.2}&{\bf 2.6 }\\
Weak FARIMA&$n=200$& $\alpha=5\%$&{\bf  6.6} &{\bf 11.3 }&4.3 &$\quad${\bf 13.8} &{\bf 14.6}&{\bf 10.2}&$\quad${\bf 12.0 }&{\bf 10.5}&{\bf  8.9}\\
 && $\alpha=10\%$& 10.9 & {\bf 18.8 }&{\bf 7.1} &$\quad${\bf 20.3}  &{\bf 21.9}&{\bf 16.3}&$\quad${\bf  17.7}&17.4&{\bf 15.4}
 \\
\cline{2-12}
&& $\alpha=1\%$&1.1  &{\bf 5.3 }&  1.3&$\quad$1.5 & 1.2&1.6&$\quad$1.2 &1.1& 0.9 \\
Weak FARIMA&$n=2,000$& $\alpha=5\%$& 5.4 &{\bf 13.4 }&5.8 &$\quad${\bf 7.0 }&{\bf 6.8}&5.5&$\quad$5.7 &{\bf 6.5}& 6.4\\
 && $\alpha=10\%$&11.4 &{\bf 21.2} & 9.6  & $\quad$ {\bf 12.8}  &{\bf 12.0} & 11.2&$\quad$ 11.3&{\bf 11.9}& {\bf 12.2}
\\
\cline{2-12}
&& $\alpha=1\%$&1.3  &{\bf 4.6} & 1.7&$\quad$1.2 & 1.3&1.2&$\quad$1.3 &1.4& 0.9 \\
Weak FARIMA&$n=5,000$& $\alpha=5\%$& 6.3 &{\bf 14.4} &6.0 &$\quad${\bf 6.7 }&6.3&5.9&$\quad$6.2 &6.2& 5.0\\
 && $\alpha=10\%$&11.5 &{\bf 22.3} & 11.6  & $\quad$ {\bf 12.1}  &{\bf 12.3} & 10.8&$\quad$ 10.6&10.9& 10.0
 \\
\hline\hline
\\\end{tabular}
}
\end{center}
\label{tab1}
\end{table}

\newpage
\begin{figure}[H]
\includegraphics[width=12cm,height=10cm]{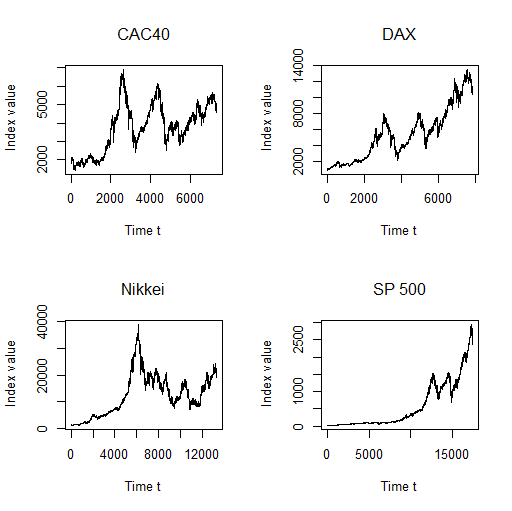}
\caption{Closing prices of the four stock market indices from  the starting date  of each index  to February 14, 2019. }
\label{graph1}
\end{figure}
\begin{figure}[H]
\includegraphics[width=12cm,height=10cm]{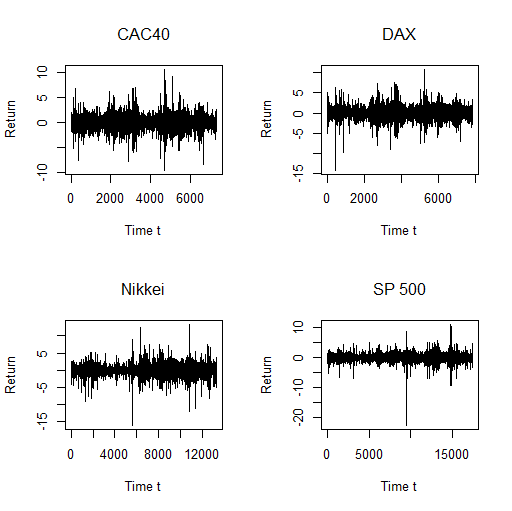}
\caption{Returns of the four stock market indices from  the starting date  of each index  to February 14, 2019. }
\label{graph}
\end{figure}

\begin{figure}[H]
\includegraphics[width=12cm,height=10cm]{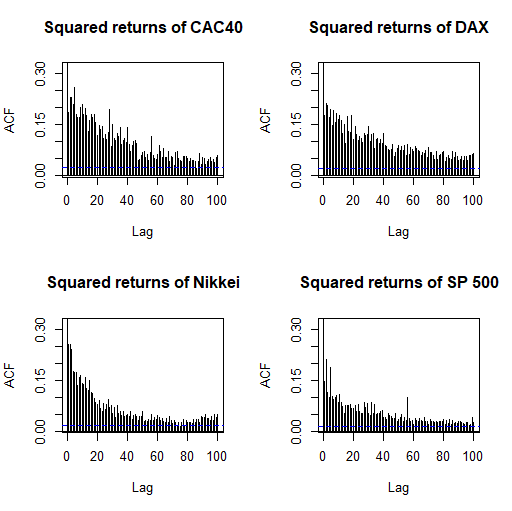}
\caption{Sample autocorrelations of squared returns of the four stock market indices. }
\label{acf}
\end{figure}

\begin{table}[H]
 \caption{\small{Fitting a FARIMA$(1,d,1)$ model to the squares of the 4 daily returns considered.
The corresponding $p-$values are reported in parentheses. The last column presents the estimated residual variance. }}
\begin{center}
\begin{tabular}{ll rrr r}
\hline\hline
Series& Length $n$ & \multicolumn{3}{c}{$\hat\theta_n$}&$\mbox{Var}(\epsilon_t)$\\
&&$\hat{a}_n$&$\hat{b}_n$&$\hat{d}_n$
&$\quad$ $\hat{\sigma}_{\epsilon}^2$\vspace*{0.1cm}\\
CAC&$n=7,341$& 0.1199 (0.1524) &0.5296 (0.0000) &0.4506 (0.0000)&$\quad$ $19.6244\times 10^{-8}$ \\
\cline{2-6}
DAX&$n=7,860$& 0.1598 (0.1819) & 0.4926 (0.0000) &0.3894 (0.0000) &$\quad$ $25.9383\times 10^{-8}$  \\
\cline{2-6}
Nikkei&$n=13,318$& -0.0217 (0.9528) &0.1579 (0.6050) &0.3217 (0.0000)&$\quad$ $25.6844\times 10^{-8}$  \\
\cline{2-6}
S\&P 500&$n=17,390$&  -0.3371 (0.0023)
 &-0.1795 (0.0227) &0.2338 (0.0000) &$\quad$ $22.9076\times 10^{-8}$  \\
\hline\hline
\\\end{tabular}
\end{center}
\label{tabreal}
\end{table}

\begin{table}[H]
 \caption{\small{Modified confidence interval at the asymptotic level $\alpha=5\%$ for the parameters estimated in Table~\ref{tabreal}.
 Modified SN stands for the self-normalized approach while Modified corresponds to the confidence interval obtained by using the sandwich estimator
 of the asymptotic variance $\Omega$ of the LSE. }}
\begin{center}
{\scriptsize
\begin{tabular}{l rrr rrr}
\hline\hline
Series& \multicolumn{3}{c}{Modified}& \multicolumn{3}{c}{Modified SN}\\
&$\hat{a}_n$&$\hat{b}_n$&$\hat{d}_n$&$\quad$ $\hat{a}_n$&$\hat{b}_n$&$\hat{d}_n$\vspace*{0.1cm}\\
CAC&$[-0.044,0.284]$ &$[0.319,0.740]$ &$[0.251,0.651]$ &$\quad$$[0.049,0.193]$ &$[0.432,0.627]$ &$[0.358,0.543]$ \\
\cline{2-7}
DAX& $[-0.075,0.394]$  &$[0.272,0.714]$  &$[0.263,0.516]$  &$\quad$$[-0.219,0.537]$  &$[0.422,0.563]$ &$[0.158,0.621]$ \\
\cline{2-7}
Nikkei&$[-0.741,0.698]$  &$[-0.440,0.756]$  &$[0.206,0.437]$ &$\quad$$[-0.823,0.779]$ &$[-0.407,0.717]$ &$[0.156,0.488]$ \\
\cline{2-7}
S\&P 500& $[-0.554,-0.121]$  &$[-0.334,-0.025]$  &$[0.162,0.306]$  &$\quad$$[-0.430,-0.244]$  &$[-0.297,-0.062]$ &$[0.101,0.366]$ \\
\hline\hline
\\\end{tabular}
}
\end{center}
\label{tab1real}
\end{table}


\end{document}